\documentclass[a4paper,11pt,onecolumn,accepted=2023-10-17]{quantumarticle}

\pdfoutput=1

\usepackage{multicol}
\usepackage{blindtext}
\usepackage{float}
\usepackage[T1]{fontenc} 
\usepackage[utf8]{inputenc}
\usepackage{cite}
\usepackage{mathtools} 
\usepackage{amsfonts}
\usepackage{amsmath} 
\usepackage{dsfont}
\usepackage{physics}
\usepackage[ruled]{algorithm2e}
\usepackage{enumerate}
\usepackage{enumitem}
\usepackage[font=small,labelfont=bf]{caption}
\usepackage{tikz}
\usetikzlibrary{hobby} 
\usepackage{subcaption}
\usepackage{graphicx}
\usepackage{placeins}  
\usepackage[normalem]{ulem} 
\usepackage{array} 
\newcolumntype{C}[1]{>{\centering\let\newline\\\arraybackslash\hspace{0pt}}m{#1}}
\usepackage{color}
\newcolumntype{L}[1]{>{\raggedright\let\newline\\\arraybackslash\hspace{0pt}}m{#1}}
\usepackage[margin=3.85cm]{geometry}
\usepackage{changepage} 
\usepackage{multirow}
\allowdisplaybreaks
\usepackage[colorlinks=true,linkcolor=blue,citecolor=citegreen]{hyperref}
\usepackage[affil-it]{authblk}
\usepackage{amsthm}
\usepackage[nameinlink,capitalise]{cleveref}
\newtheorem{theorem}{Theorem}
\newtheorem{lemma}[theorem]{Lemma}
\newtheorem{corollary}[theorem]{Corollary}
\newtheorem{definition}[theorem]{Definition}

\crefname{section}{Section}{Sections}
\crefname{subsection}{subsection}{subsections}
\crefname{theorem}{Theorem}{Theorems}
\crefname{corollary}{Corollary}{Corollaries}
\crefname{lemma}{Lemma}{Lemmas}
\crefname{appendix}{Appendix}{Appendices}
\crefname{definition}{Definition}{Definitions}
\crefname{equation}{eq.}{eqs.}
\crefname{algorithm}{Algorithm}{Algorithms}
\definecolor{myred}{RGB}{221,0,0}
\definecolor{tk}{RGB}{246,76,246}
\definecolor{citegreen}{RGB}{0,165,0}
\definecolor{mygreen}{RGB}{0,150,0}
\definecolor{mygrey}{RGB}{100,100,100}
\renewcommand{\ket}[1]{| #1 \rangle}

\newcommand{\set}[1]{\left\{ #1 \right\}}
\newcommand{\1}{\mathds{1}}
\newcommand{\mc}[1]{\mathcal{#1}}
\newcommand{\PP}{\mathds{P}}
\newcommand{\NDsquare}{\mathrm{NDS}}
\newcommand{\setn}[1]{ \{ #1 \} }
\newcommand{\Fnorm}[1]{\norm{#1}_\mathrm{F}}
\newcommand{\Fnormn}[1]{\lVert #1 \rVert_\mathrm{F}}
\newcommand{\Flip}{\mathds{F}}
\newcommand{\wtM}{\widetilde{M}}

\newcommand{\e}{\mathrm{e}}
\renewcommand{\i}{\mathrm{i}}
\newcommand{\dt}{\mathrm{d}t}
\newcommand{\avg}{\mathrm{avg}}
\newcommand{\wt}[1]{\widetilde{#1}}
\newcommand{\randomSeeds}{\mathrm{random}\_\mathrm{seeds}}

\newcommand{\degCount}{\mathrm{degeneracy}\_\mathrm{counter}}
\newcommand{\degSet}{\mathrm{degenerate}\_\mathrm{set}}
\newcommand{\conjDegSet}{\mathrm{conj}\_\mathrm{degenerate}\_\mathrm{set}}

\newcommand{\sumValues}{\mathrm{values}\_\mathrm{to}\_\mathrm{sum}}
\newcommand{\newBasis}{\mathrm{new}\_\mathrm{basis}}
\newcommand{\True}{\mathrm{True}}
\newcommand{\False}{\mathrm{False}}
\newcommand{\selfCheck}{\mathrm{check}\_\mathrm{self}\_\mathrm{adjoint}}
\newcommand{\numSelf}{\mathrm{number}\_\mathrm{self}\_\mathrm{adjoint}}
\newcommand{\numConj}{\mathrm{number}\_\mathrm{conjugate}}
\newcommand{\selfBasis}{\mathrm{self}\_\mathrm{adjoint}\_\mathrm{basis}}
\newcommand{\conjBasis}{\mathrm{conjugate}\_\mathrm{basis}}
\newcommand{\additionalSelf}{\mathrm{extra}\_\mathrm{self}\_\mathrm{adjoint}}
\newcommand{\posDegSets}{\mathrm{positive}\_\mathrm{degenerate}\_\mathrm{sets}}
\newcommand{\negDegSets}{\mathrm{negative}\_\mathrm{degenerate}\_\mathrm{sets}}
\newcommand{\complexDegSets}{\mathrm{complex}\_\mathrm{degenerate}\_\mathrm{sets}}
\newcommand{\allDegSets}{\mathrm{all}\_\mathrm{degenerate}\_\mathrm{sets}}

\newcommand{\conjPairSets}{\mathrm{conjugate}\_\mathrm{pairs}}
\DeclareMathOperator{\BigO}{O}
\DeclareMathOperator{\poly}{poly}
\DeclareUnicodeCharacter{2009}{\,} 

\newcommand{\posBases}{\mathrm{positive}\_\mathrm{bases}}
\newcommand{\negBases}{\mathrm{negative}\_\mathrm{bases}}
\newcommand{\complexBases}{\mathrm{complex}\_\mathrm{bases}}

\newcommand{\compCheck}{\mathrm{checklist}\_\mathrm{complex}}
\newcommand{\posCheck}{\mathrm{checklist}\_\mathrm{positive}}
\newcommand{\negCheck}{\mathrm{checklist}\_\mathrm{negative}}
\newcommand{\conjTest}{\mathrm{conj}\_\mathrm{test} }

\newcommand{\C}{\mathbb{C}}

\begin{document}

\title{Fitting quantum noise models to tomography data}

\author[1,2]{Emilio Onorati}
\author[1,3]{Tamara Kohler}
\author[1]{Toby S. Cubitt}
\affil[1]{University College London, Department of Computer Science, UK}
\affil[2]{Technische Universit\"at M\"unchen, Fakult\"at f\"ur Mathematik, DE}
\affil[3]{Instituto de Ciencias Matem\'aticas, Madrid, ES}

\date{}

\maketitle

% ===============================================================================
% ================================== Abstract====================================
% ===============================================================================
\begin{adjustwidth}{0.3cm}{0.3cm}
  \begin{abstract}
The presence of noise is currently one of the main obstacles to achieving large-scale quantum computation. Strategies to characterise and understand noise processes in quantum hardware are a critical part of mitigating it, especially as the overhead of full error correction and fault-tolerance is beyond the reach of current hardware. Non-Markovian effects are a particularly unfavourable type of noise, being both harder to analyse using standard techniques and more difficult to control using error correction. In this work we develop a set of efficient algorithms, based on the rigorous mathematical theory of Markovian master equations, to analyse and evaluate unknown noise processes.

In the case of dynamics consistent with Markovian evolution, our algorithm outputs the best-fit Lindbladian, i.e., the generator of a memoryless quantum channel which best approximates the tomographic data to within the given precision. In the case of non-Markovian dynamics, our algorithm returns a quantitative and operationally meaningful measure of non-Markovianity in terms of isotropic noise addition.

We provide a Python implementation of all our algorithms, and benchmark these on a range of 1- and 2-qubit examples of synthesised noisy tomography data, generated using the Cirq platform. The numerical results show that our algorithms succeed both in extracting a full description of the best-fit Lindbladian to the measured dynamics, and in computing accurate values of non-Markovianity that match analytical calculations.
  \end{abstract}
\end{adjustwidth}

\tableofcontents

%=======================================================================
%=======================================================================
\section{Introduction}
%=======================================================================
%=======================================================================

A key challenge in developing medium to large-scale quantum architectures and to design algorithm for quantum advantage is error mitigation~\cite{PreskillNISQ}.
Unfortunately, noise adversely affects all stages of quantum computation, in particular the manipulation of states through quantum gates and continuous-time quantum processes.
On large scale devices, error correcting codes will be used to suppress noise and achieve fully fault-tolerant computation. But in the NISQ era, the overhead of full fault tolerance is prohibitive, placing it beyond reach of near-term hardware.
It is thus important to understand and characterise the underlying noise dynamics in current quantum devices, both in order to inform hardware design, and to prepare error correction and mitigation protocols optimised for the specific noise in the apparatus.

Various methods have been devised to evaluate and analyse noise in quantum dynamics, making different assumptions on, and providing different information about, the noise processes~\cite{benchmarkingmethods}.
One way to understand the noise model in a device is to look for compatible Markovian\footnote{We call a quantum channel \emph{Markovian} if it is a solution of a master equation with generator in Lindblad form~\cite{lindblad1976,Gorini76}. Equivalently, it is an element of a one-parameter continuous completely positive semigroup.} evolutions.
Knowing the Lindbladian which best approximates the generator of the physical process is considerably helpful to extract information about the noisy dynamics in a quantum device, which can be leveraged in multiple ways, for instance to calibrate error mitigation techniques.

Noise in quantum devices may however deviate from a memoryless dynamics, so that no compatible Markovian description exists.
This can be problematic: strategies for mitigating or correcting non-Markovian noise are currently less established than schemes designed for memoryless errors. For example, the fault tolerance thresholds that have been calculated for non-Markovian noise models~\cite{Terhal_2005,Aharonov_2006,Ng_2009} have worse estimates than traditional threshold theorems.
Therefore, methods which can depart from the Markovian assumption, or that can provide precise tests of non-Markovian dynamics, are of considerable interest.
A number of approaches to benchmark non-Markovianity have been devised indeed;
we give a brief overview in \cref{sec:RelatedWork}.
However, some of these procedures are impractical to compute, requiring either complete knowledge of the master equation, or an unbounded number of full process tomography routines, or again a specific experimental setup.
Other methods are only capable of providing one-sided witnesses of non-Markovianity, failing even in principle to identify all non-Markovian effects, and providing no information about the underlying dynamics in the Markovian case.
Thus, finding a theoretically well-motivated but also feasible and low-resource procedure to assess whether or not dynamics is consistent with Markovian evolution is desirable.

%----------------------------------------------------------------------
\subsection{Main results}\label{sec:main_results}
%----------------------------------------------------------------------
In this work we develop a family of algorithms -- implemented in Python and benchmarked on simulated data in Cirq -- which characterise and quantify the noise processes in a quantum system from one (or a time series of) tomographic snapshot(s), that is, tomographic reconstruction(s) of the quantum process under consideration at a certain time in the evolution. This means that our method has conveniently a low demand of quantum resources and does not require an ad-hoc experimental setup. 
At the core of our algorithms is an efficient convex optimisation programme over a matrix variable minimising the distance to the matrix logarithm of the snapshot.
The constraint of the minimisation task are the three necessary and sufficient conditions for a Lindblad generator. This idea is built on previous work~\cite{Markdynamics08}, where the quantum embedding problem is tested in terms of a semi-integer programme. 
Turning this approach into a convex optimisation yields to a series of favourable features, such as the inclusion of a tolerance parameter dealing with inaccuracies of the tomographic reconstruction, or the the possibility to retrieve a Lindblad generator in presence of weakly non-Markovian dynamics. Also, convex optimisation guarantees that the solution found by our procedure is the global optimal one, and we thus do not have to be concerned with the problem of local optima.
We also highlight that our scheme does not require any a priori knowledge or assumption regarding the noise model, nor expects a specific structure of the input data, thus being suited to be adopted in realistic scenarios.

We list in the following our main contributions.

\begin{itemize}
	
\item \textbf{Algorithm for extraction of the best fit Lindblad generator}
Our first algorithm receives as input a single tomographic snapshot $M$ and extracts the full description of the best fit Lindblad generator to the measured data. More specifically, it computes for a subset of branches of the matrix logarithm of $M$ the closest Lindblad operator, and returns the one generating the closest quantum map to the tomographic data. 
With our novel convex optimisation approach we can include a precision parameter $\varepsilon$ representing our desired error tolerance regarding consistency with a Markovian evolution: if no Lindbladian generating a channel within distance $\varepsilon$ from the snapshot $M$ is found, then the algorithm informs us that the investigated channel is not compatible with a Markovian dynamics within the indicated precision.

\item \textbf{Algorithm calculating amount of non-Markovianity}
Our second algorithm takes again as input a tomographic snapshot $M$ and a precision parameter $\varepsilon$ and returns the explicit description of the quantum channel within the $\varepsilon$-ball around $M$ having the smallest measure of non-Markovianity. We express this quantities in terms of the minimal amount of isotropic noise to be added to the generator of an hermiticity- and trace-preserving map in order to ``wash out'' memory effects and render the evolution Markovian. Such a measure for non-Markovianity, first defined in~\cite{Markdynamics08},  is operationally meaningful and directly analogous to robustness of entanglement measures in entanglement theory.
This second algorithm can hence be used as a complement for our first Lindbladian-extraction algorithm when this one fails in retrieving a valid Lindbladian within the required tolerance.

\item \textbf{Lindbladian extraction over time series of snapshots}
If our single-snapshot algorithm successfully retrieves a Markovian channel which approximates the input data to the desired precision, the physical channel is then considered \emph{consistent} with Markovian dynamics.
However, it is important to note that this snapshot compatibility does not guarantee with complete certainty that the physical process is actually Markovian: there may in fact exist non-Markovian dynamics which at some point in their evolution are compatible with a Markovian channel. One can strengthen the test by taking a series of snapshots of the same channel at different times and ascertain if each of them is separately consistent. 
But also in this case -- even if unusual -- it is still possible to run into non-Markovian quantum channels where each snapshot is separately consistent with a Markovian dynamics.
In order to overcome this challenge and obtain a more rigorous test of Markovianity, we have thus devised an extension of the algorithm which takes as input a time series of tomographic snapshots, and checks whether there is a time-independent Markovian map embedding (up to some approximation) all of them through its evolution.
This offers a significantly more sensitive test of a time-independent memoryless evolution, at only linear classical computational cost with respect to the number of snapshots.

\item \textbf{Extracting Lindbladian for time dependent Markovian evolution} In a more general case, snapshots of a Markovian evolution may not be consistent with a time independent Lindblad generator, e.g. when the underlying noise dynamics are fluctuating in time. To address this type of scenario, we have modified the time series algorithm leveraging the property of continuous divisibility of (time-dependent) Markovian quantum channels. This time dependent version of the algorithm is then designed to extract a sequence of Lindbladian whose consecutive evolutions approximates the input snapshots, offering a general approach to characterise noisy dynamics for ideal and approximate Markovian evolutions.

\item \textbf{Extending the method to quantum channels with degenerate spectrum} The algorithms devised in \cite{Markdynamics08} make a crucial assumption: in order to function they require that the input channel has a non-degenerate spectrum.
It might be hoped that this assumption was not a serious barrier to applying the algorithm in general cases;
after all, the set of matrices with non-degenerate spectrum is dense in the set of matrices, and as such the probability of a degenerate matrix resulting from carrying out tomography is vanishingly small.
However, in this work we demonstrate that such optimism is not well founded.
We will discuss why, given a Markovian quantum channel with degenerate spectrum as an input for our approach inspired by~\cite{Markdynamics08}, even the smallest perturbation breaking the hermiticity-preserving structure will lead to failure. 
Hence, to correctly analyse a quantum channel with degenerate spectrum, we will show how to correctly perturb into a non-degenerate matrix within arbitrarily close distance, in a way that preserves the hermiticity-preserving property and allows our approach to successfully identify a best-fit Lindbladian within the precision parameter.

On the other hand, issues arise when applying the algorithms of~\cite{Markdynamics08} to channels which are not degenerate themselves, but are the result of a perturbation of a degenerate quantum operator.
This is a very recurrent instance when characterising noise in quantum computing devices, since the channels of interest are typically experimentally implemented quantum gates, which all have degenerate eigenvalues.
The task of dealing with channels with degenerate spectrum is delicate, due to the sensitivity to perturbation of the Lindblad structure with multi-dimensional eigenspaces.

To overcome this severe problem, we develop a strategy to restore the Lindblad structure of multi-dimensional eigenspaces by leveraging techniques from matrix perturbation theory~\cite{StewartSun}, and implement this as a set of algorithms that serve as a pre-processing phase for the convex optimisation task\footnote{The pre-processing algorithms for the case of a perturbed degenerate channel still do not handle the case of exactly degenerate channels. However, we show that in the case where an exactly degenerate channel is given as input to the algorithm it can always be perturbed to an arbitrarily close non-degenerate channel while preserving the value of the non-Markovianity measure. While this situation is of less practical interest, it means that our algorithms make absolutely no assumptions about the form of the input channel.}. 
To justify mathematically this approach, we will prove a stability theorem under perturbation for multi-dimensional eigenspaces of hermiticity-preserving operators.

\end{itemize}

We accompany our theoretical analysis with a Python implementation of all the algorithms, which we benchmark numerically on simulated tomographic data in Cirq \cite{cirq}. The numerical results show that our algorithms successfully identify Markovian-compatible dynamics for a range of 1- and 2-qubit examples, both for noisy quantum gates with degenerate spectrum and non-degenerate quantum channels.
The numerics also confirm that our algorithms are able to compute accurate values for the non-Markovianity measure of noisy, non-Markovian quantum channels, which we show to be consistent with a calculation of this measure done by hand.

%--------------------------------------------------------------------
\subsection{Related work on assessing non-Markovian noise} \label{sec:RelatedWork}

The nature of non-Markovian noise has been investigated from a variety of perspectives and with a number of different approaches (cfr. review papers~\cite{RivasMarkReview,Addis_2014,LiHaWi18} and the introduction of~\cite{ChruManiscalco14}).
One of the principal ways to quantifying non-Markovianity is based on \emph{divisibility}~\cite{WolfDivisibility,Hou11}, i.e., the property of a channel encoding evolution from time $t_0$ to $t_1$ to be implemented as a concatenation of channels from time $t_0$ to $t$ and $t$ to $t_1$ for any $t \in [t_0,t_1]$.
Indeed, a channel which is Markovian, in the Lindblad sense, is also divisible.
However, the converse is not necessarily true.\footnote{The property of being `infinitesimal divisible' is known to be equivalent to time-dependent Markovianity, but that is a stronger requirement than divisibility (see~\cite{WolfDivisibility}).}
In \cite{PhysRevLett.123.040401} the authors tackle the problem of defining divisibility in an operational manner, and demonstrate that their operational definition of divisibility gives a stricter test of Markovianity than a purely mathematical definition; however they note that operational divisibility still does not imply Markovianity. 

The original measure of non-Markovianity~\cite{Markdynamics08}, on which this work builds, determines whether the observed tomographic data is consistent with a time-independent Markovian master equation. If not, it provides a quantitative measure of how far the observed dynamics is from the closest Markovian trajectory. This task (and also that of determining finite divisibility) was shown to be NP-hard in general~\cite{CEW09,Bausch16}, but efficiently (classically) computable for any fixed Hilbert space dimension.

Other methods for detecting and measuring non-divisibility and non-Mar\-ko\-viani\-ty are based on checking monotonicity of quantities that are known to decrease under completely positive maps such as quantum correlation (see~\cite{Rivas_2010}, known as the \emph{RHP measure)}, quantum coherence~\cite{Wu20}, quantum relative entropy~\cite{Usha11} or the quantum mutual information~\cite{Luo12}.

Another approach affirms that a non-Markovian map is one that allows information~\cite{Laine_2010}, e.g. the Fisher information~\cite{Fisherflow}, to flow from the environment to the system.
Non-Markovianity can also be quantified by considering the change in the distinguishability of pairs of input states~\cite{Breuer_2009}.
The observation of non-monotonic behaviour of channel capacity~\cite{Bogna13}, the geometrical variation of the volume of the set of physical states~\cite{LorPlaPat13}, correlations in Choi-matrices associated to process tensors of quantum channels \cite{PhysRevLett.120.040405}, ensembles of Lindblad's trajectories~\cite{Kade20} and deep-neural-network and machine learning~\cite{NiuTLS19,Luchnikov_2020,luchnikov2021probing} are among the many alternative strategies to quantify dynamics with memory effects.

A limitation of many of these measures is that they provide one-sided witnesses of non-Markovianity, but cannot show that the dynamics \emph{is} Markovian, or find the master equation consistent with or closest to the observed dynamics.

%================================================================
\bigskip
%================================================================

The paper is organised as follows.
In \cref{prelim} we introduce our notation and cover preliminaries on convex optimisation.
The precise notion of Markovianity is defined in \cref{sec:Markovianity}, where we also present the non-Markovianity measure first defined in \cite{Markdynamics08}.
In \cref{sec:algorithms} we illustrate and discuss our two core algorithms, one for the extraction of the best-fit Lindblad generator to the input data and a second one to calculate the non-Markovianity component of the measured channel.
In \cref{sec:multiple_snapshots_extensions} we cover the extension to the case of multiple snapshots, with time-series algorithms designed to test a quantum map for both time-independent and time-dependent Markovianity. 
In \cref{sec:eigenspaces} we tackle the crucial problem of analysing quantum operator having a degenerate spectrum, and discuss our strategy to overcome it by finding suitable set of vectors restoring the hermiticity-preserving structure of the input map.
In \cref{sec:main_alg} we present the main algorithm of our procedure, which includes both pre-processing and the convex-optimisation.
In the last part of the manuscript,  \cref{sec:examples}, we test the validity of our method on different one-qubit and two-qubit numerical examples.

%=======================================================================
\section{Notation and theoretical background}\label{prelim}
%=======================================================================
We denote elementary basis vectors by $\ket{e_j} = (0,\dots,1,0,\dots,0)^T$ with 1 in the $j$-th position. The maximally entangled state is $\ket{\omega}= \sum_{j=1}^d \ket{e_j,e_j}/\sqrt{d}$ and $\omega_\perp = \1 - \dyad{\omega}{\omega}$ is the projection onto its orthogonal complement.
We write $\Flip$ to denote the flip operator interchanging the tensor product of elementary basis vectors, i.e.~$\Flip \ket{e_j,e_k} = \ket{e_k,e_j}$. We will use the Frobenius norm on matrices, defined by $\Fnorm{M} = \sqrt{\sum_{j,k} \abs{m_{jk}}^2}$, and the 1-norm, $\norm{M}_1=\sum_{j,k} \abs{m_{jk}}$. They are both \emph{submultiplicative}, that is, they satisfy $\norm{AB} \leq \norm{A} \norm{B}$ for all square matrices $A$ and $B$.

%===========================================================
\subsection{Channel representations and $\NDsquare$ matrices}\label{sec:channel_representations}
%===========================================================

We will consider quantum channels of finite dimension only, i.e.~completely positive and trace preserving (CPT) linear operators acting on the space of $d \times d$  matrices. To represent a channel $\mc T$ as a $d^2\times d^2$ matrix $T$, we will adopt the elementary basis representation:
\begin{equation}
	T_{(j,k),(\ell ,m)} = \bra{e_j,e_k}T\ket{e_\ell,e_m} \coloneqq \Tr \big[ \dyad{e_k}{e_j} \mc T (\dyad{e_\ell}{e_m}) \big].
\end{equation}
This is sometimes called the \emph{transfer matrix} of the channel in the elementary basis.
The corresponding representation $\ket{v}\in\C^{d^2}$ of a $d\times d$ matrix $V$ on which the channel acts is then
\begin{equation}
	v_{j,k}=\braket{e_j,e_k}{v} \coloneqq \bra{e_j}V\ket{e_k} .
\end{equation}
In this representation, the action of the channel on a matrix becomes matrix-vector multiplication, and the composition of channels corresponds to the product of their respective matrix representations.

In order to formulate necessary and sufficient conditions for the generator of a Markovian evolution, we will also make use of another representation, the Choi-matrix (or Choi-representation), defined as
\begin{equation}
	\tau (\mc T) \coloneqq d \big(\mc T \otimes \mc I) (\dyad{\omega}) \big) .
\end{equation}
Conveniently, the two representations are directly related through the $\Gamma$-involution~\cite{Markdynamics08}, acting on the elementary basis as
\begin{equation}
\dyad{e_j,e_k}{e_\ell,e_m}^\Gamma \coloneqq \dyad{e_j,e_\ell}{e_k,e_m}.
\end{equation}
Explicitly, we have
\begin{equation}
	\tau = (T)^\Gamma \qquad \text{and} \qquad T=(\tau)^\Gamma .
\end{equation}
The Choi-representation is very useful to investigate the \emph{hermiticity-preserving} property of quantum channels, i.e., quantum maps $\mc T$ such that $\mc T(X^\dagger) = \mc T(X)^\dagger$ for all $X$. Indeed we have
\begin{lemma}
	$\mc T$ is hermiticity-preserving $\iff$ $\tau$ is hermitian.
\end{lemma}
\noindent We will use the terms hermiticity-preserving and Choi-hermitian interchangeably.

Crucial for our approach is the structure of the complex matrix logarithm of matrices with non-degenerate spectrum:

\begin{definition}[Non-Degenerate Spectrum ($\NDsquare$) matrices)]
We say a matrix is $\NDsquare$ if all its eigenvalues are unique. 
\end{definition}

An important property of $\NDsquare$ matrices is that they can be uniquely diagonalised.
Their matrix logarithm is given by an infinite number of branches indexed by a vector $\vec{m}=(m_1,\dots,m_{d^2}) \in \mathds{Z}^{d^2}$. 
The $0$-branch of the matrix logarithm of a diagonalisable matrix $T= \sum_{j=1}^{d^2} \lambda_j \dyad{r_j}{\ell_j} = \sum_{j=1}^{d^2}\lambda_j \, P_j$, with $\{\lambda_1,\dots\lambda_{d^2}\}$ being the eigenvalues of $T$ and $\ell_j, r_j$ the respective left and right eigenvectors such that $\braket{\ell_j}{r_k}=\delta_{jk}$ and $\dyad{r_j}{\ell_j} = P_j$, is given by
\begin{equation}
	G_0\coloneqq \log(T) = \sum_{j=1}^{d^2} \log \lambda_j \, \dyad{r_j}{\ell_j}
	=
	\sum_{j=1}^{d^2} \log \lambda_j \, P_j;
\end{equation}
the $\vec{m}$-branch is then
\begin{equation}
	G_{\vec{m}} \coloneqq G_0 +  \sum_{j=1}^{d^2} m_j \, 2 \pi \i \, \dyad{r_j}{\ell_j}
	=
	\sum_{j=1}^{d^2} m_j \, 2 \pi \i \, P_j.
\end{equation}

\paragraph{Modified hermitian adjoint} The hermitian adjoint operation on a matrix in its vector representation $\ket{v}$ in the elementary basis is given by $\ket{v^\dagger}\coloneqq \Flip \ket{v^\ast}$. We call vectors such that $\ket{v}=\ket{v^\dagger}=\Flip \ket{v^\ast}$ \emph{self-adjoint}, and we say that two vectors $v$ and $w$ are \emph{hermitian-related} if $\ket{w}=\ket{v^\dagger}=\Flip \ket{v^\ast}$ (and equivalently $\ket{v}=\ket{w^\dagger}=\Flip \ket{w^\ast}$).
This terminology is unusual with respect to the conventional hermitian conjugation operation on vectors, but the definition adopted here for $d^2$-dimensional vectors exactly corresponds to the usual hermitian adjoint of their corresponding $d \times d$ matrices on which $\mc T$ acts. We will conversely denote the standard hermitian conjugation for a matrix $A$ by $A^H$.

%===========================================================
\subsection{Convex optimisation programmes}\label{sec:preliminaries_convex_opt}
%===========================================================
At the heart of our algorithms are convex optimisation programmes which either retrieve the closest Lindbladian to the matrix logarithm, or the smallest non-Markovianity parameters (both objects explained into detail in \cref{sec:Markovianity}). These convex optimisations over a (scalar or vector) variable $x$ have the general form~\cite{ByodConvex}
\begin{equation}
	\begin{array}{lll}
	&\textbf{standard form}& \\
	\\
	\textsl{minimize} &f_0 (x)& \\
	\textsl{subject to} &f_j (x) \leq 0, & j=1,\dots,n \\
	 & \langle a_k,x \rangle = b_k  & k=1,\dots,m \\
	\end{array}
\end{equation}
\begin{equation}
	\begin{array}{lll}
	&\textbf{epigraph form} &\\
	\textsl{minimize} &\mu& \\
	\textsl{subject to} &f_0 (x) -\mu \leq 0& \\
	&f_j (x) \leq 0, & j=1,\dots,n \\
	&\langle a_k,x\rangle = b_k  & k=1,\dots,m
	\end{array}
\end{equation}
where $f_0,f_1,\dots,f_n$ are convex functions. A fundamental property of convex optimisation problems is that any locally optimal point is also globally optimal.

A class of convex optimisation problems, called \emph{second-order cone programmes}, takes the form
\begin{equation}
	\begin{array}{lllclll}
	\textsl{minimize} &\langle f,x \rangle& \\
	\textsl{subject to} &\norm{A_jx+b_j}_2 \leq \langle c_j,x\rangle +d_j& j=1,\dots,n \\
	&Fx=g &
	\end{array}
\end{equation}
where $f$, $g$, $ b_j$'s and $c_j$'s are vectors (not necessarily of the same dimension), $F$ and $A_j$'s are matrices and $d_j$'s are scalars. The condition $\norm{A_jx+b_j}_2 \leq \langle c_j,x\rangle +d_j$ is called \emph{second-order cone constraint}. Important for our work is the fact that minimization objectives for the Frobenius norm over matrix variables can be converted into a second order cone programme via epigraph formulation.

To numerically implement convex optimisation programme, we make use of the Python library CVXPY~\cite{diamond2016cvxpy,agrawal2018rewriting}.

%=============================================================================
\subsection{Quantum Markovian channels and embedding problem}\label{sec:Markovianity}
%===============================================================================

Processes which do not retain any memory of their previous evolution are called \emph{Markovian} and satisfy the \emph{Markov property}: given a sequences of points in time $t_1<t_2<,\dots,<t_{n-1} < t_n$, a stochastic process $X_t$ taking values on a countable space has the Markov property if
\begin{equation}
\PP (X_{t_n+s}=y \, \vert \, X_{t_n} = y_n,\dots,X_{t_2} = y_2, X_{t_1} = y_1, )
=
\PP (X_{t_n+s}=y \, \vert \, X_{t_n} = y_n)
\end{equation}
for any $s>0$.

Extending this notion, a \emph{quantum Markov process} is described by a one-parameter semi-group giving rise to a continuous sequence of completely positive and trace preserving (CPT) channels. The generators of this type of semi-group is called \emph{Lindbladian} and must take the well-known \emph{Lindblad form}~\cite{lindblad1976,Gorini76},
\begin{equation}\label{eq:Lindbladian_equation}
	\mathfrak{L}(\rho) \coloneqq \i[\rho,H]
	+ \sum_{\alpha,\beta} \xi_{\alpha,\beta}
	\left[
		F_\alpha \rho F_\beta^\dagger -\frac{1}{2} \left( F_\beta^\dagger F_\alpha \rho + \rho F_\beta^\dagger F_\alpha \right)
	\right] ,
\end{equation}
where $H$ is hermitian, $\Xi=(\xi_{\alpha\beta})$ is positive semi-definite and $\set{F_\alpha}_\alpha$ are orthonormal operators.
The first term on the RHS is the Hamiltonian part and describes the unitary evolution of the density operator, while the second term represents the dissipative part of the process.
By diagonalising $\Xi$ as $ \Gamma \coloneqq U^\dagger \Xi U  = \mathrm{diag} (\gamma_\alpha)$ and defining the so-called \emph{jump operators}
$J_\alpha \coloneqq \sqrt \gamma_\alpha \sum_\beta u_{\beta\alpha} F_\beta$, we can re-write \cref{eq:Lindbladian_equation} in diagonal form,
\begin{equation}\label{eq:Lindbladian_diagonal}
\mathfrak{L}(\rho) \coloneqq \i[\rho,H]
+ \sum_\alpha
\left[
J_\alpha \rho J_\alpha^\dagger -\frac{1}{2} \left( J_\alpha^\dagger J_\alpha \rho + \rho J_\alpha^\dagger J_\alpha \right)
\right] .
\end{equation}

The question whether a given quantum map $\mc M$ is compatible with a Markovian process, in the sense that there exists a memoryless evolution that at a certain time is equal to $\mc M$, has been investigated from different perspectives, e.g. in the context of complexity~\cite{CEW09}, channel divisibility~\cite{WolfDivisibility}, regarding spectrum~\cite{davies2010} and toward the goal of achieving a quantum advantage~\cite{Lostaglio20}, and it is sometimes referred to as the \emph{embedding problem}. A method to ascertain whether a given channel is compatible with a Markovian dynamics has been developed in ref.~\cite{Markdynamics08}, which provides three properties for $L$ that are necessary and sufficient to satisfy \cref{eq:Lindbladian_equation} (where $L$ is the elementary basis representation of $\mathfrak{L}$). These are
\begin{enumerate}[label=(\roman*)]
	\item $L$ is hermiticity-preserving, that is, $L\ket{v^\dagger} =(L\ket{v})^\dagger$ for all $\ket{v}$.
	\item $(L)^\Gamma$ is \emph{conditionally completely positive}~\cite{Evans77}, that is,
	$\omega_\perp \, (L)^\Gamma \, \omega_\perp \geq 0,$
	where $\omega_\perp = (\1 - \dyad{\omega}{\omega})$.
	\item $\bra{\omega} L= \bra{0}$, which corresponds to the trace-preserving property.
\end{enumerate}
A time-independent Markovian quantum channel at time $t$ is then given by $T(t)=\e^{L t}$ for some time-independent Lindblad operator~$L$.
We will call \emph{quantum embeddable} any map whose matrix logarithm admits at least one complex branch satisfying these properties.

Interpreting the above conditions, we observe that they impose a rigid structure on the operator in matrix form, since they involve both eigenvalues and eigenvectors. In particular, from the hermiticity-preserving condition we note that, if $\lambda$ is an eigenvalue of $L$ and $\ket{v}$ the corresponding eigenvector, then it follows that:
\begin{equation}
	 L \ket{v^\dagger} = (L\ket{v})^\dagger = (\lambda\ket{v})^\dagger = \lambda^\ast \ket{v^\dagger}.
\end{equation}
Thus $\lambda^\ast$ and $\ket{v^\dagger}= \Flip \ket{v^\ast}$ are an eigenvalue and the corresponding eigenvector of $L$ too.

This implies an important property of Lindbladians in their elementary basis representation $L$: complex eigenvalues must necessarily come in complex-conjugate pairs $\lambda$, $\lambda^\ast$ and have the same multiplicity. Moreover, the eigenspace of~$\lambda$ must admit a set of basis vectors whose hermitian conjugates span the eigenspace of~$\lambda^\ast$. If $\lambda$ is real, then it must necessarily admit a set of vectors spanning its eigenspace which either come in hermitian-related pairs or are self-adjoint. For a Markovian channel $T=e^{L}$ the conditions for eigenvalues that are either complex or positive and for the vectors spanning their subspaces follow exactly the same rules. On the other hand, negative eigenvalues of $T$ must have even multiplicity (implying that no non-degenerate negative eigenvalue can occur); the eigenspace of a negative eigenvalue must admit a basis of hermitian-related pairs of vectors. This particular structure for Markovian maps and their Lindblad generators will be leveraged in \cref{sec:pre-processing} to reconstruct the eigenspace of an originally degenerate eigenvalue from a set of perturbed eigenvectors.

Conditions~(i), (ii) and~(iii) will run throughout our analysis, and we will implement them in the algorithms in~\cref{sec:algorithms} as constraints of a convex optimisation problem. We will also sometimes need to express these conditions in the Choi representation. By expanding condition~(iii) as
\begin{align}
\bra{\omega} L
	&=
	\frac{1}{\sqrt{d}} \sum_j \bra{j,j} \sum_{\substack{a,b \\ c,r}} L_{\substack{a,b\\c,r}} \dyad{a,b}{c,r}
	=
	\frac{1}{\sqrt{d}} \sum_{j,c,r} L_{\substack{j,j \\ c,r}} \bra{c,r}\\
	&=
	\frac{1}{\sqrt{d}} \sum_{j,c,r} \left((L)^\Gamma\right)_{\substack{j,c \\ j,r}} \bra{c,r} = \bra{0},
\end{align}
we can re-formulate it as
\begin{equation}
	\Tr_1 \left[(L)^\Gamma\right]
	=
	\sum_{j,c,r} \left((L)^\Gamma\right)_{\substack{j,c\\j,r}} \dyad{c}{r}
	=
	0_{d,d} .
\end{equation}
This is equivalent to $\Big\lVert \Tr_1 \Big[(L)^\Gamma\Big] \Big\rVert = 0$ in any matrix norm.

\bigskip

In order to quantify the degree of non-Markovianity of quantum channels, we consider again \cite{Markdynamics08} where a computable, continuous and basis-independent measure of non-Markovianity has been introduced. This evaluates the minimal amount of isotropic noise to render a quantum process Markovian, and more precisely, the smallest value $\mu$ such that $L= G_{\vec{m}} - \mu \, \omega_\perp$ is a legitimate Lindblad operator for some branch $G_{\vec{m}} $ of the generator of an hermiticity- and trace-preserving quantum channel~$M$.
Moving to the Choi representation,
we define
\begin{equation}\label{eq:definition_mu}
	\mu_{\min} (M) \coloneqq \min_{\vec{m}} \min \Big\{ \mu\, : \omega_\perp (G_{\vec{m}} )^\Gamma \omega_\perp + \frac{\mu}{d}\1 \geq 0\Big\}
\end{equation}
as the \emph{non-Markovianity parameter} of $M$. In this second formulation we recognise $\mu_{\min}$ as the smallest value to be added to the spectrum of the Choi matrix of~$G_{\vec{m}} $ to fulfill the conditionally completely positive condition, hence ensuring that the so modified generator gives rise to a completely positive map at any time. As previously discussed, this is a necessary property of memoryless dynamics.

This measure is very much alike the one proposed in~\cite{Rivas_2010}, calculating the distance of the averaged map from a completely positive channel, namely,
\begin{equation}
	\mu_{\mathrm{RHP}} \coloneqq \int_0^\infty dt \lim_{\varepsilon \rightarrow 0} \frac{\Vert\1 + \varepsilon L^\Gamma\Vert_1 - 1}{\varepsilon}.
\end{equation}
Furthermore, a connection between $\mu_{\mathrm{RHP}}$ and another non-Markovianity measure based on the non-monotonicity of the trace distance of two quantum states evolving under the same (non-Markovian) channel~\cite{Breuer_2009} has been rigorously established~\cite{ChruManiscalco14}, thus linking -- at least on a mathematical level -- our measure to information backflow.
A physical interpretation for measures restoring the conditionally completely positive conditions like our $\mu_{\min}$ parameter has been discussed in~\cite{Addis_2014} in relation to \emph{reverse jump operators}. These appear in the non-Markovian master equation when the continuous completely positive condition is violated~\cite{Piilo08} and recover previously lost coherence, thus manifesting of memory effects.

%==============================================================
%==============================================================
\section{Core algorithms for retrieval of best-fit Lindbladian and computing non-Markovianity measure $\mu_{\min}$} \label{sec:algorithms}
%==============================================================
%==============================================================

In this section we present a high level overview of the core algorithms developed in this paper.
For more information, including detailed pseudocode; for more details we refer the Reader to \cref{app:optimisation}.
The code itself is available at \cite{gitlab_repo}.

Here we cover the case of investigating non-Markovianity with a single tomographic snapshot of the quantum process under consideration.
In this scenario we have developed two algorithms: one to find the quantum Markovian channel which best fits the tomographic snapshot, and a second algorithm to compute the measure of non-Markovianity~$\mu$, which we can run e.g. when no Markovian channel is found in the proximity of the input map.
These algorithms stem from the semi-integer programme worked on in \cite{Markdynamics08}. Compared to this previous formulation, our novel convex optimisation approach has the very desirable benefit of dealing with weakly non-Markovian noise and tomographic inaccuracies by introducing an error tolerance parameter. Moreover, it guarantees that the solution is a global optimum. 
Our choice for the Frobenius norm in the objective function is indeed dictated by the fact that, as mentioned in \cref{sec:preliminaries_convex_opt}, this yields to a second-order cone programme, whereas other more physically-motived norms such as the diamond norm would not allow us to work with a convex-optimisation frame.

More specifically, as only input we are given the tomographic snapshot $M$ and a desired tolerance parameter~$\varepsilon$.
\cref{alg:MR:Lindbladian} looks for the closest Lindbladian to $\log M$ in the $\varepsilon$-neighbourhood by casting a convex optimisation task over a matrix variable $X$ whose constraints are exactly the necessary and sufficient conditions for a Lindblad generator in the Choi representation, as discussed in \cref{sec:Markovianity}.

We iterate over a restricted set of branches of the matrix logarithm -- limited by the parameter $m_{\max}$ -- and pick the resulting Lindbladian from the convex optimisation programme whose generated channel is the closest to~$M$.
Iterating over different branches of the matrix logarithm is necessary, as in the general case not all branches (which the 0-branch may not necessarily part of) of a Markovian quantum channel fulfil the conditions for a Lindblad operator.
Hence, while all being mathematically equivalent, only a subset of branches of  $\log M$ may be close to the physical generator of a Markovian quantum channel. It is e.g. immediate to see that not all branches will preserve the existence of conjugate pairs $\lambda,\lambda^\ast$ for the hermiticity-preserving property.
As discussed in the preliminaries, infinitely many branches exist, which would clearly make a complete search unfeasible. Fortunately, both heuristic and theoretical argument show that a restricted search of few branches around the 0-branch suffices: more details are given in \cref{sec:Mloop}.

\bigskip

%===============================================
%-------------------First algorithm------------
%===============================================
\begin{algorithm}\small
	\SetKwInOut{Input}{Input}\SetKwInOut{Output}{Output}
	\Input{matrix $M$ positive real number $\varepsilon$, positive integer $m_{\max}$}
	\Output{$L$ closest Lindbladian to $\vec{m}$-branch of $\log M$ such that $\Fnorm{M-\exp L} < \varepsilon$ is minimal over all
		$\vec{m}\in\{-m_{\max},\dots,0,\dots,m_{\max}\}^{\times d^2}$}
	$G_0 \leftarrow \log(M)$ \\
	\For{$G_{\vec{m}} \leftarrow G_0 + 2\pi \i \sum_{j=1}^{d^2} m_j \, P_j $ \  \textit{(branches of $G_0$)}}{
		\smallskip
		Run convex optimisation programme on variable $X(\vec{m})$:\\
		\Indp $\begin{array}{ll}
			\text{minimise} & \Fnorm{X(\vec{m})-L_{\vec{\vec{m}}}^\Gamma}   \\
			\text{subject to } &X(\vec{m}) \text{ hermitian}\\
			&\omega_\perp X(\vec{m}) \omega_\perp \geq 0 \\
			&\norm{\Tr_1[X(\vec{m})]}_{1} = 0
		\end{array}$\\
		\Indm
		
		\smallskip
	
		\If{$\Fnorm{M - \exp X^\Gamma(\vec{m})}<\varepsilon$}{
			
		\smallskip
			
		Store $X (\vec m)$ \\
		$\mathrm{distance}(\vec{m}) \leftarrow \Fnorm{M - \exp X^\Gamma(\vec{m})}$}
	
		\medskip
		
	}%end for loop
	\Return$L=X^\Gamma(\vec{m}') \text{ for } \vec{m}'=\mathrm{argmin}\, \{\mathrm{distance}(\vec{m})\}$
	\caption{Retrieve best-fit Lindbladian}\label{alg:MR:Lindbladian}
\end{algorithm}

If no Lindblad generator is found in $\varepsilon$-ball of $M$,  
\cref{alg:MR:NoMark_parameter}  is designed to find the channel in the $\varepsilon$-neighbourhood with the smallest non-Markovianity parameter according to the definition in \cref{eq:definition_mu}. This can again be retrieved by formulating a convex optimisation programme.

%===============================================
%-------------------Second algorithm------------
%===============================================

\begin{algorithm}[t]\small
	\SetKwInOut{Input}{Input}\SetKwInOut{Output}{Output}
	\Input{matrix $M$, positive real number $\varepsilon$, positive integer $m_{\max}$}
	\Output{smallest positive real $\mu_{\min}$ and hermiticity-preserving map $G'$ such that $G'-\mu_{\min}\omega_\perp$ is a Lindbladian and $\Fnorm{M-\e^{G'}}<\varepsilon$ }
	
	\medskip
	
	$\mu_{\min} \in \mathds{R}^{+}$ \hspace{12pt} \textit{(initialize it with an high value)}
	
	\smallskip
	
	$\delta$ satisfying $\varepsilon = \exp ( \delta) \cdot \delta \cdot \Fnorm{G_0}$, \hspace{8pt} $\delta_{\max}=10 \cdot \delta$, \hspace{8pt} $\delta_{\mathrm{step}} \in \mathds{R}^{+}$ \\
	$G_0 \leftarrow \log(M)$ 
	\medskip
	
	\For{$\delta < \delta_{\max}$}{
		
		\smallskip
		
		\For{$G_{\vec{m}} \leftarrow G_0 + 2\pi \i \sum_{j=1}^{d^2} m_j \, P_j $ \  \textit{(branches of $G_0$)}}{

			Run convex optimisation programme on variables $\mu(\vec{m},\delta)$ and $X(\vec{m},\delta)$:\\
			\Indp $\begin{array}{ll}
				\text{minimise} & \mu(\vec{m},\delta) \\
				\text{subject to } &X(\vec{m},\delta) \text{ hermitian}\\
				&\Fnorm{X(\vec{m},\delta)-G_{\vec{m}} ^\Gamma } \leq \delta \\
				&\omega_\perp X(\vec{m},\delta) \omega_\perp + \mu \frac{\1}{d} \geq 0 \\
				&\norm{\Tr_1[X(\vec{m},\delta)]}_{1} = 0
			\end{array}$\\
			\Indm
			
			\medskip
			
			\If{$\Fnorm{M - \exp X^\Gamma(\vec{m},\delta) } < \varepsilon$}{
				
				\smallskip
				
			Store  $(\mu(\vec{m},\delta), X^\Gamma(\vec{m},\delta))$}
			
		}%end for loop in m
		
		$\delta \leftarrow \delta+\delta_{\mathrm{step}}$
		
	}%end for loop for delta
	
	\smallskip

	\Return{$( \mu_{\min},G')=( \mu(\vec{m}',\delta'), X^\Gamma(\vec{m}',\delta') )$ \hspace{5pt} for $(\vec{m}',\delta')=\mathrm{argmin} \ \mu(\vec{m},\delta)$ }
	
	\caption{Compute non-Markovianity measure~$\mu_{\min}$}\label{alg:MR:NoMark_parameter}
\end{algorithm}

Here the variable $\delta$ reflects the neighbourhood around $G_{\vec{m}} $ that maps under the exponential into the $\varepsilon$-neighbourhood of $M$; a lower bound is given by Lemma 11 in~\cite{CEW09}, that is,  $\varepsilon \leq \exp ( \delta) \cdot \delta \cdot \Fnorm{G_0}$. An upper bound is again provided in Corollary 15 of~\cite{CEW09}. However, the boundaries of this region around $G_{\vec{m}} $ related to the $\varepsilon$-ball are not precisely characterised. This can cause problems for the convex optimisation programme, which in some cases can retrieve a Lindbladian for $\delta$ in the interval defined by the above bounds whose generated map falls outside the $\varepsilon$-ball. We further discuss this matter with a relevant example in \cref{sec:Unital_quantum_channel} (cfr. also \cref{fig:mu_illustration}). To solve this issue in a practical way, we run over increasing $\delta$ values up to 10 times the value determined by the lower bound.

Note that throughout this section we have assumed that the input for all the algorithms is $\NDsquare$.
However, following the argument of~\cref{sec:matrix_perturbation} it is straightforward to include the possibility of inputs which are not $\NDsquare$.
This would just require an additional step to check the input, and if it was not $\NDsquare$ perturb it to a matrix that is.

%-----------------------------------------------------------------------------------------
\paragraph{The error tolerance parameter}~\newline 
%-----------------------------------------------------------------------------------------
Both algorithms take as input a value~$\varepsilon$, which can be thought of as an error tolerance parameter allowing for a flexible treatment of noisy channels in realistic scenarios and at the same time allowing us to constraint our search over solutions within an expected interval, and observe if we find any proper outcome. In other words, this parameter can be thought as the boundary within which we are accepting a Markovian estimation to the observed channel.
This input is a figure that encompasses different characteristics, both theoretical and experimental, namely (i)~tomographic inaccuracies (ii)~state preparation and measurement errors, 
(iii)~the relation between two matrices and their logarithms, which makes the task of finding the closest Markovian map distinct -- but related -- to the task of finding the closest Lindblad generator (we refer the Reader to \cref{eq:relation_log,eq:tolerance_and_relation_log} and the text wrapping them),
(iv)~the amount of non-Markovianity of the ideal quantum channel and the noisy dynamics.
Precisely establishing the value for~$\varepsilon$ may thus be complicated without prior information on the noise model and the quantum apparatus. In practice this is however a very easily treatable problem, that we can resolve by making an initial educated guess for this parameter, then varying it over different calls of our algorithm and subsequently observing how the outputs differ; we took this approach to numerically investigate a non-Markovian quantum channel as illustrated in \cref{sec:Unital_quantum_channel}. 
Also, in the situation where no information is known, or where we are interested in retrieving the closest Lindblad generator without testing whether this is within a specific region around the physically observed map, the~$\varepsilon$ parameter can be set to some arbitrarily large value.

While the output of our procedure is the Markovian channel which is closest to the tomographic snapshot, it may be noted that unless we have a more informed idea on error bars any Markovian channel in the $\varepsilon$-ball has equally good claim to being the Lindbladian driving the physical process: the best-fit operator does not automatically correspond to the most likely one. 
Our best-fit approach is motivated by the willingness to design a practically feasible and robust scheme based on the well-known method of convex optimisation, requiring as few assumptions or prior information as possible.
We also argue that retrieving a Lindbladian compatible with the snapshot using the optimisation approach delivers insightful information about the noise process, even if it does not corresponds to the exact physical generator.

%-----------------------------------------------------------------------------------------
\paragraph{Treating channels with degenerate spectrum (cfr. \cref{sec:eigenspaces})}~\newline 
%-----------------------------------------------------------------------------------------
A critical case is the treatment of the experimental implementation of a channel having degenerate spectrum. Notably, this is the case of any unitary transformation. Without a specific pre-processing phase, the convex optimisation approach (as well as other schemes) is bound to fail in recognising a legitimate Lindbladian, even for very small perturbations of the ideal Markovian process. We examine this problem (with a concrete example) in \cref{sec:perturbation_example}. 
Solving this very fundamental issue is indeed one of the main technical contributions of this work. The strategy will be about searching new set of vectors spanning a perturbed multi-dimensional eigenspace satisfying the hermiticity-preserving condition and using them to construct a new matrix with the same spectrum as the original snapshot, but now with a different set of eigenvectors; we explain the details in \cref{sec:reconstructing_HP_structure}, then implemented as a pre-processing algorithm as discussed in \cref{sec:pre-processing}. Importantly, this approach is rigorously justified in the light of matrix perturbation theory: in \cref{sec:stability_theorem} we prove that, given an hermiticity-preserving operator with a pair of degenerate complex eigenvalues, it is always possible to retrieve a set of vectors spanning the two perturbed multi-dimensional eigenspaces respecting an hermiticity-preserving structure. The same can be done for a degenerate real eigenvalue.

%-----------------------------------------------------------------------------------------
\paragraph{Connection to other post-processing techniques}~\newline 
%-----------------------------------------------------------------------------------------
While in principle the tomographic snapshot required as input $M$ of our algorithm can be an arbitrary $\NDsquare$ matrix, we interpret it as being the tomographic reconstruction of a quantum channel at a certain time in the evolution. 
This means that ideally we would like to analyse data where state preparation and measurement errors have been filtered out, 
so that $M$ represents a proper (completely positive) quantum channel obtained from the raw data of the experimental observation. However, thanks to the flexibility given by the convex-optimisation approach, we can deal with input matrices which still incorporate state preparation and measurement errors (even though this will inevitably affect the accuracy of the Lindbladian estimation) and with tomographic errors.
In particular, we can accept negative probabilities resulting from limits of quantum tomography, which may infer negative eigenvalues. While this shortcoming affects methods such as \emph{maximum likelihood estimation}~\cite{MLE97,MLE01} yielding a potentially problematic rank-deficient density state reconstruction~\cite{MLEreliable}, our scheme removes the minus sign of isolated negative eigenvalues (by getting rid of the $\i \pi$ term in the logarithm) returning a strictly positive one.

%==============================================================
%==============================================================
\section{Extending to the case of multiple snapshots} \label{sec:multiple_snapshots_extensions}
%==============================================================
%==============================================================

Two of the major contributions of this work involve the extension to the case where we are given multiple tomographic snapshots of a channel.
In this situation there are two criteria we can test for -- whether there exists some time-independent Markovian evolution which is compatible with all snapshots, and whether the channel is consistent with time-dependent Markovian evolution.

For the case of testing for compatibility with time-independent Markovian evolution the theory is largely unchanged from the single snapshot case, and we cover the updates required to the algorithm in \cref{sec:mulSnap}. 
Testing whether the channel is consistent with time-dependent Markovian evolutions requires some additional theoretical work, which we cover in the remainder of this section. 

We view the time-dependent and the time-independent schemes as complementary to each other in practice, rather then just viewing the former as a refinement of the latter.
Especially for small time scales, we expect both approaches to successfully fit the observed data, where the time-independent algorithm returning a single Lindbladian may be more convenient for experimentalists to design error mitigation techniques. Conversely, for longer run times one would expect the time-dependent version to successfully identify a set of Lindbladians fitting the data for any continuous Markovian process affected by noise varying over time, while the time-independent algorithm certifies stable evolutions driven by a constant Lindblad generator.

%-----------------------------------------------------------------------------------
\subsection{Testing time-independent Markov evolution from time series} \label{sec:mulSnap}

The approach for single snapshot can be extended to multiple tomographic snapshots taken at a sequence of times. Augmenting the number of measurements is a way to make the conditions for a compatible Markovian evolution, or detection of non-Markovian effects, more stringent, since the requirement is that there exist a single time-independent Lindbladian that generates a dynamical trajectory that passes close to \emph{every} snapshot.
We highlight that this extension comes with a classical computational overhead that is only linear in the number of snapshots.

Consider a sequence of tomographic snapshots $M_1,\dots, M_N$ associated with measurement times $t_1,\dots,t_N$.
In \cref{alg:MR:multiple_snapshots}  
we formulate a convex optimisation programme that finds the Lindbladian minimising the sum of the distances from the logarithms of the input matrices (for simplicity, we present the algorithm for the case of $\NDsquare$ matrices with no cluster of eigenvalues, and discuss the more general case in \cref{app:mulSnap}). We iterate over different branches and pick the Lindbladian $L$ for which $\sum_c \Fnorm{M_c - \exp  t_c\, L}$ is the smallest. 
We certify the evolution as Markovian if the distance in each snapshot is less than~$\varepsilon$.

\begin{algorithm}[h]\small
	\SetKwInOut{Input}{Input}\SetKwInOut{Output}{Output}
	\Input{$(d^2 \times d^2)$-dimensional matrices $M_1\dots,M_N$ , positive real numbers $t_1,\dots,t_N$, positive real number $\varepsilon$, positive integer $m_{\max}$}
	\Output{$L$ Lindbladian minimising $\sum_{c=1}^N \Fnorm{t_c\,L-\log M_c}$ such that $\Fnorm{M_c-\exp t_c\,L} < \varepsilon$ for all $c$}
	
	\medskip

	\smallskip
	
	\For{$G^c_{\vec{m}} \leftarrow G^c_0 + 2\pi \i \sum_{j=1}^{d^2} m_j \, P^c_j $ \  \textit{(branches of $G_0$)}}{

		Run convex optimisation programme on variable $X(\vec{m})$:\\
		\Indp $\begin{array}{ll}
			\text{minimise} & \sum_c \Fnorm{ t_c\,X(\vec{m})-(G_{\vec{m}} ^c)^\Gamma}   \\
			\text{subject to } &X(\vec{m}) \text{ hermitian}\\
			&\omega_\perp X(\vec{m}) \omega_\perp \geq 0 \\
			&\norm{\Tr_1[X(\vec{m})]}_{1} = 0\\
		\end{array}$\\
		\Indm
		
		\smallskip
		\If{$\Fnorm{M_c - \exp  t_c\, X^\Gamma(\vec{m})} < \varepsilon$ for all $c=1,\dots,N$}{
			
			\smallskip 
			
			Store $X(\vec m)$ \\
			$\mathrm{distance}(\vec{m}) \leftarrow \sum_c \Fnorm{M_c - \exp  t_c\, X^\Gamma(\vec{m})}$}

	}
	\medskip
	
	\Return $L=X^\Gamma(\vec{m}') \text{ for } \vec{m}'=\mathrm{argmin}\, \{\mathrm{distance}(\vec{m})\}$
	\caption{Retrieve best-fit Lindbladian for multiple snapshots}\label{alg:MR:multiple_snapshots} 
\end{algorithm}

\FloatBarrier

%-------------------------------------------------------------------------------------
\subsection{Characterising time-dependent Markovian dynamics} \label{sec:timeDep}
%-------------------------------------------------------------------------------------

The more general instance of a time-dependent quantum Markov process $\Phi$ from time $t_1$ to $t_2$ can be written in terms of a \emph{time-ordered integral} corresponding to the limit of a product integral of time-independent Lindbladians $L_j$ satisfying \cref{eq:Lindbladian_equation}, that is,
\begin{equation}
	\Phi(t_1,t_2) 
	=
	\mathbb{T} \exp \int_{t_1}^{t_2} L(t) \dt
	=
	\lim_{\dt\rightarrow 0} \prod_{j=N}^1 \e^{L_j},
\end{equation}
where $\mathbb T$ is the time-ordering operator, $ N=(t_2-t_1)/\dt$. Then, $\Phi(t_1,t_2) $ is divisible at any intermediate time~\cite{WolfDivisibility}, namely,
\begin{equation}\label{eq:divisibility_property}
	\Phi (t_1,t_3) = \Phi (t_2,t_3) \Phi(t_1,t_2) \quad \text{for all} \ 0\leq t_1\leq t_2 \leq t_3.
\end{equation}

We will leverage this divisibility property to construct a series of operators from a time series of tomographic snapshots $M_1,M_2, \dots, M_N$  taken over a total time $\mathfrak T$. Our goal is to extract a set of Lindblad operators whose evolution approximates each snapshot $M_j$ within some tolerance parameter~$\tau$.
Iteratively, we then construct $T_p = M_p M_{p-1}^{-1}$ for $p= 2, \dots , N$ with~$T_1 = M_1$. 
Under the  Markovian assumption, the divisibility property in \cref{eq:divisibility_property} yields that all $T$s are completely positive and trace preserving maps. 

\begin{figure}[h]
	\begin{center}
	\begin{tikzpicture}
		
		\node at (-.75,.6) [align=center] {\color{black}{Snapshots}};
		\node at (-.75,-1.1) [align=center] {\color{blue}{Lindbladians} \\ \color{blue}{$\Theta_p \approx \e^{L_p}$}};
		
		\draw [line width=2pt, mygrey] (0,0) -- (6.5,0);
		\draw [dotted,line width=2pt, mygrey] (6.75,0) -- (7.4,0);
		\draw [->,line width=2pt, mygrey] (7.6,0) -- (9,0);
		
		\node at (10,0) [align=center] (a){\color{mygrey}{quantum} \\ \color{mygrey}{channel}};
		
		\node at (0,0) [fill,circle](b0) {};
		\node at (-.85,0) [align=center] {$t=0$};
		
		\foreach \y in {1,2,3}{
			\node at (\y*2,0) [fill,black,circle](b\y) {};
			\draw [->,very thick,black] (b0.north) to[out=45,in=135] (b\y.north);
			\node at (.4+2*\y,.6) [align=center] {\color{black}{$M_{\y}$ at $t_\y$}};
			\node at (2*\y,-.65) [align=center] {\color{mygreen}{$\Theta_{\y}$}};
			
			\pgfmathtruncatemacro{\z}{\y-1}
			
			\draw [->,very thick,mygreen] (b\z.south) to[out=300,in=240] (b\y.south);
			\node at (-1+2*\y,-1.1) [align=center] {\color{blue}{$L_{\y}$}};
		}
		
		\node at (8,0) [fill,black,circle](bN) {};
		\node at (9 ,.6) [align=center] {\color{black}{$M_{N}$ at $t_N=\mc T$}};
		\node at (8.1,-.65) [align=center] {\color{mygreen}{$\Theta_{N}$}};
		\draw [->,very thick,black] (b0.north) to[out=45,in=135] (bN.north);
		\draw [->,very thick,black] (b0.north) to[out=45,in=135] (b1.north);
		
		\node at (7,-1.1) [align=center] {\color{blue}{$L_4$ to $L_N$}};
		
		\draw [->,very thick,mygreen] (b3.south) to[out=300,in=170] (6.65,-.6);
		\draw [dotted,very thick,mygreen] (6.75,-.6) -- (7.4,-.6);
		\draw [->,very thick,mygreen] (7.45,-.6) to[out=10,in=240] (bN.south);
		
	\end{tikzpicture}
	\end{center}
\end{figure}

As a second assumption for $\Phi(t)$ in addition to Markovianity,
we specifically consider an evolution where the Lindbladian is time-dependent but with reasonably small fluctuation; more specifically, we impose Lipschitz continuity
$\norm{L(t_2)-L(t_1)} \leq \eta (t_2 - t_1)$ for some fixed~$\eta>0$.
This condition is physically justified whenever the generator does not change significantly in short time scales. 
In this way we translate our problem into finding separately for every time interval between two snapshots $M_{p-1}$ and $M_p$ a time-independent Lindbladian $\set{L_p}_{p=1}^N$ whose evolution $\exp L_p$ approximates the mapping $T_p$ and at the same time we retain consistency with a reasonably stable Lindblad generator over the whole evolution. In other words, we characterise each snapshot $M_p$ as a sequence of time-independent Markovian evolutions, i.e.,
$M_p' \approx \prod_{p=p'}^1 \e^{L_p}$,
which will converge in the limit of infinitely small time intervals (we provide a quantitative expression for the bound for the error in \cref{app:timeDep}).
We remark that, conveniently, the times $t_p$ when the snapshots $M_p$ are taken is not required input.

\medskip

The algorithm implementing this approach goes in a similar vein as the ones illustrates previously, where at its core is a convex optimisation task.

\begin{algorithm}\small
	\SetKwInOut{Input}{Input}\SetKwInOut{Output}{Output}
	\Input{integer $N$, matrices $M_1,\dots,M_N$, 
		positive integer $m_{\max}$,
		positive real~$\beta$ }
	\Output{set of Lindbladians  $\set{L_p}_p$ generating best fit map to $\Theta_p = M_{p}M_{p-1}^{-1}$}
	
	\For{ $\vec m \in \set{m_{\max}}^{\times d^2}$}{
		
		\For{$p=1,\dots,N$}{
			$\Theta_p \leftarrow M_{p}M_{p-1}^{-1}$ \\ 
			$\set{\ell_j, r_j}_{j=1}^{d^2} \leftarrow$ set of left and right eingenvectors of $\Theta_p$ \\
			$G(p,\vec m) \leftarrow \log \Theta_p + 2\pi \i \sum_{j=1}^{d^2} m_j \, \ketbra{\ell_j}{r_j}$ 
			
			\smallskip
			
			Run convex optimisation programme on variable $X(p,\vec m)$:\\
			\Indp $\begin{array}{ll}
				\text{minimise} & \Fnorm{X(p,\vec m)-G^\Gamma(p,\vec{m})}   \\
				\text{subject to } &X(p,\vec m) \text{ hermitian}\\
				&\omega_\perp X(p,\vec m) \omega_\perp \geq 0 \\
				&\norm{\Tr_1[X(p,\vec m)]}_{1} = 0 \\
				&\Fnorm{X(p,\vec m)-X(p-1, \vec m)}\leq  \beta
			\end{array}$\\
			\Indm  
			Store $X(p,\vec m)$
		}%end snapshots p loop
		
		$\mathrm{distance}(\vec{m}) \leftarrow \sum_p \Fnorm{\Theta_p - \exp X^\Gamma(p,\vec m)}$  
	}%end loop matrix branches m
	\Return $\set{L_{p} = X^\Gamma(p,\vec{m}')}_p \text{ for } \vec{m}'=\mathrm{argmin}\, \{\mathrm{distance}(\vec{m})\}$
	\caption{time-dependent Markovian evolution}\label{alg:time_dependent}
\end{algorithm}

\FloatBarrier

The Lipschitz continuity constraint is here characterised by the parameter $\beta \equiv \beta(\eta)$ -- for the difference between two consecutive generators. (Note that the Frobenius norm is invariant with respect to the $\Gamma$ involution.)
This ensures that even though we are not taking measurements at infinitesimally short time intervals, the resulting channel $\widetilde{M} = \Pi_{p=1}^Te^{L_p}$ is close to the true channel in the case of a time-dependent Markovian evolution.
We make this argument rigorous by calculating error bounds in the resulting channel with respect to the number of snapshots taken (see \cref{sec:bounds}).
When the snapshots are taken at regular intervals $\mc T /N$, we have
\begin{equation}
	\beta(\eta) = \eta \frac{\mc T^2}{N^2} + 4 \mathfrak{R},
\end{equation}
where $ \mathfrak{R}$ is the error for the truncation of the Magnus expansion of $\Phi(t)$ at the first order (cfr. \cref{sec:Lipschitz_bound}). 

In a more practical sense, when using the algorithm one can start with an estimate for $\beta$ and then vary this parameter over multiple runs, observing how the output changes.
Increasing the value of $\beta$ allows for a search over a larger space around the Lindbladian obtained in the previous iteration of the $p$-loop (that is, the generator of the previous time interval), which will likely give a better fit to the most suited branch of~$\log \Theta_p$.
From a physical perspective, if the distance is significantly reduced when augmenting $\beta$, this suggests that the Lindblad generator is varying more rapidly than anticipated.
Conversely, obtaining good fits to the data when setting $\beta$ with small values suggests that the process is close to time-independent dynamics.
Against intuition, relaxing the bound on the difference between two consecutive Lindbladians can instead increase the distance between $M_p$ and its Markovian estimation $\Pi_{j=1}^p L_j$, but not beyond explicitly derived bounds~\cite{CEW09}. The underlying reason for this behaviour is again that there is no direct correspondence between distance of two matrices and their logarithms.

%======================================================================
\section{Reconstructing perturbed degenerate eigenspaces}\label{sec:eigenspaces}
%======================================================================
The analysis in~\cite{Markdynamics08} is based on the assumption that all eigenvalues are non-degenerate: the set of matrices with this characteristic form a dense set\footnote{{cfr.} the \emph{Zariski} topology \cite{DanilovAlgGeo}.} in the matrix space.

Working with an $\NDsquare$ matrix is needed in order to deal with a set of eigenvectors where each of them is unique (up to a scalar factor). This also ensures that the matrix logarithm is unique, up to complex branches. Conversely, if $M$ is not an $\NDsquare$ matrix, there is then a continuous freedom in the choice of eigenbasis, and thus uncountably infinitely many different matrix logarithms (not just the countable infinity of complex branches of the logarithm). This is due to the fact that, in case the matrix is diagonalizable but has an eigenvalue which is not simple, the corresponding degenerate eigenspace allows for infinite number of choices of basis vectors. In the more general case, when $M$ is not diagonalizable, the Jordan canonical form again admits an uncountably infinite number of choices of generalized eigenvectors~\cite{Weintraub}.
In \cref{sec:matrix_perturbation} we show that it is always possible, given a defective or derogatory Choi-hermitian matrix, to produce an arbitrarily close $\NDsquare$ matrix that preserves the hermiticity-preserving property, so that this does not alter the output of our algorithm.

In the reminder of this section we tackle the opposite (and more physically relevant) situation, when we are given an $\NDsquare$ matrix which arises as a \emph{perturbation} of a channel with degenerate subspaces.
This would be the case if, for example, the channel was a perturbation of a quantum gate -- the typical situation when benchmarking quantum computing devices.

%----------------------------------------------------------------
\subsection{Explaining the problem with degenerate subspaces}\label{sec:perturbation_example}

Imagine we are given an $\NDsquare$ matrix $M$ which may come from the perturbation of a quantum embeddable operator having degenerate subspaces. For instance perturbations of many of the standard unitary gates in quantum computation, such as Pauli gates. This is a delicate situation since in the general case the hermiticity-preserving basis vectors structure characterising Lindblad operators, as discussed in \cref{sec:Markovianity}, will be broken under perturbation (even when this is very small) due to the instability of the basis of multi-dimensional eigenspaces. Thus, when looking for the closest Lindbladian, the convex optimisation approach will possibly retrieve a Lindbladian whose matrix exponential is very distant from the original unperturbed operator~$M$.

To illustrate this argument, consider the Pauli $X$-gate and restrict our attention to the hermiticity-preserving condition (noting that the closest hermiticity-preserving matrix will always be closer than the closest Lindbladian since the latter imposes more constraints).
The operator $X$ has a two-fold degenerate eigenvalue~1 with eigenspace $\mathrm{span} \{(1,1,1,1);$ $(1,-1,-1,1)\}$ and another two-dimensional eigenspace $\mathrm{span} \set{(1,0,0,-1);(0,1,-1,0)}$ with respect to eigenvalue~$-1$. We denote these vectors by $w_1,w_2,w_3,w_4$, respectively; observe that all of the eigenvectors are self adjoint ($w_3$ is self-adjoint up to an irrelevant overall phase). 
Write
$E = \varepsilon \big( \dyad{w_1} -\dyad{w_2} + \dyad{w_3}-\dyad{w_4} \big)$.
Then the $\NDsquare$ perturbed operator $X+E$ then has eigenvalues
$1+\varepsilon, 1 - \varepsilon, -1+\varepsilon,-1-\varepsilon$
with respect to the eigenbasis $\set{w_1,w_2,w_3,w_4}$, and its 0-branch logarithm $\log (X+E)$ has eigenvalues $\varepsilon,-\varepsilon, \i \pi-\varepsilon, \i\pi +\varepsilon$ (up to first order in  $\varepsilon$) with respect to the same eigenbasis. At this point, if we look for the closest hermiticity-preserving operator, since all the eigenvectors are self-adjoint we obtain a matrix having again the same eigenbasis and keeping the real part of the eigenvalues of $\log (X+E)$, i.e.,
$\varepsilon,-\varepsilon,-\varepsilon, \varepsilon$. Clearly, the exponential of this matrix is close to the identity map and not the expected Pauli $X$-gate, even for very small~$\varepsilon$. The same will apply for any complex branch of $\log (X+E)$ where we can add $2\pi \i \ \mathrm{mod}\, k$ to any eigenvalue of $\log (X+E)$.

If we instead consider a perturbation of the same magnitude but along hermitian-related vectors of the eigenspace of~$-1$, say,
$E = \varepsilon \big( \dyad{w_1}-\dyad{w_2} +\dyad{w_5}-\dyad{w_6} \big)$ with
$w_5 = w_3 + w_4=(1,1,-1,-1)$ and $w_6 = w_3 - w_4=(1,-1,1,-1)$ so that $w_5^\dagger=w_6$,
then $\log(X+E)$ has again eigenvalues
$\varepsilon,  - \varepsilon, \i \pi +\varepsilon, \i\pi-\varepsilon$
but this time with respect to eigenbasis $\set{w_1,w_2,w_5,w_6}$. In this case, by choosing the branch appropriately (i.e. picking $G_0 - 2\pi \i \dyad{w_6}$) the logarithm has eigenvalues $\varepsilon,-\varepsilon, \i \pi-\varepsilon, -\i\pi +\varepsilon$ and its closest hermiticity-preserving map has eigenvalues $\varepsilon,-\varepsilon, \i \pi, -\i\pi$. As expected, taking the exponential of this matrix will give a map very close to $X$.
This example highlights both the importance of reconstructing a pair of hermitian-related eigenvectors for the eigenvalue~$-1$ as well as choosing the correct complex branches of the matrix logarithm.

Our strategy to overcome this complication is to reconstruct a compatible hermiticity-preserving structure for the invariant subspaces of those eigenvalues of $M$ that are close to each others and that presumably stem from a perturbation of a unique degenerate eigenvalue. Hence we will look for a new basis of eigenvectors that we will interchange with the actual eigenbasis, creating a new operator $R$ on which to run the convex optimisation problem to retrieve the closest Lindbladian to $\log R$.

To provide an insight on our approach, in the next session we discuss the single-qubit case; the analysis of the general case can be found in \cref{app:multi-qubit}.

%---------------------------------------------------------------------------------------
\subsection{Reconstructing the hermiticity preserving structure: one-qubit case}\label{sec:reconstructing_HP_structure}
%---------------------------------------------------------------------------------------

For single-qubit channels we have the following possibilities: (a)~one pair of close eigenvalues, (b)~three close eigenvalues, (c)~two different pairs, or (d)~all four eigenvalues are close.

Consider case~(a) with a pair of eigenvalues that is close to the real negative axis and where $w_1$ and $w_2$ are the corresponding eigenvectors. Assume that they come from a real negative 2-fold degenerate eigenvalue. We seek a new pair of hermitian-related eigenvectors $\setn{v,v^\dagger}$ such that $\mathrm{span}\setn{v,v^\dagger} = \mathrm{span}\setn{w_1,w_2}$.\newline
Thus we want to find coefficients $\alpha,\beta,\mu,\nu$ such that $\ket{v} = \alpha \ket{w_1} + \beta \ket{w_2}$ and $\ket{v^\dagger} = \alpha^\ast  \ket{w_1^\dagger} + \beta^\ast  \ket{w_2^\dagger} = \mu \ket{w_1} + \nu \ket{w_2}$.
The solution will parametrise a set of compatible hermitian-related eigenvectors that we will interchange with the vectors $w_1$ and $w_2$.

If the two close eigenvalues are near the positive real axis, then if we assume they come from a real eigenvalue we have an additional option: a pair of two self-adjoint eigenvectors $\set{v_1,v_2}$ spanning the eigenspace of $w_1$ and $w_2$. In other words, we look for coefficients $\alpha,\beta,\mu,\nu$ such that $\ket{v_1} =\alpha \ket{w_1} + \beta \ket{w_2}$ and $\ket{v_2} =\mu \ket{w_1} + \nu \ket{w_2}$ with $\ket{v_1} = \ket{v_1^\dagger}$ and $\ket{v_2} = \ket{v_2^\dagger}$. Again, we should implement any possible solution $\set{v_1,v_2}$ as a new basis of eigenvectors related to the pair of eigenvalues.

The third option for case (a) is a pair of complex eigenvalues not close to the real axis. If this was originally a unique, two-fold degenerate complex eigenvalue $\lambda$, the hermiticity-preserving condition implies a second two-dimensional eigenspace with respect to eigenvalue $\lambda^\ast$; this will then be case~(c).

Now consider case~(b), where three eigenvalues of $M$ are all close. In order to represent the perturbation of a quantum embeddable channel, they cannot originate from a 3-fold degenerate complex eigenvalue, since it is not possible to pair three eigenvalues each with a complex partner in a 4-dimensional space. They cannot originate from a real negative eigenvalue either, since the same argument will apply for the logarithm of $M$, which must also be hermiticity-preserving. $M$ may instead be compatible with a Markovian dynamics if the eigenvalues are close to the real positive axis. Denoting by $w_1,w_2$ and $w_3$ the corresponding eigenvectors, we want to substitute them with a set $\setn{v,v^\dagger,z}$ with $z=z^\dagger$ and $\mathrm{span}\setn{v,v^\dagger, z} =  \mathrm{span}\setn{w_1,w_2,w_3}$ reconstructing an original unperturbed 3-dimensional eigenspace. Alternatively, we should find a new eigenbasis of three self-adjoint vectors.

In case (c) we say that we have two pairs of different eigenvalues. As we discussed in case (a), if one of this pair stems from a 2-fold degenerate and complex eigenvalue $\lambda$, then by the hermiticity-preserving condition the other pair should be close to $\lambda^\ast$.
We should thus find a basis $v_1,v_2$ for the eigenspace of $\lambda$ and $v_3,v_4$ for the eigenspace of $\lambda^\ast$ such that $v_3=v_1^\dagger$ and $v_4=v_2^\dagger$.
Now consider pairs close to the real axis. If both can be associated with a real negative eigenvalue, then the basis of each 2-dimensional subspace corresponding to one pair of eigenvalues should be chosen to be hermitian-related vectors as in case~(a). If one pair is close to a real positive number and the other to a negative one, than the eigenspace of the positive eigenvalue can be spanned either by an hermitian-related pair of eigenvectors or two self-adjoint vectors. The last option in case (c) is that both pairs comes from two different 2-fold degenerate real positive eigenvalues. In this case, again each eigenspace can be spanned by a pair of hermitian-related vectors or two self-adjoint vectors.

If all four eigenvalues are close (case~(d)) and we presume that they come from a single 4-dimensional eigenspace, then $M$ must necessarily be the (perturbed) identity channel up to a real scalar factor. We check if this is close enough according to the error tolerance parameter.

In \cref{tab:1-qubit_deg} we summarize the above described scenarios for the unperturbed operator.
The general multi-qubit case is discussed in \cref{app:multi-qubit}.

\begin{table}[htbp]
	\footnotesize
	\begin{center}\renewcommand{\arraystretch}{2.5}
		\begin{tabular}{|C{3cm}|C{3.7cm} C{3.7cm}|}
			\hline
			\textbf{one single 2-dim degeneracy} & (i) positive eigenvalue with either h.r. or s.a. basis vectors & (ii) negative eigenvalue with h.r. basis vectors \phantom{aaaaaaaaaaaa} \\
			\hline
			\textbf{one single 3-dim degeneracy} & \multicolumn{2}{C{8.3cm}|}{positive eigenvalue with either 1 h.r. and 1 s.a. basis vectors, or 3 s.a. basis vectors } \\
			\hline
			\multirow{2}{3cm}{\vspace{1cm} \\ \centering\textbf{two distinct 2-dim degeneracies}} & (i) two positive eigenvalues each with  1 h.r. pair or 2 s.a. basis vectors & (ii)
			two negative eigenvalues each with h.r. basis eigenvectors \phantom{aaaaaaaaaaaa}\\[-8pt]
			& (iii) a positive eigenvalue with either h.r. or s.a. basis vectors and a negative eigenvalue with  h.r. basis vectors & (iv) a pair of complex conjugate eigenvalues $\lambda$ and $\lambda^\ast$ with   two  h.r. partner  vectors in the partner subspaces \phantom{aaaaaaaaaaaa} \\
			\hline
			\textbf{one single 4-dim degeneracy} &  \multicolumn{2}{C{8.3cm}|}{a single 4-degenerate real eigenvalue \break (identity channel)}\\
			\hline
		\end{tabular}
	\end{center}
	\caption{Structure of a multi-dimensional eigenspace for an hermiticity-preserving operator on one qubit. Here we abbreviate ``hermitian related'' by h.r. and ``self-adjoint'' by s.a.}\label{tab:1-qubit_deg}
\end{table}

\FloatBarrier

%-------------------------------------------------------------------------------------------
\subsection{Pre-processing algorithms for a single snapshot} \label{sec:pre-processing}
%-------------------------------------------------------------------------------------------

We have implement the strategy of retrieving a set of hermiticity-preserving eigenvectors spanning the subspaces of the clusters as a pre-processing algorithm to our complex optimisation programme.
The first task of this algorithm is to detect sets of eigenvalues which are close together -- with closeness parametrised by a precision parameter~$p$, given as input -- likely stemming from perturbations of degenerate eigenvalues.
(If there are no degenerate eigenvalues, then the algorithm proceeds straightforwardly to convex optimisation stage to find $L$ or $\mu_{\min}$.)
Once these clusters of eigenvalues have been identified, the successive step of the pre-processing procedure is to construct bases for the degenerate eigenspaces which have a compatible hermiticity-preserving structure.

Note that there is an infinite number of possible basis choices which respect the hermiticity-preserving structure, and which basis is chosen will affect the distance to the closest Lindbladian.
To handle this, we use a randomised construction to generate $r$ bases satisfying the hermiticity-preserving property, and we run the convex optimisation algorithm on each one, keeping the optimal result.
Clearly, running the algorithm with higher values for~$r$ will give a better fit~(cfr.~\cref{sec:examples}).~

\medskip

The implementation of pre-processing procedure is rather technical and we refer the Reader to \cref{app:pre-processing} -- and in particular to \cref{alg:construct_basis_complex}, \cref{alg:construct_basis_real} and \cref{alg:construct_basis_random} -- for full details.

%------------------------------------------------------------------
\subsection{Stability of hermiticity preserving subspaces}\label{sec:stability_theorem}
%------------------------------------------------------------------

In principle, one critical occurrence in our pre-processing approach is that, given a tomographic snapshot $M$ with two (or more) approximately degenerate eigenvalues, it may be the case that no compatible hermiticity-preserving basis can be found, even if $M$ is the result of the perturbation of an hermiticity-preserving channel.
We resolve this question in a positive way, namely, using tools from matrix perturbation theory we show that it is actually always possible to find a basis for $M$ which satisfies the hermiticity-preserving conditions, at least approximately\footnote{We note that for the numerical examples in this paper this situation did not arise, and as such our algorithms only deal with the case where the hermiticity preserving basis is exact.}. That is,

\begin{theorem}[informal version of \cref{thm:stability_h-r-_structure}]
	Let $\Lambda$ be an hermiticity-preserving map and $M= \Lambda + E$ its perturbed version. Assume $\lambda, \lambda^\ast$ are a pair of degenerate eigenvalues of $\Lambda$, giving rise under perturbation $E$ to two clusters of eigenvalues for $M$. 
	Let $\set{v_j}_j$ and $\{v_j'\}_j$ be the eigenvectors of $M$ to the clusters of $\lambda$ and $\lambda^\ast$, respectively.
	Then there exists a set of paired vectors $\{(w_j,w_j')\}_j$ such that $w_j^\dagger \approx w_j'$, and such that $\mathrm{span}\set{w_j}_j=\mathrm{span}\set{v_j}_j$ and $\mathrm{span}\{w_j'\}_j=\mathrm{span}\{v_j'\}_j$.
\end{theorem}
This theorem constitutes an analytical founding for our search over a hermiticity-preserving sets of vectors in the pre-processing phase.
We refer the Reader to \cref{sec:perturbation_theory} for a theoretical background in matrix perturbation theory, followed by the formal version of the theorem (including the case where $\lambda$ is a real eigenvalue) and the proof thereof.

\FloatBarrier

%======================================================
\section{The main algorithm} \label{sec:main_alg}
%======================================================

Our main algorithm provides a recipe for evaluating non-Markovianity from a single-snapshot. (The multiple-snapshot case building on \cref{alg:MR:multiple_snapshots}  and \cref{alg:time_dependent} is a straightforward modification.)
It takes as input an estimate of the channel, $M$, a precision parameter, $p$, an accuracy parameter $\varepsilon$, and an integer, $r$, which determines how many random basis choices will be tested in the case of degenerate eigenvalues.
The first step in the algorithm is to determine whether the input matrix has eigenvalues that may originate from a perturbation of degenerate eigenspaces.
If the matrix has no such eigenvalues, we run the convex optimisation algorithm directly on the input.
Conversely, if the matrix has eigenvalues that may originate from perturbed degenerate eigenspaces, before running the convex optimisation we need to construct bases for the degenerate eigenspaces which have a compatible hermiticity-preserving structure, as discussed in \cref{sec:pre-processing}.

Once these matrices have been constructed we run the convex optimisation \cref{alg:MR:Lindbladian} on them to determine whether or not there exists a memoryless channel in the $\varepsilon$-neighbourhood of $M$.
If a Markovian map, $T$, exists within the $\varepsilon$-ball, the algorithm returns the Lindbladian $L$ satisfying $T=e^{L}$. 
If no Markovian channel is retrieved, the main algorithm calls \cref{alg:MR:NoMark_parameter} which calculates the non-Markovianity parameter~$\mu_{\min}$  .

\medskip

A simplified version of the algorithm is presented here as \cref{alg:MAIN}. The complete pseudo-algorithm can be instead found in \cref{app:main}.

\begin{algorithm}\small
	\SetKwInOut{Input}{Input}\SetKwInOut{Output}{Output}
	\SetKw{Continue}{continue}
	\Input{$\NDsquare$ matrix $M$, integer $\mathrm{random}\_\mathrm{samples}$, positive real numbers~$p, \varepsilon$}
	\Output{Lindbladian, $L$, consistent with $M$ (within error tolerance $\varepsilon$), or if no such $L$ exists, non-Markovianity parameter~$\mu$}
	
	\uIf{$M$ \emph{has no degenerate eigenvalues within precision} $p$}
	{
		$L \leftarrow$ Run \cref{alg:MR:Lindbladian} on $M,M,\varepsilon$
		
		\uIf{ $\|M - \exp(L) \| < \varepsilon$}
		
		{\Return{$L$}}
		
		\Else{
			$\mu \leftarrow$ Run \cref{alg:MR:NoMark_parameter} on input $M, M, \varepsilon$
			
			\Return{$\mu$}
	}}
	\Else{
		
		\For{$i \in (0,\mathrm{random}\_\mathrm{samples})$}{
			
			$ R \leftarrow $ hermiticity-preserving matrix arising from $M$ \\
			$L_R \leftarrow$ Run \cref{alg:MR:Lindbladian} on $M,R,\varepsilon$ \\
			\uIf{$\|M - \exp(L_R) \| < \varepsilon$}
			{\Continue}
			
			\Else{
				$\mu_R \leftarrow$ Run \cref{alg:MR:NoMark_parameter} on input $M, R, \varepsilon$}	
			
		}%End  FOR
		\uIf{$\min{\|M - \exp(L_R') \| < \varepsilon}$}
		{\Return{$L_{R}$ \emph{with} $\min \|M - \exp(L_R) \|$}}
		\Else
		{\Return{$\min \mu_R$}}
	}%End main ELSE
	\caption{Main algorithm, including pre-processing}\label{alg:MAIN}
\end{algorithm}

\FloatBarrier

%======================================================================
\section{Numerical Examples with Cirq}\label{sec:examples}
%======================================================================

In this section, we present the results of testing our algorithm numerically on noisy dynamics synthesised in Cirq~\cite{cirq}. The numerics serve as a benchmark of both our convex optimisation and pre-processing algorithms.
Since we know the `ideal' channel in each test case, we can compare the outcomes against the true values. The algorithms performed well in all cases.
It is worth emphasising that the algorithm itself does not require any information about what the `ideal' channel is -- it merely needs the tomographic data of the channel under consideration. Here the ideal operator is only used to benchmark the results against.

The tomography was carried out using code available in \cite{gitlab_repo}; linear inversion process tomography was performed in the standard basis 
(we chose to carry out linear inversion process tomography as it is unbiased).
The drawback of linear inversion process tomography is that it can lead to process matrices that are unphysical. 
However, since our algorithms only search over physical channels to find the best Markovian fit we do not view this as a serious drawback.
Moreover, the algorithms themselves are not sensitive to what sort of process tomography was used, so it is straightforward to run the algorithms with a different choice of settings for the process tomography to investigate whether this affects the results at all. 
We leave this to future work.

\subsection{One-qubit numerics}

In every one-qubit example each measurement in the simulated process tomography was repeated 10,000 times. Throughout this section we will express distances between matrices using the Frobenius norm.

\subsubsection{Unitary 1-qubit example: $X$-gate}
In \cref{sec:perturbation_example} we demonstrate how the naive algorithm (with no pre-processing) is not guaranteed to find the closest Markovian channel if the input is a `noisy' $X$-gate, as the degenerate eigenbases are not stable with respect to perturbations.\footnote{The transfer matrix for an $X$-gate has two two-fold degenerate eigenspaces: one with eigenvalue $+1$, and one with eigenvalue $-1$.} This numerical test shows that our pre-processing scheme solves the problem.

We used Cirq's density matrix simulator to simulate process tomography on a 1-qubit $X$-gate.
We then applied our convex optimisation algorithm to extract the full description of the best-fit Lindbladian; we denote the generated Markovian channel by $T_X$.
In \cref{fig:X-gate} we show how the process fidelity varies with the number of random samples we set the code to run for.
The optimum fidelity between the tomographic snapshot and $T_X$ was $100.92\%$ -- while a fidelity of over $100\%$ is clearly not physical, small `overfidelities' due to numerical errors are not uncommon, and the result shows that the optimum fidelity was close to $100\%$.
As expected, increasing the number of random samples increases the process fidelity between $M_X$ and $T_X$, although there is some fluctuation due to the randomised nature of the algorithm.
This shows that the `direction' which the matrix is perturbed in is crucial -- as we saw in \cref{sec:perturbation_example}, perturbations of the same magnitude along different directions can have very different effects on whether a compatible Lindbladian is found.

\begin{figure}[htbp]
    \centering
    \begin{subfigure}[t]{0.45\textwidth}
        \centering
        \includegraphics[width=\linewidth]{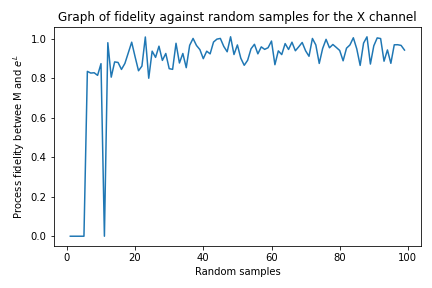}
        \caption{Random samples in the range $(0,100)$} \label{fig:X100}
    \end{subfigure}
    \hfill
    \begin{subfigure}[t]{0.45\textwidth}
        \centering
        \includegraphics[width=\linewidth]{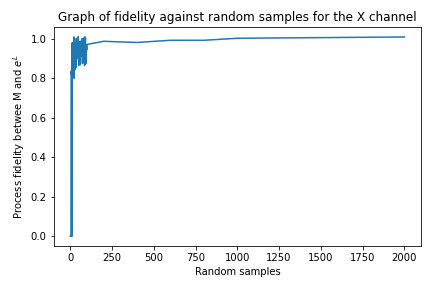}
        \caption{Random samples in the range $(0,2000)$} \label{fig:X1000}
    \end{subfigure}
\caption{Results for a simulated 1-qubit $X$-gate. In both cases the algorithm ran with $\varepsilon = 1$. A fidelity of zero indicates that no Markovian channel was found in that run.}\label{fig:X-gate}
 \end{figure}

\subsubsection{Markovian 1-qubit example: depolarizing channel}

The depolarizing channel, implementing the evolution
\begin{equation}
\rho \rightarrow (1-p)\rho + \frac{p}{3}\left(X\rho X + Y\rho Y + Z\rho Z \right) ,
\end{equation}
is an example of a non-unitary, but Markovian, quantum channel.

We simulated process tomography on a depolarizing channel with $p=0.3$.
The transfer matrix for this channel has a non-degenerate $+1$ eigenvalue, and a three-fold degenerate eigenspace with eigenvalue $0.6$.
We used our convex optimisation algorithm to construct the closest Markovian channel.

In \cref{fig:depolarizing} we show how the process fidelity of the tomography result and the closest Markovian channel varies with the number of random samples we allowed the code to run for.
The optimum fidelity found between the Markovian channel found by the algorithm and the tomographic snapshot was $99.96\%$.

Comparing \cref{fig:depolarizing} with \cref{fig:X-gate} (and noticing the different $y$-axis scales in the two graphs) we see that unlike with the $X$-gate, \cref{alg:MAIN} always finds a good approximation to the tomographic result, even for very few random samples.
This suggests that the direction of perturbation is less important for the depolarizing channel than for the $X$-gate.

\begin{figure}[htbp]
    \centering
    \begin{subfigure}[t]{0.45\textwidth}
        \centering
        \includegraphics[width=\linewidth]{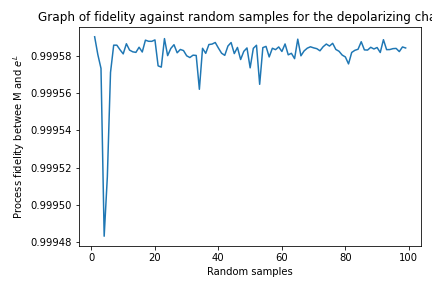}
        \caption{Random samples in the range $(0,100)$} \label{fig:D100}
    \end{subfigure}
    \hfill
    \begin{subfigure}[t]{0.45\textwidth}
        \centering
        \includegraphics[width=\linewidth]{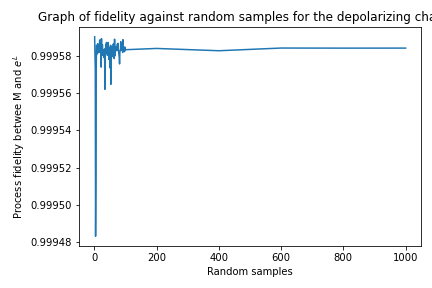}
        \caption{Random samples in the range $(0,1000)$} \label{fig:D1000}
    \end{subfigure}
\caption{Results for a simulated 1-qubit depolarizing channel.}\label{fig:depolarizing}
 \end{figure}

\subsubsection{Non-Markovian and Markovian 1-qubit examples: unital quantum channel}\label{sec:Unital_quantum_channel}

In \cite{Andersson_2007} Kraus operators are constructed for a unital quantum channel, and conditions for the channel to be Markovian have been derived.\footnote{A channel is said to be unital if the maximally mixed state is a fixed point of the evolution.}
For a master equation
\begin{equation}
\dot{\rho}(t) = \gamma_1\sigma_1 \rho \sigma_1 +  \gamma_2\sigma_2 \rho \sigma_2+  \gamma_3\sigma_3 \rho \sigma_3 -  (\gamma_1 + \gamma_2 + \gamma_3)\rho ,
\end{equation}
the Kraus operators for the evolution are:
\begin{equation}
\begin{split}
& A_0 = \frac{1}{2}\left(1 + \Gamma_1 + \Gamma_2 +\Gamma_3 \right)^\frac{1}{2} \mathbb{I} \\
&A_1 = \frac{1}{2}\left(1 + \Gamma_1 - \Gamma_2 -\Gamma_3 \right)^\frac{1}{2} \sigma_1 \\
&A_2 = \frac{1}{2}\left(1 - \Gamma_1 + \Gamma_2 -\Gamma_3 \right)^\frac{1}{2} \sigma_2 \\
&A_3 = \frac{1}{2}\left(1 - \Gamma_1 - \Gamma_2 +\Gamma_3 \right)^\frac{1}{2} \sigma_3,
\end{split}
\end{equation}
where
\begin{equation}
\Gamma_i \coloneqq e^{- \int_0^t ds\left( \gamma_j(s) + \gamma_k(s)\right)}
\end{equation}
and $\{i,j,k\}$ is a permutation of $\{1,2,3\}$.

The inequality for the channel to be completely positive is
\begin{equation}
\Gamma_i + \Gamma_j \leq 1 + \Gamma_k,
\end{equation}
where $\{i,j,k\}$ is a permutation of $\{1,2,3\}$.
The condition for the channel to be Markovian is given by $\gamma_i(t) > 0$, for all~$t$ and~$i$.

We simulated process tomography in Cirq on a non-Markovian unital quantum channel with $\gamma_1 = -200$, $\gamma_2 = 201$, $\gamma_3 = 200.5$.
The transfer matrix for this channel is non-degenerate, and the process fidelity between the ideal transfer matrix and the tomographic snapshot we obtained is $99.96\%$.
Applying \cref{alg:MAIN} with a large error tolerance parameter ($\varepsilon = 1$) found a Markovian channel, $T$, which had process fidelity of $82.18\%$  with the tomographic snapshot.

In \cref{fig:UQC} we show the results of comparing the non-Markovianity parameter $\mu$ against the accuracy~$\varepsilon$.
We first ran \cref{alg:MAIN} with varying~$\varepsilon$ using a step size $\delta_{\mathrm{step}}=0.5$ in the \cref{alg:MR:NoMark_parameter} subroutine:
the results are shown in  \cref{fig:uqc_big_steps}.
For $\varepsilon \in (0,0.186)$ no  $\mu$-parameter was found by the algorithm.
Conversely in \cref{fig:uqc1}, where the same analysis is run with a step size $\delta_{\mathrm{step}}=0.01$, a non-Markovianity parameter was found for all $\varepsilon \geq 0.025$.

\begin{figure}[htbp]
	\centering
	\begin{subfigure}[t]{0.45\textwidth}
		\centering
		\includegraphics[width=\linewidth]{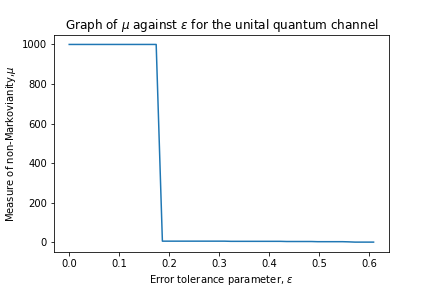}
		\caption{$\varepsilon$ in the range $(0,0.609)$ with step size of $0.5$} \label{fig:uqc_big_steps}
	\end{subfigure}
	\hfill
	\begin{subfigure}[t]{0.45\textwidth}
		\centering
		\includegraphics[width=\linewidth]{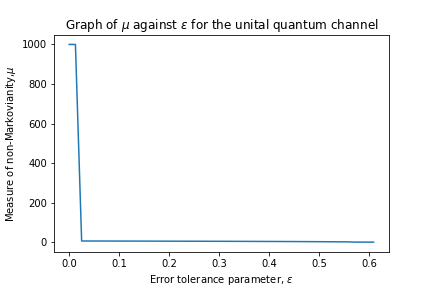}
		\caption{$\varepsilon$ in the range $(0,0.609)$ with step size of 0.01} \label{fig:uqc1}
	\end{subfigure}

	\begin{subfigure}[t]{0.45\textwidth}
		\centering
		\includegraphics[width=\linewidth]{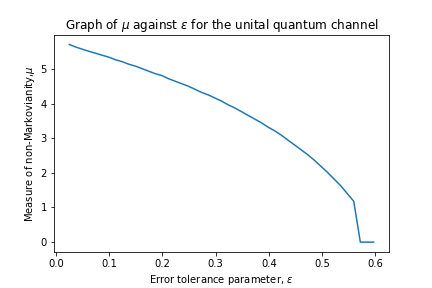}
	\caption{$\varepsilon$ in the range $(0.025,0.609)$} \label{fig:uqc2}
	\end{subfigure}
	\begin{subfigure}[t]{0.45\textwidth}
		\centering
		\includegraphics[width=\linewidth]{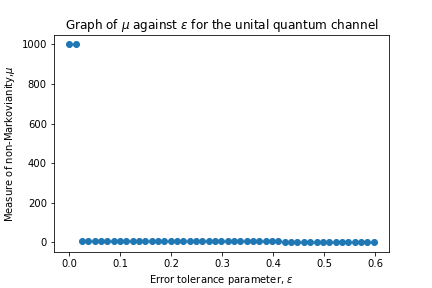}
	\caption{\cref{fig:uqc1} as a scatter plot. The diagonal line from the case where no $\mu$ was found to the case where $\mu$ was found in \cref{fig:uqc1} is simply a result of the distance between data points on the $x$-axis, as demonstrated by the scatter plot.}
	\end{subfigure}
	\caption{Results for a simulated unital quantum channel. A value of $\mu = 1000$ indicates that no $\mu$ was found by the algorithm. In \cref{fig:uqc_big_steps} the step size is 0.5. In all other figures the step size is 0.01.}\label{fig:UQC}
\end{figure}

This can be understood by recalling that \cref{alg:MR:NoMark_parameter} does not search for the channel $T$ that minimises $\mu$ and is within an $\varepsilon$-ball of $M$.
It searches for the Lindbladian $L$ that minimises $\mu$ and is within a $\delta$-ball of $\log(M)$, then checks whether $\|e^{L} - M \| \leq \varepsilon$.
Since the bounds on $\delta$ are not tight, it repeats this process for $\delta \in \{\delta_i = \delta_{\min} + \delta_{\mathrm{step}} \mid \delta_i < \delta_{\max}  \}$, and keeps the best result.
If $\delta_{\mathrm{step}}$ is too large, this can lead to the algorithm failing to select a $\mu$-value within the $\varepsilon$-ball of $M$, even when one exists: this is what has happened in \cref{fig:uqc_big_steps}. The problem can be eliminated by running \cref{alg:MR:NoMark_parameter} with a smaller step size.
An illustration explaining this intricacy is given in \cref{fig:mu_illustration}.

\begin{figure}[htbp]
	\centering
	\scalebox{.81}{
	\begin{tikzpicture}[use Hobby shortcut,closed=true]

		\definecolor{modgray}{RGB}{190,190,190}

		\filldraw[fill=modgray, thick] (0,0) circle (2.5);
		\draw[<->](-0.05,0)--(-2.5,0);
		\node at (-1.25,-0.25) {\large $\varepsilon$};

		\filldraw[red](0,0) circle (2pt);
		\filldraw[blue](1.2,.75) circle (2pt);
		\filldraw[blue](.3,-2.7) circle (2pt);

		\node at (0.17,0.25) {$M$};
		\node at (0.17+1.2,0.25+.75) {$T_1$};
		\node at (.65,-2.7) {$T_2$};

		%--------------------------------
		\draw [->,thick] (3,1) to [out=30,in=150] (5.5,1);
		\node at (4.25,.9) {\large $\log$};
		%---------second figure----------

		\begin{scope}[shift={(8.5,0)}]
			\filldraw[color=modgray] (0,-1.8)..(-.6,-1.1)..(-1.9,-0.7)..(-1.6,0)..(-1.4,.5)..(-1.1,0.9)..(-1.5,1.4)..(-.3,2.2)
			..(0,2.45) ..(.2,1.8)..(.8,1.5)..(.9,1.3)..(1.7,.7)..(1.9,0.2)
			..(1.3,-.3)..(1.8,-.7)..(1.7,-1.4)..(1.3,-1.5);

			\draw[thick] (0,0) circle (2.5); %big circle
			\draw[thick] (0,0) circle (1); %small circle

			\filldraw[red](0,0) circle (2pt);
			\filldraw[blue](1.2,0.5) circle (2pt);
			\filldraw[blue](0.3,-2.1) circle (2pt);

			\node[label={\small $\log(M)$}] at (0.4,-0.12) {};
			\node[label={\small $\log(T_1)$}] at (1.7,0.4) {};
			\node[label={\small $\log(T_2)$}] at (0.75,-2.2) {};

			\draw[<->](0.05,0)--(0.936,-0.35);
			\draw[<->](-0.05,0)--(-2.5,0);
			\node at (0.5,-0.4) {$\delta_{\min}$};
			\node at (-1.6,-0.2) {$\delta_{\max}$};
		\end{scope}
	\end{tikzpicture}
	}
	\caption{A schematic illustration of the distortion of the $\varepsilon$-ball (gray area) under the matrix logarithm, whose boundaries are included between two balls of radius $\delta_{\min}$ and $\delta_{\max}$. Let $T_i$ (for $i=1,2$) be quantum channels with non Markovianity-parameter $\mu_i$ such that $\|\log(T_i) - \log(M) \| = \delta_i$  and $\|T_i - M \| = \varepsilon_i$ where $\mu_1 > \mu_2$, $\varepsilon_2 > \varepsilon > \varepsilon_1$ and $\delta_2>\delta_1$. If \cref{alg:MR:NoMark_parameter} is run on the channel $M$ with $\delta_{\mathrm{step}} > \delta_2 - \delta_1$ the algorithm will not return a non-Markovianity parameter because the convex optimisation finds $T_2$, but this is rejected by the step which checks if $\|e^{L} - M \| \leq \varepsilon$. Running \cref{alg:MR:NoMark_parameter} with a smaller $\delta_{\mathrm{step}}$ solves this issue.} \label{fig:mu_illustration}
\end{figure}

In \cref{fig:uqc1,fig:uqc2} we show the results of running  \cref{alg:MAIN} with a varying $\varepsilon$ parameter, and where the \cref{alg:MR:NoMark_parameter} subroutine has a step size $\delta_{\mathrm{step}}=0.01$.
In this example, for $\varepsilon = 0$ no $\mu$ is found by the convex optimisation algorithm, indicating that there is no compatible Markovian channel in the $\varepsilon$-ball around $M$, even with white noise addition.
This occurs for certain channels because the addition of white noise cannot affect the hermiticity-preserving or trace-preserving condition for Markovianity.
So if the input tomographic matrix does not satisfy these conditions then no amount of white noise addition will render the channel Markovian.
However, for non-zero $\varepsilon$ there may be channels that are hermiticity-preserving and trace-preserving within the $\varepsilon$ ball.
For these values of $\varepsilon$ it is possible to find values of the non-Markovianity parameter $\mu_{\min}$.

In order to benchmark \cref{alg:MR:NoMark_parameter}, we have derived by hand the value of~$\mu$ for this example. 
Following our analytical calculation, we retrieve a non-Markovianity value $\mu=5.76983$ with error tolerance $\varepsilon = 0.01788$.  
As a comparison, by setting $\varepsilon = 0.01788$ in our algorithm (with $\delta_{\mathrm{step}}=10^{-5}$) this returns $\mu_{\min}=5.76539819$, and this is indeed the smallest value for $\varepsilon$ where we retrieve a valid measure $\mu_{\min}$. 
We refer the Reader to \cref{sec:mu_by_hand} for full details.

%-----------------------------------------------------------------------------------------
\subsection{Two-qubit examples}
%-----------------------------------------------------------------------------------------

In every two-qubit example each measurement in the simulated process tomography was repeated 100,000 times.

\subsubsection{Unitary 2-qubit example with degenerate eigenvalues:\\ ISWAP gate}

The action of the ISWAP gate is to swap two qubits, and introduce a phase of $i$ to the $\ket{01}$ and $\ket{10}$ amplitudes.
It has a six-fold degenerate eigenspace with eigenvalue $+1$, a two-fold degenerate eigenspace with eigenvalue $-1$, and two four-fold degenerate eigenspaces with eigenvalues $\pm i$.

We used Cirq's density matrix simulator to simulate process tomography on the ISWAP-gate.
We then used our convex optimisation algorithm to construct the closest Markovian channel, $T_I$.
We ran \cref{alg:MAIN} for 10 random samples, 642 times.
The results are shown in  \cref{fig:iswap}.
The optimum fidelity found between the Markovian channel found by the channel and the tomographic snapshot was $95.69\%$.

It is clear that a higher number of random samples is required than in the 1-qubit case with the $X$-gate.
This is to be expected, since the degenerate eigenspaces are higher-dimensional, and there are more of them.
The probability of randomly choosing a `good' basis is hence lower.

\begin{figure}[htbp]
	\centering
	\begin{subfigure}[t]{0.45\textwidth}
		\centering
		\includegraphics[width=\linewidth]{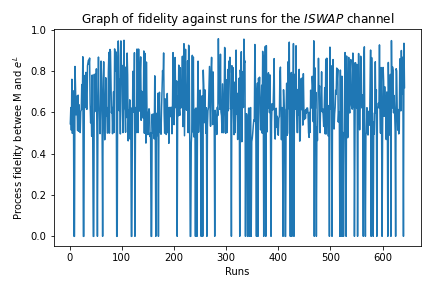}
		\caption{The process fidelity between the tomography result and the Markovian channel constructed by \cref{alg:MAIN} in each 10-random sample run.} \label{fig:iswap_full}
	\end{subfigure}
	\hfill
	\begin{subfigure}[t]{0.45\textwidth}
		\centering
		\includegraphics[width=\linewidth]{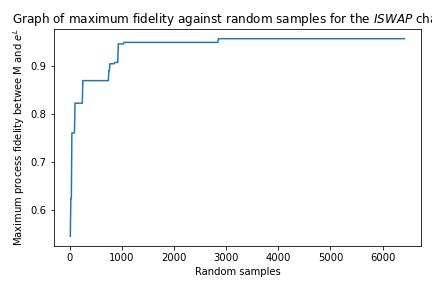}
		\caption{Results from the same simulation as \cref{fig:iswap_full}, where now we show \emph{maximum} process fidelity achieved against total number of random samples.} \label{fig:iswap_best}
	\end{subfigure}
	\caption{Numerics for the simulated ISWAP gate. \cref{alg:MAIN} was run for 10 random samples 642 times.} \label{fig:iswap}
\end{figure}

\subsubsection{Markovian 2-qubit example: depolarizing C$Z$ channel}
The C$Z$ is a two-qubit quantum gate.
Its action is to apply a $Z$-gate to the second qubit if the first qubit is in the $\ket{1}$ state.
Otherwise it acts as the identity on both qubits.

We simulated process tomography on a depolarizing C$Z$ channel, which applied the C$Z$-gate with probability 0.1, or an $XX$ / $YY$ / $ZZ$ -gate with probabilities 0.07, 0.08 and 0.09 respectively.
This is an example of a non-unitary, but Markovian, quantum channel.
The transfer matrix of the channel has six two-fold degenerate eigenspace with eigenvalue $+1$, $0.7$, $0.93$, $0.57$, $0.67$, and $0.47$, and non-degenerate eigenvalues of $0.68$, $0.66$, $0.48$ and $0.46$.

We ran \cref{alg:MAIN} on $M_{\textrm{C$Z$depol}}$ for 10 random samples, 503 times.
The results are shown in  \cref{fig:C$Z$D}.
The optimum fidelity found between the Markovian channel found by the channel and the tomographic snapshot was $94.75\%$.

As in the $ISWAP$ case, it is clear that a significantly higher number of random samples are needed in the two qubit case. 

\begin{figure}[htbp]
	\centering
	\begin{subfigure}[t]{0.45\textwidth}
		\centering
		\includegraphics[width=\linewidth]{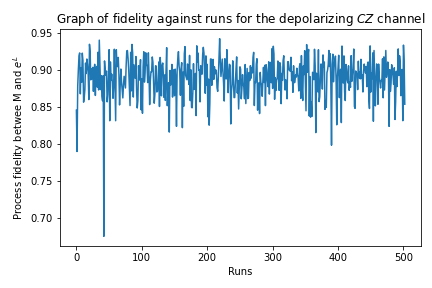}
		\caption{The process fidelity between the tomography result and the Markovian channel constructed by \cref{alg:MAIN} in each 10-random sample run.} \label{fig:C$Z$D_full}
	\end{subfigure}
	\hfill
	\begin{subfigure}[t]{0.45\textwidth}
		\centering
		\includegraphics[width=\linewidth]{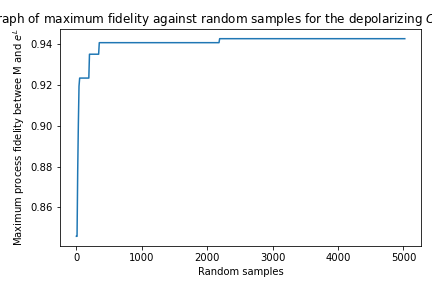}
		\caption{Results from the same simulation as \cref{fig:C$Z$D_full}, where now we show \emph{maximum} process fidelity achieved against total number of random samples.} \label{fig:C$Z$D_best}
	\end{subfigure}
	\caption{Numerics for the simulated depolarizing C$Z$ channel. \cref{alg:MAIN} was run for 10 random samples 503 times.} \label{fig:C$Z$D}
\end{figure}

\subsection{Time series examples}

One of the key advances of this work over previous papers on the subject is the extension to a time-series algorithm.
We have benchmarked the time-series algorithm on a simulated ideal $T$-gate, as well as noisy implementations of the $T$-gate, in Cirq.
The results for the ideal $T$-gate are shown in \cref{fig:idealT}. 
The time series algorithm finds a time independent Lindbladian which is an excellent fit for all times, as would be expected.
The results for noisy gates are shown in \cref{fig:depolT} for the depolarizing channel and \cref{fig:bitflipT} for the bit-flip channel. 
For each noisy implementation, we ran the algorithm in a `high-noise' and a `low-noise' regime. 
From the figures it can be seen that even in the presence of significant noise, where the measured channel is no longer close to the ideal $T$-gate, our algorithm still finds a time-independent Markovian channel which is a good fit to the overall evolution.

\begin{figure}[htbp]
\centering
\includegraphics[width=0.5\linewidth]{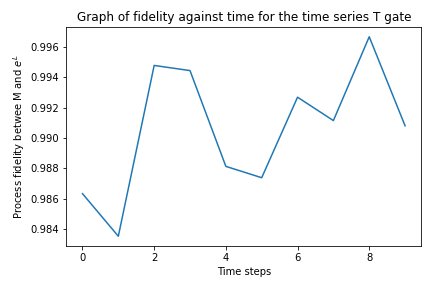}
\caption{The process fidelity between the tomography result and the Markovian channel constructed by \cref{alg:MR:multiple_snapshots} for a $T$-gate, repeated for 10 time steps}\label{fig:idealT}
\end{figure}

\begin{figure}[htbp]
	\centering
	\begin{subfigure}[t]{0.45\textwidth}
		\centering
		\includegraphics[width=\linewidth]{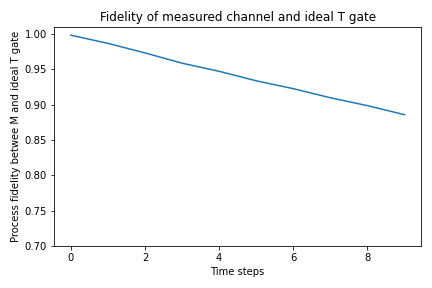}
		\caption{The process fidelity between the ideal T gate and the measured channel with an error rate of $\gamma=0.01$.}
	\end{subfigure}
	\hfill
	\begin{subfigure}[t]{0.45\textwidth}
		\centering
		\includegraphics[width=\linewidth]{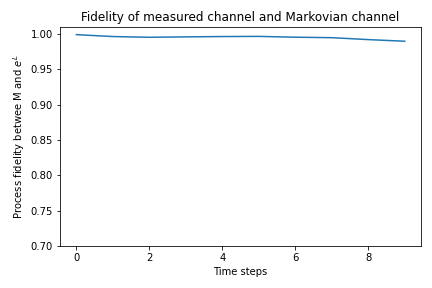}
		\caption{The process fidelity between the measured channel and the Markovian channel found by \cref{alg:MR:multiple_snapshots} with an error rate of $\gamma=0.01$}
	\end{subfigure}
		\hfill
	\begin{subfigure}[t]{0.45\textwidth}
		\centering
		\includegraphics[width=\linewidth]{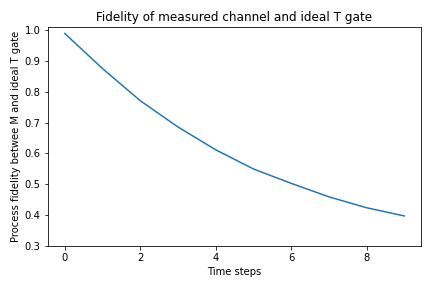}
				
			\caption{The process fidelity between the ideal T gate and the measured channel with an error rate of $\gamma=0.1$}
	\end{subfigure}
			\hfill
	\begin{subfigure}[t]{0.45\textwidth}
		\centering
		\includegraphics[width=\linewidth]{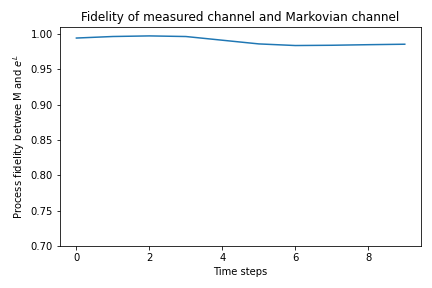}
	\caption{The process fidelity between the measured channel and the Markovian channel found by \cref{alg:MR:multiple_snapshots} with an error rate of $\gamma=0.1$}

	\end{subfigure}
	\caption{Numerics for a $T$-gate with depolarizing noise applied. We explicitly constructed a Markovian channel which was generated by a Lindblad with a Hamiltonian term that generated a T gate, and jump operators that generated a depolarizing channel at various (symmetric) error rates, $\gamma$.} \label{fig:depolT}
\end{figure}

\begin{figure}[htbp]
	\centering
	\begin{subfigure}[t]{0.45\textwidth}
		\centering
		\includegraphics[width=\linewidth]{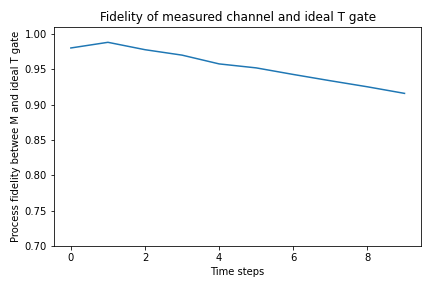}
		\caption{The process fidelity between the ideal T gate and the measured channel with an error rate of $\gamma=0.01$.}
	\end{subfigure}
	\hfill
	\begin{subfigure}[t]{0.45\textwidth}
		\centering
		\includegraphics[width=\linewidth]{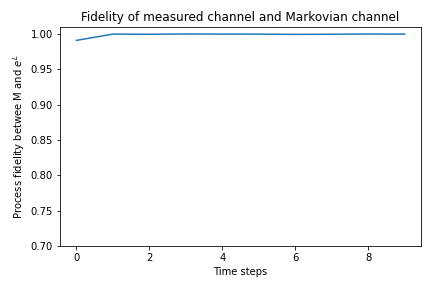}
		\caption{The process fidelity between the measured channel and the Markovian channel found by \cref{alg:MR:multiple_snapshots} with an error rate of $\gamma=0.01$}
	\end{subfigure}
		\hfill
	\begin{subfigure}[t]{0.45\textwidth}
		\centering
		\includegraphics[width=\linewidth]{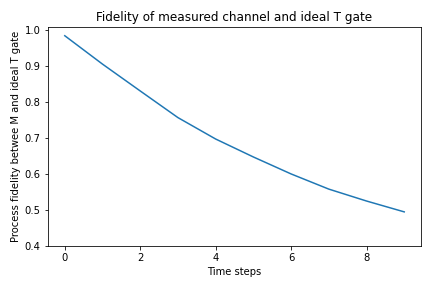}
		\caption{The process fidelity between the ideal T gate and the measured channel with an error rate of $\gamma=0.1$}
	\end{subfigure}
			\hfill
	\begin{subfigure}[t]{0.45\textwidth}
		\centering
		\includegraphics[width=\linewidth]{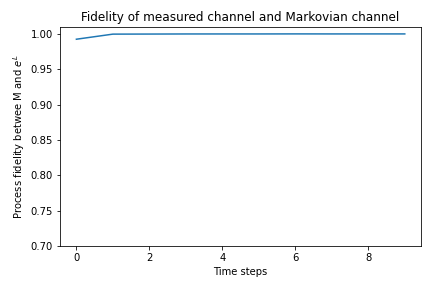}
				\caption{The process fidelity between the measured channel and the Markovian channel found by \cref{alg:MR:multiple_snapshots} with an error rate of $\gamma=0.1$}
	\end{subfigure}
	\caption{Numerics for a $T$-gate with bit flip noise applied. We explicitly constructed a Markovian channel which was generated by a Lindblad with a Hamiltonian term that generated a T gate, and jump operators that generated a bit-flip channel at various error rates $\gamma$.} \label{fig:bitflipT}
\end{figure}

\FloatBarrier

\section{Conclusions}

We have developed novel methods and algorithms, based on previous work~\cite{Markdynamics08}, to retrieve the best-fit Lindbladian to a quantum channel. We have implemented these algorithms in Python, and benchmarked them on synthetic tomography data generated in Cirq.

The key strengths of our method is that it can be applied to a single (or small number of) tomographic snapshot(s), is completely assumption-free regarding the structure of the analysed operator, and does not rely on any prior knowledge of the environment or the noise model. At the core of the method is a convex optimisation programme which searches for the closest Lindblad generator within a given distance from the matrix logarithm of the input. This approach is successful in dealing with imprecise tomographic data, extracting Markovian dynamics within any desired regime of tolerance, and can be used to look for time-independent and time-dependent Markovian dynamics consistent with a series of tomographic snapshots. If no Markovian channel is found, the scheme provides a well-defined quantitative measure of non-Markovianity in terms of the minimal addition of white noise required to ``wash-out'' memory effects, and render the evolution Markovian. 
Moreover, the extensions to the case where we are given as input a series of tomographic snapshots allows us to compute a far more fine-grained test of non-Markovianity, and to tackle the case of time-dependent channels.

A significant part of the work is focused on the treatment of input matrices that are perturbations of some unknown process with a degenerate spectrum. This situation commonly arises when analysing noisy unitary gates in quantum computation. In order to address the susceptibility of the convex optimisation programme with respect to perturbation of multi-dimensional eigenspaces, we have designed a series of pre-processing algorithms, rigorously rooted in the theory of matrix perturbation, which identify and re-construct the unperturbed hermiticity-preserving structure of the original channel.
To test our theoretical formulation, we have implemented and numerically benchmarked our algorithm on simulated data from the Cirq platform, demonstrating that our algorithms are able both to successfully identify channels consistent with an underlying Markovian dynamics, as well detect and quantify non-Markovianity.

One drawback of the algorithms developed here is that they require a full tomographic snapshot.
This only requires 12 measurement settings for analysing the dynamics of a single qubit, and 144 measurement settings for analysing the dynamics of two qubits.
Both single and two qubit process tomography are experimentally feasible and can be completed in less than 30 minutes per gate with 1000 shots (single qubit tomography takes significantly less time).
However, extending to the case of three qubits requires 1728 measurement settings, and at this point the required experiments is possibly infeasible.
This is due both  to the cost involved with running long experiments on quantum devices, and because we don't expect the noise processes to stay stable for the amount of time it would take to run process tomography on three qubits (or more).
While this is not necessarily a major drawback -- most gates on quantum qubits are one and two qubit gates, and understanding the noise models for each gate on a device is already useful for designing error mitigation processes -- nonetheless lifting this requirement may lead to more applications for the techniques in this paper. 
Extending these algorithms to the case where we only have access to measurement data that gives incomplete information about the full dynamics, or incorporating partial prior information, will be explored in future work.

Other enhancements for future work include the improvement of the run time of the algorithms.
In the present instance we have focused on showing that the algorithms are theoretically well founded, and can be benchmarked on simulated data and shown to give correct results.
So far we have made no attempt to optimise the classical computational phase of our procedure.
However, for running on real data a run time of hours is likely to be necessary, as opposed to days (as we have for the two qubit numerics).
There are a number of complementary ways this could be achieved, a straightforward example being the parallelisation of the loop over different branches of the matrix logarithm. 

\bigskip

Shortly after this paper appeared on the arXiv another work investigating fitting noise models to tomography data was posted \cite{samach2021lindblad}.
It takes a complementary approach to assessing Markovianity in near-term hardware using information-flow-based measures (cfr.~\cref{sec:RelatedWork}).

%\newpage

\section*{Code availability}

The python code implementing the algorithms presented in this work and used to produce the numerical results is available at \cite{gitlab_repo}.

The 1-qubit $X$-gate analysis (with 10,000 random samples) took 2 hours on a standard Intel x86 2.0GHz laptop, and the 2-qubit ISWAP analysis (8,500 random samples) took 2 weeks on a standard Intel  x86 3.40GHz desktop machine.
We made no effort to optimise the algorithm implementation.
In particular, the most costly part of the algorithms, namely the random sampling in the pre-processing, is trivially parallelisable, but we did not do this here.
The run-time can certainly be reduced significantly if desired.

\section*{Acknowledgements}
E.O.\ and T.S.C.\ are supported by the Royal Society. 
E.O.\ is also supported by the Bavarian state government with funds from the Hightech Agenda Bayern Plus as part of the Munich Quantum Valley, by the German Research Foundation DFG via the SFB/Transregio~352, and by the UK Hub in Quantum Computing and Simulation, part of the UK National Quantum Technologies Programme with funding from UKRI EPSRC (grant EP/T001062/1).
T.K.~is supported by the Spanish Ministry of Science and Innovation through the ``Severo Ochoa Program for Centres of Excellence in R\&D'' (CEX2019-00904-S) and PID2020-113523GB-I0 and started the work while being supported by the EPSRC through the Centre for Doctoral Training in Delivering Quantum Technologies (grant EP/L015242/1). \newline
This work was supported by Google Research Award ``Assessing non-Markovian noise in NISQ devices''.
We thank James Seddon and Calvin Liu for helpful comments on an earlier version of the manuscript.

\phantomsection

%=====================================================================================
%=====================================================================================
%=====================================================================================
%=====================================================================================
%=====================================================================================

\newpage

\phantomsection
\addcontentsline{toc}{section}{Appendix}
\appendix

\noindent {\LARGE \textbf{Appendix} }

%===============================================================
\section{Detailed discussion of algorithms}\label{app:algorithms}
%===============================================================

%-----------------------------------------------------------------------
\subsection{Pre-processing algorithms} \label{app:pre-processing}
%-----------------------------------------------------------------------

The algorithm to prepare the basis for an eigenvalue which is either complex or negative real is \cref{alg:construct_basis_complex}.
It takes as input the transfer matrix of the channel, $M$, and two sets of integers.
The first set, $ \degSet$, contains the indices of the eigenvectors corresponding to the degenerate eigenvalue itself.
The second set, $\conjDegSet$, contains the indices of the eigenvectors corresponding to the conjugate eigenvalue.\footnote{Clearly we require $|\degSet| = |\conjDegSet|$. Furthermore in the case of a negative real degenerate eigenvalue we have $\degSet= \conjDegSet$.}

The idea behind \cref{alg:construct_basis_complex} is that, for each set of eigenvectors $w_1$, $w_2$,...,$w_n$ associated with a degenerate eigenvalue $\lambda$ and eigenvectors  $u_1$, $u_2$,...,$u_n$ associated with $\lambda^\ast$ we solve the equation
\begin{equation} \label{kernel-equation}
	{\alpha_1}^\ast \ket{w_1^\dagger} + {\alpha_2}^\ast\ket{w_2^\dagger} + \cdots + \alpha_n^\ast\ket{w_n^\dagger} - \beta_1 \ket{u_1} - \beta_2 \ket{u_2} - \cdots - \beta_n \ket{u_n}  = \ket{0}
\end{equation}
for $\{\alpha_j\}_j$ and $\{\beta_j\}_j$, allowing us to construct a basis with the correct hermiticity preserving structure.

We solve \cref{kernel-equation} by arranging the vectors into the columns of a matrix
\begin{equation}
	A = ({w_1}^\dagger,{w_2}^\dagger , \cdots, {w_n}^\dagger, u_1, u_2, \cdots u_n)
\end{equation}
and finding its kernel.
If the variable $\mathrm{Nullity}(A) \coloneqq \mathrm{dim}(\ker(A))$ is equal to the dimension of the degenerate eigenspace then there is a hermiticity-preserving basis which spans the  eigenspace (in fact there are uncountably infinitely many choices of hermiticity-preserving bases).
The algorithm returns the basis vectors
\begin{equation}
	\ket{v_i} = \sum_j \alpha^{(i)}_j \ket{w_j}
\end{equation}
for $i \in (0, \mathrm{Nullity}(A))$, where $\alpha^{(i)}_j$ is the value of $\alpha_j$ in the $i^{\mathrm{th}}$ solution to \cref{kernel-equation}.

The pseudo-code for the case of degenerate, positive real eigenvalues is given by \cref{alg:construct_basis_real}.
It follows a similar idea to \cref{alg:construct_basis_complex}, with some additional processing to account for the possibility of self-adjoint eigenvectors.
The equation to solve in order to obtain a basis with the correct hermiticity-preserving structure is
\begin{equation} \label{kernel-equation_pos}
	{\alpha_1}^\ast \ket{w_1^\dagger} + {\alpha_2}^\ast\ket{w_2^\dagger} + \cdots + \alpha_n^\ast\ket{w_n^\dagger} - \beta_1 \ket{w_1} - \beta_2 \ket{w_2} - \cdots - \beta_n \ket{w_n}  = 0
\end{equation}
again for $\{\alpha_j\}_j$ and $\{\beta_j\}_j$.
As in the previous case, if $\mathrm{Nullity}(A)$ is equal to the dimension of the degenerate eigenspace then there is a hermiticity-preserving basis which spans the eigenspace given by
\begin{equation}
	\ket{v_i} = \sum_j \alpha^{(i)}_j \ket{w_j}.
\end{equation}
However, now the basis may include self-adjoint eigenvectors, as well as hermitian conjugate eigenvectors.
These need to be treated separately in the next stage of the pre-processing.
To verify if the $i^{\mathrm{th}}$ basis-vector is self-adjoint we check whether $(\alpha^{(i)}_j)^* - \beta^{(i)}_j  < p$, for all $ j$.
If the total number of non-self-adjoint basis vectors is even, the algorithm returns two sets of basis vectors -- the self-adjoint set and the non-self-adjoint set.
If the total number of non-self-adjoint basis vectors is odd, the algorithm finds the non-self-adjoint basis vector for which $\sum_j (\alpha^{(i)}_j)^* - \beta^{(i)}_j$ is minimised, and returns this with the set of self-adjoint basis vectors instead of the non-self-adjoint set.
Note that there is an infinite number of possible basis choices which respect the hermiticity-preserving structure, and which basis is chosen will affect the distance to the closest Lindbladian.
To handle this, we use a randomised construction to generate $r$ bases satisfying the hermiticity-preserving property, and we run the convex optimisation algorithm on each one, keeping the optimal result.
The number of random bases, $r$, is an input to \cref{alg:MAIN}.
Running the algorithm with higher $r$ will give a better result (numerical results on degenerate 1- and 2-qubit channels are given in \cref{sec:examples}).

The algorithm to prepare a random basis is given in \cref{alg:construct_basis_random}.
It takes as input $M$, lists indicating the indices of eigenvectors associated to each degenerate eigenvalue $\{sl_\lambda \mid \mathrm{for \ degenerate \ \lambda}\}$, and the bases constructed by \cref{alg:construct_basis_complex} and \cref{alg:construct_basis_real}.
It returns a random choice of basis, $S$.
In order to build $S$, \cref{alg:construct_basis_random} checks whether $i$ is in any of the $sl_\lambda$ for each $i \in \mathrm{dim}(M)$.
If it is not, we set the $i^{\mathrm{th}}$ column of $S$ equal to the $i^{\mathrm{th}}$ eigenvector of $M$.
If it is, we check whether the $i^{\mathrm{th}}$ column needs to be constructed randomly from either self-adjoint basis vectors or non-self-adjoint basis vectors; or whether it needs to be obtained by conjugating another column.
For $\lambda$ real, which columns are produced randomly (from self-adjoint or non-self-adjoint basis vectors) and which are found via conjugation is arbitrary.
For $\lambda$ complex, either all columns associated to $\lambda$ are constructed randomly, or all are retrieved by conjugating the basis vectors associated to $\lambda^\ast$. This is determined by which order $sl_\lambda$ and $sl_{\lambda^\ast}$ are given as input to \cref{alg:construct_basis_complex}.
If the column needs to be prepared randomly then we generate $|sl_\lambda|$ random seeds $\kappa_j$.
The $i^{\mathrm{th}}$ column of $S$ is then
\begin{equation}
	S[:,i] = \sum_j \kappa_j v_j,
\end{equation}
where $v_j$ are the hermiticity-structure preserving basis vectors associated with $\lambda$, and the sum is over either all self-adjoint or all non-self-adjoint basis vectors.
If the column needs to be obtained by conjugation we find which column it is conjugate to, denoted $k$, and the $i^{\mathrm{th}}$ column of $S$ is then given by
\begin{equation}
	S[:,i] = \Flip S[:,k]^\ast .
\end{equation}
The final step of pre-processing is to compute $R = SMS^{-1}$ for each random basis.

\begin{algorithm}\small
	\SetKwInOut{Input}{Input}\SetKwInOut{Output}{Output}
	\Input{Matrix $M$, Array of integers $\degSet$, Array of integers $\conjDegSet$}
	\KwResult{Matrix $\newBasis$}
	
	$\mathrm{eigvecs} \leftarrow \mathrm{eigenvectors}(M)$
	
	$\degCount \leftarrow \mathrm{length}(\degSet)$
	
	\For{$i \in (0,\degCount)$}
	{$j \leftarrow \degSet[i]$
		
		$k \leftarrow \conjDegSet[i]$
		
		$\ket{w_i} \leftarrow \mathrm{eigvecs}[j]$
		
		$\ket{u_i} \leftarrow \mathrm{eigvecs}[k]$
	}
	
	$A_1 \leftarrow [ \Flip \ket{w_i^\ast} \mid i \in (0,\degCount)]$

	$A_2 \leftarrow [ -\ket{u_i} \mid i \in (0,\degCount) ]$

	$A = (A_1,A_2)^T$				(Kernel of A is solution to
	$\alpha_1^\ast \Flip \ket{w_1^\ast} + \alpha_2^\ast \Flip \ket{w_2^\ast} +\cdots+\alpha_n^\ast \Flip \ket{w_n^\ast} - \beta_1 \ket{u_1} - \beta_2 \ket{u_2}-\cdots- \beta_n \ket{u_n} = \ket{0}$~)
	$\mathrm{nullity} = \mathrm{dim}\left(\ker(A)\right)$

	\eIf{$\mathrm{nullity}==\degCount$}
	{
		\For{$i \in (0, \degCount)$}
		{
			$\sumValues = [ w_j * \ker(A)_{i,j}^\ast  \mid j \in (0,\degCount)]$

			$\newBasis[:,i] = \sum_j(\sumValues)$
		}
		\Output {$\True$, $\newBasis$}}
	{\Output {$\False$, $[]$}
	}
	\caption{Construct conjugate basis for cluster of a complex or real negative eigenvalue}\label{alg:construct_basis_complex}
\end{algorithm}

\begin{algorithm}\small
	\SetKwInOut{Input}{Input}\SetKwInOut{Output}{Output}
	\Input{Matrix $M$, Array of integers $\degSet$, Real number $p$}
	\KwResult{Matrices $[\selfBasis, \conjBasis]$}
	
	$\mathrm{eigvecs} \leftarrow \mathrm{eigenvectors}(M)$
	
	$\degCount \leftarrow \mathrm{length}(\degSet)$
	
	\For{$i \in (0,\degCount)$}
	{
		$j \leftarrow \degSet[i]$
		
		$\ket{w_i} \leftarrow \mathrm{eigvecs}[j]$
	}

	$A_1 \leftarrow [ \Flip \ket{w_i^\ast} \mid i \in (0,\degCount)]$
	
	$A_2 \leftarrow [ -\ket{w_i} \mid i \in (0,\degCount) ]$

	$A = (A_1,A_2)^T$				(Kernel of A is solution to
	$\alpha_1^\ast \Flip \ket{w_1^\ast} + \alpha_2^\ast\Flip \ket{w_2^\ast} +\cdots+ \alpha_n^\ast \Flip \ket{w_n^\ast} 
	- \beta_1 \ket{w_1} - \beta_2 \ket{w_2}-\cdots- \beta_n \ket{w_n}  = 0$)
	
	$\mathrm{nullity} = \mathrm{dim}\left(\ker(A)\right)$

	\eIf{$\mathrm{nullity}==\degCount$}
	{
		
		$\selfCheck \leftarrow [\True \mathrm{\ if \ }i^{\mathrm{th}}\mathrm{\ column \ of\ }\ker(A)\mathrm{\ is\  self-adjoint \ within\  precision \ }p  \mid i \in (0,\degCount)]$
		
		$\numSelf \leftarrow \sum_i(\selfCheck)$
		
		$\numConj \leftarrow \degCount - \numSelf$
		
		\If{$\numConj \mod2 \neq 0$}
		{$\additionalSelf \leftarrow \mathrm{index \ of \ extra \ dimension \ which \ is \ closest \ to \ self \ adjoint}$
			
			$\selfCheck[\additionalSelf] \leftarrow \True$
			
			$\numSelf =\numSelf + 1 $
			
			$\numConj = \numConj - 1 $
			
		}
		$k=0$
		
		$l=0$
		
		\For{$i \in (0, \degCount)$}
		{
			
			$\sumValues = [\ket{w_j} * \ker(A)_{i,j}^\ast  \mid j \in (0,\degCount)]$
			
			\eIf{$\selfCheck[i]$}
			{$\selfBasis[:,k] = \sum_j(\sumValues)$
				
				$k =k+1$}
			{$ \conjBasis[:,l] = \sum_j(\sumValues)$
				
				$l=l+1$}
		}
		\Output {$\True$, $[\selfBasis,\conjBasis]$}}
	{\Output {$\False$, $[]$}
	}
	\caption{Construct self-adjoint / conjugate basis for cluster of a real positive eigenvalue}\label{alg:construct_basis_real}
\end{algorithm}

\begin{algorithm}\small
	\SetKwInOut{Input}{Input}\SetKwInOut{Output}{Output}
	\Input{Matrix $M$, Three lists of lists $\posDegSets$, $\negDegSets$, $\complexDegSets$, Three sets of bases $\posBases$, $\negBases$, $\complexBases$, List of boolean values $\conjTest$}
	\KwResult{Matrix $\newBasis$}
	
	$\mathrm{eigvecs} \leftarrow \mathrm{eigenvectors}(M)$
	
	$\allDegSets \leftarrow [ sl \mid sl \in \posDegSets \cup \negDegSets \cup \complexDegSets)$

	\For{$sl \in \allDegSets$}
	{$sl.\mathrm{sort}()$
	}
	
	\For{$i \in \mathrm{dim}(M)$}
	{\eIf{$i  \notin sl \forall sl \in \allDegSets$}
		{
			$\newBasis[:,i] = \mathrm{eigvecs}[:,i]$
			
		}
		{\uIf{$i \in sl \textrm{\emph{\textbf{  for  }}} sl \in \posDegSets$}
			{$\mathrm{sets} = \posDegSets$
				
				$\mathrm{bases} = \posBases$}
			
			\uElseIf{$i \in sl \textrm{\emph{\textbf{  for  }}} sl \in \negDegSets$}
			{$\mathrm{sets} = \negDegSets$
				
				$\mathrm{bases} = \negBases$}
			
			\Else
			{$\mathrm{sets} = \complexDegSets$
				
				$\mathrm{bases} = \complexBases$}
			
			$j \leftarrow  \mathrm{index \ of \ list \ } i \mathrm{\ is \ in \ within \ list \ of \ lists \ ``sets"}$
			
			$k \leftarrow  \mathrm{index \ of  \ } i \mathrm{\ within \ list } j$
			
			$n \leftarrow \mathrm{length}\left(\mathrm{list \ } j \right)$

			\eIf{$\conjTest[i]$}
			{
				$\randomSeeds \leftarrow \mathrm{list \ of\ } n \mathrm{\ random\  numbers}$
				
				$\mathrm{tilde}\_\mathrm{basis} \leftarrow \sum\left(\randomSeeds[k] * \mathrm{bases}[j][:,k]\right)$
				
				$\newBasis[:,i] \leftarrow \mathrm{tilde}\_\mathrm{basis} / \mathrm{norm}\left(\mathrm{tilde}\_\mathrm{basis} \right)$}
			{
				
				$\mathrm{conjugate}\_\mathrm{value} \leftarrow \mathrm{index \ which \ }j\mathrm{\ is \ conjugate \ with}$
				
				$\newBasis[:,i] \leftarrow \Flip (\newBasis[:,\mathrm{conjugate}\_\mathrm{value}])^\ast$
			}
	}}

	\Output {$\newBasis$}
	\caption{Construct conjugate random choice of basis for quantum channel with cluster(s) of eigenvalues}\label{alg:construct_basis_random}
\end{algorithm}

\FloatBarrier

%-----------------------------------------------------------------------
\subsection{Convex Optimisation algorithms for a single snapshot} \label{app:optimisation}
%-----------------------------------------------------------------------

In \hyperref[alg:app:Lindbladian]{Algorithm 1} we present a fleshed out version of the algorithm to find the best fit Lindbladian to a tomographic snapshot.
It looks for the closest Lindbladian to $\log R$ in the $\varepsilon$-neighbourhood of $M$ by formulating a convex optimisation task whose constraints are exactly the necessary and sufficient conditions for a Lindblad generator, as discussed in \cref{sec:Markovianity}.
We iterate over different branches of the matrix logarithm and pick the resulting Lindbladian from the convex optimisation programme whose generated channel is the closest to~$M$.

The first remark about the algorithm is that we are not searching for the closest Markovian channel, but for the closest Lindbladian to its matrix logarithm. A natural question is then whether the two objects are precisely related, that is, if the closest Lindbladian generates the closest Markovian channel. In the general case, this is not true.
A simple counter-example is provided by the perturbed $X$-gate discussed in \cref{sec:perturbation_example}, where the closest Lindbladian generates a map close to the identity, although the unperturbed $X$-gate is a closer Markovian channel.
However, consider the upper bound for general matrices $A$ and $B$ (Lemma~11 in~\cite{CEW09})
\begin{equation}\label{eq:relation_log}
	\norm{\exp A-\exp B}
	\leq
	\norm{A-B}\exp \norm{A-B} \, \exp \norm{A}.
\end{equation}
If we now ask $B$ to be Lindbladian, then the closest Lindbladian $L$ to $A$ is also the operator minimising this upper bound.
This implies that the distance between the Markovian channel $\exp L$ retrieved by our algorithm and the input matrix $M=\exp L_R$ is upper-bounded by
\begin{equation}\label{eq:tolerance_and_relation_log}
	\varepsilon' \coloneqq \norm{C-L_R}\exp \norm{C-L_R} \, \exp \norm{L_R},
\end{equation}
where we assume $\exp C$ to be the closest Markovian operator to $M$. Thus, by incorporating $\varepsilon '$ in the tolerance parameter $\varepsilon$ (together with the tomographic inaccuracy and the predicted amount of non-Markovianity), we expect to find a compatible Markovian evolution for the input $M$.

Note that the convex optimisation problem contains a second-order cone constraint (when formulating the Frobenius norm minimisation in epigraph form), a linear matrix inequality (LMI) constraint to ensure the conditionally completely positive condition and a first-order cone constraint representing the trace preserving condition. As in any convex optimisation programme, every minimum will be a global minimum.

Moreover we have $\Fnormn{A} =\Fnormn{A^\Gamma}$ for $A$ in the elementary basis representation. Thus, $\Fnormn{L-L} = \Fnormn{(L-L)^\Gamma} = \Fnormn{X-L^\Gamma}$.
This is useful in order to run the convex optimisation task on the variable $X$ without involutions.

\medskip

In \hyperref[alg:app:mu]{Algorithm 2} we present a fleshed out version of the algorithm to find the non-Markovianity parameter~$\mu$.

%===============================================
%-------------------First algorithm------------
%===============================================
\begin{algorithm}[htbp]\small
	\SetAlgoRefName{1 (restatement)}
	\SetKwInOut{Input}{Input}\SetKwInOut{Output}{Output}
	\Input{matrix $M$, $\NDsquare$ Matrix $R$  with $\mathrm{dim} R = \mathrm{dim} M$, positive real number~$\varepsilon$, positive integer $m_{\max}$}
	\KwResult{$L$ closest Lindbladian to $\vec{m}$-branch of $\log R$ such that $\Fnorm{M-\exp L} < \varepsilon$ is minimal over all
		$\vec{m}\in\{-m_{\max},-m_{\max}+1,\dots,0,\dots,m_{\max}-1,m_{\max}\}^{\times d^2}$}
	
	\medskip
	
	$d\leftarrow \sqrt{\mathrm{dim} M}$ \\
	$G_0 \leftarrow \log R$ 
	
	\smallskip
	
	$P_j \leftarrow \dyad{r_j}{\ell_j}$ \  \textit{($\ket{r_j}$ and $\bra{\ell_j},\, j=1,\dots,d^2$  right and  left eigenvectors of~$M$)}
	
	\smallskip
	
	$\ket{\omega}\leftarrow \sum_{j=1}^d \ket{j,j}$,
	\ $\omega_\perp \leftarrow \1-\dyad{\omega}{\omega}$
	
	\medskip
	
	\For{$G_{\vec{m}} \leftarrow G_0 + 2\pi \i \sum_{j=1}^{d^2} m_j \, P_j $ \  \textit{(branches of $G_0$)}}{
		\smallskip
		Run convex optimisation programme on variable $X(\vec{m})$:\\
		\Indp $\begin{array}{ll}
			\text{minimise} & \Fnorm{X(\vec{m})-L_{\vec{\vec{m}}}^\Gamma}   \\
			\text{subject to } &X(\vec{m}) \text{ hermitian}\\
			&\omega_\perp X(\vec{m}) \omega_\perp \geq 0 \\
			&\norm{\Tr_1[X(\vec{m})]}_{1} = 0
		\end{array}$\\
		\Indm
		
		\smallskip
		
		\If{$\Fnorm{M - \exp X^\Gamma(\vec{m})}<\varepsilon$}{
			
			\smallskip
			
			Store $X(\vec m)$ \\
			$\mathrm{distance}(\vec{m}) \leftarrow \Fnorm{M - \exp X^\Gamma(\vec{m})}$}
		
		\medskip
		
	}%end for loop
	\Return$L=X^\Gamma(\vec{m}') \text{ for } \vec{m}'=\mathrm{argmin}\, \{\mathrm{distance}(\vec{m})\}$
	\caption{Retrieve best-fit Lindbladian} \label{alg:app:Lindbladian}
\end{algorithm}

%===============================================
%-------------------Second algorithm------------
%===============================================

\begin{algorithm}[htbp]\small
	\SetAlgoRefName{2 (restatement)}
	\SetKwInOut{Input}{Input}\SetKwInOut{Output}{Output}
	\Input{matrix $M$, $\NDsquare$ Matrix $R$  with $\mathrm{dim} R = \mathrm{dim} M$, positive real number~$\varepsilon$, positive integer $m_{\max}$}
	\KwResult{generator $H'$ of the hermiticity- and trace-preserving channel in the $\varepsilon$-ball of $M$ with the smallest non-Markovianity measure $\mu_{\min}$ over all $\vec{m}\in\{-m_{\max},-m_{\max}+1,\dots,0,\dots,m_{\max}-1,m_{\max}\}^{\times d^2}$.}
	
	\medskip
	
	$d\leftarrow \sqrt{\mathrm{dim} M}$ \\
	$G_0 \leftarrow \log R$
	
	\smallskip
	
	$P_j \leftarrow \dyad{r_j}{\ell_j}$ \  \textit{($\ket{r_j}$ and $\bra{\ell_j},\, j=1,\dots,d^2$  right and  left eigenvectors of~$M$)}
	
	\smallskip
	
	$\ket{\omega}\leftarrow \sum_{j=1}^d \ket{j,j}$,
	%\smallskip
	\ $\omega_\perp \leftarrow \1-\dyad{\omega}{\omega}$
	
	\medskip
	
	$\mu_{\min} \in \mathds{R}^{+}$ \hspace{12pt} \textit{(initialize it with an high value)}
	
	\smallskip
	
	$\delta$ satisfying $\varepsilon = \exp ( \delta) \cdot \delta \cdot \Fnorm{G_0}$, \hspace{8pt} $\delta_{\max}=10 \cdot \delta$, \hspace{8pt} $\delta_{\mathrm{step}} \in \mathds{R}^{+}$ \hspace{15pt}
	\medskip
	
	\For{$\delta < \delta_{\max}$}{
		
		\smallskip
		
		\For{$G_{\vec{m}} \leftarrow G_0 + 2\pi \i \sum_{j=1}^{d^2} m_j \, P_j $ \  \textit{(branches of $G_0$)}}{

			Run convex optimisation programme on variables $\mu(\vec{m},\delta)$ and $X(\vec{m},\delta)$:\\
			\Indp $\begin{array}{ll}
				\text{minimise} & \mu(\vec{m},\delta) \\
				\text{subject to } &X(\vec{m},\delta) \text{ hermitian}\\
				&\Fnorm{X(\vec{m},\delta)-G_{\vec{m}} ^\Gamma } \leq \delta \\
				&\omega_\perp X(\vec{m},\delta) \omega_\perp + \mu \frac{\1}{d} \geq 0 \\
				&\norm{\Tr_1[X(\vec{m},\delta)]}_{1} = 0
			\end{array}$\\
			\Indm
			
			\medskip
			
			\If{$\Fnorm{M - \exp X^\Gamma(\vec{m},\delta) } < \varepsilon$}{
				
				\smallskip
				
				Store  $(\mu(\vec{m},\delta), X^\Gamma(\vec{m},\delta))$}
			
		}%end for loop in m
		
		$\delta \leftarrow \delta+\delta_{\mathrm{step}}$
		
	}%end for loop for delta

\smallskip

\Return{$( \mu_{\min},G')=( \mu(\vec{m}',\delta'), X^\Gamma(\vec{m}',\delta') )$ \hspace{5pt} for $(\vec{m}',\delta')=\mathrm{argmin} \ \mu(\vec{m},\delta)$ }
	
	\caption{Compute non-Markovianity measure~$\mu_{\min}$} \label{alg:app:mu}
\end{algorithm}

\FloatBarrier

%=======================================================================
\subsection{Multiple snapshots} \label{app:mulSnap}
%=======================================================================

A detailed version of the convex optimisation algorithm for the case of multiple snapshots is given in \cref{alg:app:multiple_snapshots}.

\begin{algorithm}[h]\small
	\SetKwInOut{Input}{Input}\SetKwInOut{Output}{Output}
	\Input{$(d^2 \times d^2)$-dimensional matrices $M_1\dots,M_N$ , positive real numbers $t_1,\dots,t_N$, positive real number $\varepsilon$, positive integer $m_{\max}$}
	\KwResult{$L$ Lindbladian minimising $\sum_{c=1}^N \Fnorm{t_c\,L-\log M_c}$ such that $\Fnorm{M_c-\exp t_c\,L} < \varepsilon$ for all $c$}
	
	\medskip
	
	\For{ $c=1,\dots,N$}{
		
		\smallskip
		
		$G_0^c \leftarrow \log M_c$
		
		\smallskip
		
		$P_j^c \leftarrow \dyad{r_j}{\ell_j}$ \ \textit{($\ket{r_j}$ and $\bra{\ell_j},\, j=1,\dots,d^2$ right and left eigvecs of~$M_c$)}
	}

	\medskip
	
	$\delta$ satisfying $\varepsilon = \exp ( \delta) \cdot \delta \cdot \Fnorm{G_0}$, \hspace{8pt} $\delta_{\max}=10 \cdot \delta$, \hspace{8pt} $\delta_{\mathrm{step}} \in \mathds{R}^{+}$ \hspace{15pt}
	\medskip
	
	\For{$\delta < \delta_{\max}$}{
		
		\smallskip
		
		\For{$G^c_{\vec{m}} \leftarrow G^c_0 + 2\pi \i \sum_{j=1}^{d^2} m_j \, P^c_j $ \  \textit{(branches of $G_0$)}}{
			%$c=1,\dots,q$
			
			Run convex optimisation programme on variable $X(\vec{m})$:\\
			\Indp $\begin{array}{ll}
				\text{minimise} & \sum_c \Fnorm{ t_c\,X(\vec{m}, \delta)-(G_{\vec{m}} ^c)^\Gamma}   \\
				\text{subject to } &X(\vec{m}, \delta) \text{ hermitian}\\
				&\Fnorm{X(\vec{m},\delta)-(G_{\vec{m}} ^c)^\Gamma} \leq \delta \hspace{5pt} \text{for all } $c=1,\dots,N$\\
				&\omega_\perp X(\vec{m}, \delta) \omega_\perp \geq 0 \\
				&\norm{\Tr_1[X(\vec{m}, \delta)]}_{1} = 0\\
			\end{array}$\\
			\Indm
			
			\smallskip
			\If{$\Fnorm{M_c - \exp  t_c\, X^\Gamma(\vec{m}, \delta)} < \varepsilon$ \hspace{5pt} \emph{for all} $c=1,\dots,N$}{
				
				\smallskip
				
				Store $X(\vec m, \delta)$ \\
				$\mathrm{distance}(\vec{m},\delta) \leftarrow \sum_c \Fnorm{M_c - \exp  t_c\, X^\Gamma(\vec{m}, \delta)}$}
			
		}
		
		$\delta \leftarrow \delta+\delta_{\mathrm{step}}$
		
	}%end delta loop
		
		\Return $L=X^\Gamma(\vec{m}', \delta') \hspace{5pt} \text{ for } (\vec{m}', \delta')=\mathrm{argmin}\, \{\mathrm{distance}(\vec{m}, \delta)\}$
		
	\caption{Retrieve best-fit Lindbladian for multiple snapshots}\label{alg:app:multiple_snapshots}
\end{algorithm}

Once again, dealing with degeneracies requires more work. If the perturbation of an $n$-degenerate complex eigenvalue is identified, then we should check whether this is consistent with all other snapshots, that is, all measurements show a cluster of $n$ eigenvalues. Because of the hermiticity-preserving condition we also expect a partner cluster of $n$ eigenvalues, corresponding to the complex-conjugate partner subspace. Note that the two clusters will overlap when they approach the negative axis, and conversely that any cluster close to the negative axis is expected to split in two partner clusters representing the perturbation of two eigenvalues related by complex conjugation: this means that when taking successive snapshots, we can avoid the case of perturbed negative eigenvalues by choosing suitable measurement times.
Secondly, we should also verify if in all snapshots both the multi-dimensional subspace of a cluster and its partner subspace are consistent with the same unperturbed pair of hermitian-related eigenspaces, up to some approximation. Indeed, recall that eigenspaces are stable with respect to perturbation.
Moreover, since here we are checking for time-\emph{independent} Markovianity from multiple snapshots, we are interested in the case where the Lindbladians for all snapshots are the same.

Given this, and the results of \cref{par:approximate_HP_structure}, our approach to handling an $n$-degenerate complex eigenvalue $\lambda$ is as follows.
We select one of the snapshots, $c'$, at random, and apply the pre-processing steps from \cref{alg:MAIN} to $M_{c'}$ in order to obtain a random hermitian-related basis of vectors for the $\lambda$ and $\lambda^\ast$ subspaces. Denote these by $\{v_j\}$ and $\{v_j^\dagger \}$ respectively, and denote the projections onto the subspace spanned by these vectors by $\Pi_{v}$ and $\Pi_{v^\dagger}$.
We then determine whether this choice of basis is compatible with the other snapshots, by checking that
\begin{equation} \label{eq:constraint_multiple_1}
	\sum_c \Big\Vert  \Pi_c(\lambda) - \Pi_v\Big\Vert  + \sum_c \Big\Vert  \Pi_c(\lambda^\star) - \Pi_{v^\dagger}\Big\Vert  \leq \varsigma_1,
\end{equation}
where $\varsigma_1$ is a tolerance parameter for the perturbation of subspaces (which we set arbitrarily, but that one can derive rigorously using refs.~\cite{StewartSun} and~\cite{Kato}).
A set of basis vectors retrieved with the above approach can then be used as the new eigenvectors of the input matrices $\setn{M_c}_c$ for the clusters of eigenvalues, in analogous fashion as for the single-snapshot case, and then run \cref{alg:MR:multiple_snapshots} on these modified operators.
This procedure for constructing a random basis should be repeated a number of times, and the optimal result kept, as in the single snapshot case.

A special case is the perturbation of degenerate real eigenvalues that do not turn into complex ones through the sequence of snapshots, corresponding to real eigenvalues for the Lindblad generator. In this case, we do not have a partner invariant subspace. As discussed already, an eigenspace of an hermiticity-preserving operator with respect to a real eigenvalue admits self-adjoint eigenvectors in addition to hermitian-related pairs.
Consider an $n$-degenerate real eigenvalue $\kappa$.
As in the complex case, we pick a snapshot $c'$ at random, and apply the pre-processing steps from \cref{alg:MAIN} to $M_{c'}$ to find a basis for the subspace $\Pi_{c'}(\kappa)$.
This basis will be composed of a set of vectors $\{v_j\}_{j=1}^p$, a set of hermitian-related vectors $\{v_j^\dagger\}_{j=1}^p$, and a set of self-adjoint vectors $\{s_j \}_{j=1}^{n-2p}$, for some value $p=1,\dots,n/2$.
Denoting by $\Pi_v$, $\Pi_{v^\dagger}$, $\Pi_s$ the corresponding projections onto the subspaces spanned by these sets of vectors, the constraint~\cref{eq:constraint_multiple_1} turns into:
\begin{equation} \label{eq:constraint_muptiple_2}
	\sum_c \Big\Vert  \Pi_c(\kappa) - \Pi_v - \Pi_{v^\dagger} -  \Pi_s   \Big\Vert  \leq \varsigma_2.
\end{equation}
As in the complex case, we run \cref{alg:app:multiple_snapshots} on the input matrices resulting by using these vectors as the bases for the degenerate subspace in each $M_c$, and then repeat the randomised process a number of times.

%----------------------------------------------------------------
\subsection{The main algorithm for single snapshot -- complete version}\label{app:main}
%----------------------------------------------------------------
Our main algorithm, which provides a recipe for evaluating non-Markovianity in the single snapshot case, is given in \cref{alg:app:MAIN}.
It takes as input an estimate of the channel, $M$, a precision parameter, $p$, an accuracy parameter $\varepsilon$, and an integer, $r$, which determines how many random basis choices will be tested in the case of degenerate eigenvalues.

\begin{algorithm}\small
	\SetKwInOut{Input}{Input}\SetKwInOut{Output}{Output}
	\Input{$\NDsquare$ matrix $M$, integer $\mathrm{random}\_\mathrm{samples}$, positive real numbers~$p,\, \varepsilon$}
	\KwResult{Lindbladian, $L$, consistent with $M$ (within error tolerance $\varepsilon$), or if no such $L$ exists, non-Markovianity parameter~$\mu$}
	
	\medskip
	
	$\lambda \leftarrow \mathrm{eigenvalues}(M)$
	
	\smallskip

	$\posCheck \leftarrow \left[(\lambda_i - \lambda_j < p) \cap \left(\Im(v_i) < p \right) \cap \left(\Re(\lambda_i > 0)\right) \mid i,j \in (0,\mathrm{dim} M), i>j\right]$
	
	$\negCheck \leftarrow \left[(\lambda_i -\lambda_j < p)\cap\left(\Im(\lambda_i) < p \right)\cap \left(\Re(\lambda_i < 0)\right) \mid i,j \in (0,\mathrm{dim} M), i>j\right]$

	$\compCheck \leftarrow \left[(\lambda_i - \lambda_j < p)\cap \left(\Im(\lambda_i) > p \right) \mid i,j \in (0,\mathrm{dim} M), i>j\right]$
	
	$\mathrm{checklist} \leftarrow [\posCheck,\negCheck,\compCheck]$
	
	\medskip
	
	\uIf{$\mathrm{sum}(\mathrm{checklist}) = 0$}
	{
		$L \leftarrow$ Run \hyperref[alg:app:Lindbladian]{Algorithm 1} on $M,M,\varepsilon$
		
		\uIf{ $\|M - \exp(L) \| < \varepsilon$}
		
		{\Output{$L$}}
		
		\Else{
			$\mu \leftarrow$ Run \hyperref[alg:app:mu]{Algorithm 2} on input $M, M, \varepsilon$
			
			\Output{$\mu$}
	}}
	\uElseIf{$\mathrm{sum}(\mathrm{\posCheck}) = {\mathrm{dim}M \choose 2}$}
	{ \Output{The channel is consistent with the identity map}}
	\Else{

		$\posDegSets \leftarrow \mathrm{list \ of \ sets \ of \ mutually \ deg. \ real \ positive \ eigenvalues}$
		
		$\negDegSets \leftarrow \mathrm{list \ of \ sets \ of \ mutually \ deg. \ real \ negative \ eigenvalues}$

		$\complexDegSets \leftarrow \mathrm{list \ of \ sets \ of \ mutually \ deg. \  complex \ eigenvalues}$

		$\conjPairSets \leftarrow  \mathrm{2d \ array \ indicating \ pairs \ of \ deg. \ complex \ sets \ are \ conjugate}$
		
		$\posBases \leftarrow []$, \hspace{10pt} $\negBases \leftarrow []$, \hspace{10pt}	$\complexBases \leftarrow []$
		
		\smallskip
		
		\For{$i \in (0,\mathrm{length}(\posDegSets))$}
		{$basis \leftarrow $ Run \cref{alg:construct_basis_real} on inputs $M$, $\posDegSets[i]$, $p$
			
			$\posBases.\mathrm{append}(basis)$}

		(continuing in Part II) $\dots$
		
	}%End main ELSE
	\caption{Main algorithm, including pre-processing -- Part I%to handle clusters of eigenvalues
	}\label{alg:app:MAIN}
\end{algorithm}
%
%-------------------------------------------------------------------------------
%
\NoCaptionOfAlgo
\begin{algorithm}\small
	\SetAlgoRefName{}
	\SetKwInOut{Input}{Input}\SetKwInOut{Output}{Output}
	\Input{$\NDsquare$ matrix $M$, integer $\mathrm{random}\_\mathrm{samples}$, positive real numbers $p$, $\varepsilon$}
	\KwResult{Lindbladian, $L$, which is consistent with $M$ (within error tolerance $\varepsilon$), or if no such $L$ exists, non-Markovianity parameter $\mu$}
	
	\Else{
		
		$\dots$ (continuing from Part I)
		
		\For{$i \in (0,\mathrm{length}(\posDegSets))$}
		{$basis \leftarrow $ Run \cref{alg:construct_basis_complex} on inputs $M$, $\negDegSets[i]$,  $\negDegSets[i]$
			
			$\negBases.\mathrm{append}(basis)$}
		
		\For{$i \in (0,\mathrm{length}(\conjPairSets))$}
		{$basis \leftarrow $ Run \cref{alg:construct_basis_complex} on inputs $M$, $\conjPairSets[i,0]$, $\conjPairSets[i,1]$
			
			$\complexBases.\mathrm{append}(basis)$}
		
		$\conjTest \leftarrow \mathrm{list \ of \ boolean \ values \ indicating \ if \ } i^{\mathrm{th}} \mathrm{\ eigvect \ needs \ to \ be \ conjugated}$
		
		\For{$r \in (0,\mathrm{random}\_\mathrm{samples})$	}
		{$S \leftarrow $ Run \cref{alg:construct_basis_random} on inputs $M$, $\posDegSets$, $\negDegSets$, $\complexDegSets$, $\posBases$, $\negBases$, $\complexBases$, $\conjTest$
			
			$R_{r} \leftarrow S M S^{-1}$

			$L_r \leftarrow$ Run \hyperref[alg:app:Lindbladian]{Algorithm 1} on $M$, $R_r$, $\varepsilon$
		}
		
		\uIf{$\min (\|\exp(L_r) - M \| )< \varepsilon$}
		
		{$L_{\min} \leftarrow L_r $ that minimises $ \|\exp(L_r) - M \|$
			
			\Output{$L_{\min}$}}
		
		\Else{
			\For{$r \in (0,\mathrm{random}\_\mathrm{samples})$	}
			{
				$H_r, \mu_r \leftarrow$ Run \hyperref[alg:app:mu]{Algorithm 2} on input $M, R_r, \varepsilon$
			}
			
			$\mu_{\min} \leftarrow \min(\mu_r)$

			\Output{$\mu_{\min}$}}

	}
	\caption{Main algorithm, including pre-processing -- Part II}
\end{algorithm}

\RestoreCaptionOfAlgo

\FloatBarrier

%============================================
\section{Perturbation theory and hermiticity-preserving structures}\label{app:perturbations}
%============================================

%---------------------------------------------------------------- 
\subsection{Dealing with exactly degenerate input matrices}\label{sec:matrix_perturbation}
%----------------------------------------------------------------

In this section we show that it is always possible, given a degenerate Choi-hermitian matrix, to produce an arbitrarily close $\NDsquare$ matrix that preserves the hermiticity-preserving property. Formally,

\begin{theorem}[$\NDsquare$ matrices are dense in the Choi-hermitian matrix set]\label{thm:Choi-hermitian_dense}
	Let $M$ be an hermiticity-preserving matrix, either non-diagonalisable or having degenerate spectrum (or both). Then for any $\epsilon$ there exists an hermiticity-preserving $\NDsquare$ matrix $\wtM$ such that $\lVert\wtM - M\rVert<\epsilon$.
\end{theorem}

This result allows us to resolve the problem of the freedom of basis choice and reduce to a unique principle branch of the matrix logarithm. In real tomographic data, eigenvalues will invariably be non-degenerate. 
But we include this section to make clear that our algorithm can cope with any input, and makes absolutely no assumptions on the form of the matrices.

For the proof, we will make use of the argument presented in ref.~\cite{DenseDiagMatrices} showing that the set of NDS matrices is dense to prove that the same is true when we restrict to the subset of hermiticity-preserving matrices. In particular, we will apply the following definition and results.
\begin{definition}[Resultant of two polynomials]
	Let $p(x) = a_n x^n+\dots+a_1 x +a_0$ and $q(x) = b_m x^m+\dots+b_1 x +b_0$ be two polynomial in $x$ of degree $n$ and $m$, respectively. The \emph{resultant of $p$ and $q$}, denoted by $\mathrm{Res}(p,q)$, is then given by the determinant of the $(n+m) \times (n+m)$ \emph{Sylvester matrix} of $p$ and $q$.
\end{definition}

\begin{lemma}
	\begin{equation}
		\mathrm{Res}(p,q) = a_n^m \, b_m^n \prod_{j=1}^n\prod_{k=1}^m\left(r_j - s_k\right),
	\end{equation}
	where $\set{r_j}_j$ and $\set{s_k}_k$ are the roots of $p$ and $q$, respectively.
\end{lemma}
\begin{corollary}\label{resultant_roots}
	$p$ and $q$ have a common root if and only if $\mathrm{Res}(p,q)=0$.
\end{corollary}

\begin{proof}[Proof of \cref{thm:Choi-hermitian_dense}]
	Given a Choi-hermitian matrix $M$, possibly defective or derogatory, we consider its perturbation $\widetilde M = (1-\alpha) M + \alpha E$, where $\alpha > 0$ and $E$ being an hermiticity-preserving $\NDsquare$ matrix. The operator $\widetilde M$ is then Choi-hermitian for any $\alpha>0$.
	
	We note that the characteristic polynomial $p$ of $\widetilde M $ has multiple roots if and only if $p$ and its derivative $p'$, have at least one common root; according to~\cref{resultant_roots}, this means that $\mathrm{Res}(p,p')=0$. Note that $\mathrm{Res}(p,p')$ is a polynomial in $\alpha$; because of the fundamental theorem of algebra, we then know that it can only have a finite number of roots or that it is equal to 0 for all $\alpha$. We can immediately rule out the second case, since we know that for $\alpha=1$ we have $\widetilde M = E$, which has no degenerate eigenvalues. This implies that there exists some value $\alpha_{\min}$ so that for every $0<\alpha<\alpha_{\min}$ $\widetilde M$ is $\NDsquare$ and $\lVert\wtM - M\rVert<\alpha \lVert E - M \rVert$.
	Choosing $\alpha < \min \{ \epsilon/\lVert E - M \rVert,\alpha_{\min} \}$ concludes the proof.
\end{proof}

For our specific task, \cref{thm:Choi-hermitian_dense} implies that we can choose a perturbation matrix~$E$ and a perturbation parameter~$\alpha$ so that we can obtain a $\NDsquare$ matrix from any input channel, such that any (possibly defective or derogatory) quantum embeddable channel will be recognized as compatible with Markovian dynamics by our algorithm even after being perturbed, given a tolerance parameter $\varepsilon \geq \alpha \left( \Fnorm{M}+\Fnorm{E} \right)$. Analogously, our algorithm running for the perturbation of a map with a non-Markovianity parameter~$\mu$ will retrieve a channel with a parameter smaller or equal to $\mu$.
Therefore, without loss of generality, we can substitute any matrix with a sufficiently close $\NDsquare$ matrix and take its logarithm instead.

A simple way to construct a Choi-hermitian $\NDsquare$ perturbation operator $E$ is to first consider a matrix $D$ with $\dim D = \dim M$, diagonal in the elementary basis, with the constraint
\begin{equation}
	d_{(j,k)} +\bar{d}_{(k,j)} \neq  d_{(\ell,m)} +\bar{d}_{(m,\ell)} \qquad \text{for all} \quad j,k,\ell,m\ \text{with} \ (j,k)\neq(\ell,m) ,
\end{equation}
where $d_{(j,k)}$ is the diagonal element in the $((j-1)d+k)$-th row and column. Then,
$E \coloneqq D+\mathds{F} D^\ast \mathds{F}$ is hermiticity-preserving and $\NDsquare$.

\subsection{Reconstructing perturbed degenerate eigenspaces in the multi-qubit case} \label{app:multi-qubit}
Assume that a multi-qubit channel has a cluster of $n$ eigenvalues that are close to each others, presumably stemming from an $n$-dimensional eigenspace with respect to an eigenvalue $\lambda$. We denote this subspace of $M$ by $\mc S_\lambda$ and we distinguish three possible cases.

If $\lambda$ is complex, then as previously discussed the hermiticity-preserving condition implies the existence of a second $n$-dimensional subspace with respect to the eigenvalue $\lambda^\ast$, $\mc S_{\lambda^\ast}$. This remark already tells us that $n \leq d^2/2$; indeed, for the single-qubit case we have ruled out the possibility of a 3-dimensional eigenspace for a complex eigenvalue. If we do not identify a second cluster of $n$ eigenvalues that are close to $\lambda^\ast$, then the channel $M$ is not quantum embeddable. Otherwise, we look for a basis $\set{v_1,\dots,v_n}$ of $\mc S_\lambda$ such that $\{v_1^\dagger,\dots,v_n^\dagger\}$ is a basis for $\mc S_{\lambda^\ast}$. Given the eigenvectors $\set{w_1,\dots,w_n}$ of the cluster of eigenvalues of $M$ related to the perturbed eigenvalues originated from $\lambda$, and $\set{z_1,\dots,z_n}$ related to $\lambda^\ast$, we hence seek $n$ vectors of the form $\ket{v} \coloneqq \alpha_1 \ket{w_1}+\dots+\alpha_n \ket{w_n}$ such that $v^\dagger \in \mathrm{span}\{z_1,\dots,z_n\}$. The set of $n$ independent vectors satisfying the condition will be chosen as a new eigenbasis of $\mc S_\lambda$ and their hermitian counterparts as the eigenbasis of $\mc S_{\lambda^\ast}$. The algorithm to retrieve the closest Lindbladian should then run ideally over all feasible solutions of this eigenbasis problem, although in practice this won't be possible since this constitute an infinite set.

\begin{table}[htbp]
	\small
	\begin{center}\renewcommand{\arraystretch}{2.5}
		\begin{tabular}{|C{2.65cm}|C{2.65cm} C{2.65cm} C{2.65cm}|}
			\hline
			\textbf{Type of eigenvalue} & \textbf{constraint on other eigenspaces} & \textbf{constraint on dimension $n$} & \textbf{allowed basis vectors} \\
			\hline
			complex & existence of partner subspace with same dimension & $n\leq d^2/2$ & each basis vector has an h.r. vector in partner eigenspace\\
			\hline
			negative & none & $n$ even & $n/2$ pairs of h.r. basis vectors \\
			\hline
			positive & none & none & $p$ pairs of h.r. basis vectors and  $n-2p$ s.a. basis vectors \\
			\hline
		\end{tabular}
	\end{center}
	\caption{Structure of an $n$-dim eigenspace of a multi-qubit hermiticity-preserving operator. Here we abbreviate ``hermitian related'' by h.r. and ``self-adjoint'' by s.a.}\label{tab:multi-qubit_deg}
\end{table}

If $\lambda$ is negative, we recall that in this case the relevant constraint is that $\log M$ is Choi-hermitian. The related perturbed subspace of $M$ then needs $n/2$ vectors $v_1,\dots v_n$ such that $\{v_1,v_1^\dagger,\dots,v_{n/2},v_{n/2}^\dagger\}$ is a basis of $\mc S_\lambda$. This implies that $n$ must necessarily be even, and indeed in the discussion of the one-qubit case we ruled out the possibility of having 3 eigenvalues stemming from a degenerate negative eigenvalue. Thus, again denoting by $\setn{w_1,\dots,w_n}$ the eigenvectors of the perturbed eigenvalues of $\lambda$, we seek~$n/2$ vectors $\ket{v} \coloneqq \alpha_1 \ket{w_1}+\dots+\alpha_n \ket{w_n}$ such that $v^\dagger \in \mc S_\lambda$. The set of vectors given by the $n/2$ pairs $\setn{v,v^\dagger}$ will be chosen as the new eigenbasis. Again, the algorithm for the Lindbladian should run for all feasible sets of hermitian related eigenbasis for  the subspace $\mc S_\lambda$.

The remaining option is a real, positive unperturbed eigenvalue $\lambda$ generating a cluster of $n$ perturbed eigenvalues of $M$. Here we have an additional freedom due to the possibility of admitting self-adjoint eigenvectors. More precisely, for each $p= 0,1, \dots, n/2$, we seek an eigenbasis of $p$ pairs of hermitian-related vectors and $n-2p$ self-adjoint vectors. Again, in principle all sets of vectors with this structure should be used for the algorithms, for all values of $p$.

This analysis constitutes the theoretical ground for the pre-processing algorithms given in \cref{sec:pre-processing}. We schematically illustrate the argument in \cref{tab:multi-qubit_deg}.

\subsection{Stability of hermiticity preserving subspaces}\label{sec:perturbation_theory}

Of course it may not be possible to perfectly reconstruct a hermiticity-preserving basis structure, via the procedure outlined in \cref{app:multi-qubit}.
In these cases we show that it is still possible to construct an \emph{approximately} hermiticity-preserving basis.
To do that we need to introduce some notation and concepts from matrix perturbation theory. 

%===========================================================
\subsubsection{Matrix perturbation theory}\label{sec:preliminaries_perturbation}
%===========================================================

In order to account for a parameter $\varepsilon$ reflecting the error tolerance with respect to the inaccuracy of the tomographic measurement, we make use of the following techniques and results in perturbation theory.

We will call \emph{eigenspace of $A$} a subspace of a matrix $A$ related to a single eigenvalue, and \emph{simple invariant subspace} an invariant subspace whose corresponding eigenvalues are all distinct from the ones of its complement subspace.

\paragraph{Perturbation of eigenvalues}~\newline
Given an $n$-dimensional eigenspace of an operator $T$ with respect to the eigenvalue $\lambda$, an analytic perturbation of the form
$T(\varepsilon) = T + \varepsilon \,  T^{(1)}  + \varepsilon^2 \, T^{(2)}+\dots$
will create a \emph{$\lambda$-cluster} of eigenvalues $\tilde \lambda_1,\dots,\tilde \lambda_r$ (also referred to as \emph{$\lambda$-group} in~\cite{Kato}) with multiplicity $n_1,\dots, n_r$ so that $\sum_{j=1}^r n_j = n$. The branches of perturbed eigenvalues are analytic functions with respect to the perturbation parameter $\varepsilon$, except for algebraic singularities. In our work, we deal with $\NDsquare$ matrices only, such that every eigenvalue has multiplicity~1; this also implies that the perturbation of a multi-dimensional eigenspace will produce simple subspaces, and more generally that any subspace of the perturbed operator $T(\varepsilon)$ is simple. 

%---------------------------------------------------------------------
%---------------------------------------------------------------------
\paragraph{Perturbation of invariant subspaces}~\newline
To establish whether a given operator is a noisy implementation of a quantum Markovian map, we will make use of and slightly adapt the matrix perturbation framework developed in~\cite{StewartSun}.

Let $\mathcal S$ be an $n$-dimensional invariant subspace of an $d^2 \times d^2$-dimensional operator $M$ and let the columns of $U_1$ be a set of orthonormal vectors spanning it. Then, we extend it into a unitary matrix $(U_1, U_2)$, where the columns of $U_2$ are a basis of the orthogonal complement $\mc S^\perp$ of $\mathcal S$. Under this basis transformation, the matrix representation of $M$ with respect to $(U_1, U_2)$ is given by
\begin{equation}
	(U_1, U_2)^H \, M \, (U_1, U_2)
	=
	\begin{pmatrix} M_1 & G \\ 0 & M_2 \end{pmatrix}
\end{equation}
since $U_2^H \, M U_1 = 0_{d^2-n,n}$.

It is usually more convenient to represent $M$ in a block-diagonal form: for $U_1, U_2, M_1, M_2$ as above, there always exist matrices $Z_1$ and $Z_2$ such that (cfr. Section V, Thm 1.5 of~\cite{StewartSun})
\begin{equation}\label{eq:spectral_resolution}
	(U_1,Z_2)^{-1} = (Z_1,U_2)^H
	\qquad \text{and} \qquad
	(Z_1,U_2)^H \, M \, (U_1,Z_2)
	=
	\begin{pmatrix} M_1 & 0 \\ 0 & M_2 \end{pmatrix}.
\end{equation}
This is referred to as the \emph{spectral resolution of $M$}.
Then, the invariant subspace $\mc S$ spanned by columns of $U_1$ is called simple if the eigenvalues of $M_1$ are all distinct from the ones of $M_2$, i.e.,
\begin{equation}
	\mathrm{spec} (M_1) \cap \mathrm{spec}(M_2) = \emptyset .
\end{equation}
It follows immediately that the orthogonal complement of $\mathcal S$ under the spectral resolution is also a simple invariant subspace.

The analysis on perturbed subspaces is based on the \emph{separation function} for a simple invariant subspace $\mc S$ of $M$,
\begin{equation}
	\mathrm{sep}(M_1,M_2)
	\coloneqq
	\min_{\norm{P}=1} P M_1 - M_2P ,
\end{equation}
which is an expression to quantify the distance between $M_1$ and $M_2$.

Now back to our task, we consider some hermiticity-preserving matrix $\Lambda$ and a perturbation operator $E$ such that $M = \Lambda + E$ is the quantum channel which we want to either identify as or discriminate from a Markovian process.

Let $\Lambda$ have an $n$-dimensional eigenspace with respect to the eigenvalue $\lambda$, and let $U_1$ be the matrix whose orthonormal columns span the $n$-dimensional invariant subspace of $M$ with respect to the corresponding $\lambda$-cluster, producing the spectral resolution of $M$ as per \cref{eq:spectral_resolution}. Under this basis, we write the representation of the perturbation matrix $E$ in the block-form
\begin{equation}\label{eq:E_block_form}
	(Z_1,U_2)^H \, E \, (U_1,Z_2)
	=
	\begin{pmatrix} E_{11} & E_{12} \\ E_{21} & E_{22} \end{pmatrix}.
\end{equation}
Then we have as a special case of Theorem 2.8 in Chapter V of Stewart and Sun~\cite{StewartSun}, where here we consider the backward perturbation $F=-E$ such that $\Lambda= M + F$ and then use the property for matrix norms $\norm{-A}=\norm{A}$,
\begin{theorem}\label{thm:subspace_perturbation}
	If
	\begin{equation}
		\gamma \coloneqq \mathrm{sep}(M_1,M_2) - \norm{E_{11}} - \norm{E_{22}} >0 \quad \text{and} \quad
		4\cdot \norm{E_{12}}\norm{E_{21}}/\gamma^2 <1
	\end{equation}
	for some submultiplicative matrix norm $\lVert \, \cdot \, \rVert$,
	then there exists a unique matrix $P$ with $\norm{P}<2\cdot\norm{E_{21}}/\gamma$ such that the columns of
	\begin{equation}
		V_1 = U_1 + Z_2P
	\end{equation}
	form a basis of the $\lambda$-eigenspace of $\Lambda$.
\end{theorem}

\subsubsection{Approximate hermiticity-preserving eigenspace structure}\label{par:approximate_HP_structure}
When perturbations ``mix'' eigenspaces in a way such that it is not anymore possible  to obtain an exact choice of vectors according to the prescription given in  \cref{app:multi-qubit}, we can search for vectors that are close to satisfying those conditions. The motivation comes from the following analysis based on the relation between the set of vectors spanning an unperturbed eigenspace and its perturbed version,
where we will make full use of the tools presented in \cref{sec:preliminaries_perturbation}. We show that the hermiticity-preserving structure is a stable property with respect to perturbations.
Formally,

\begin{theorem}[Stability of the hermiticity-preserving structure]\label{thm:stability_h-r-_structure}
	Let $\Lambda$ be an hermiticity - preserving map and $M=\Lambda+E$ its perturbed version.
	If~$\lambda<0$, $n$-degenerate eigenvalue of $\Lambda$, then there exists a set of basis vectors $\set{\tilde w_i}_{i=1}^n$ spanning the right invariant subspace of $M$ with respect to the $\lambda$-cluster such that
	\begin{equation}\label{eq:almost_h-r_vectors}
		\norm{\tilde w_i^\dagger - \tilde{w}_{i+n/2}} = \mc O \big(\norm{Z_2}_1\norm{E_{21}}_1 \big)
		\qquad \text{for } i=1,\dots ,n/2,
	\end{equation}
	where $Z_2$ is the submatrix of the basis transformation $(U_1,Z_2)$, $(Z_1,U_2)^H$ for the spectral resolution of $M$ for the invariant subspace of the $\lambda$-cluster and its complement as per \cref{eq:spectral_resolution} and $E_{21}=U_2^H E U_1$ as per \cref{eq:E_block_form}.
	If~$\lambda \geq 0$, then there exists a set of basis vectors $\set{\tilde w_i}_{i=1}^n$ spanning the right invariant subspace of $M$ with respect to the $\lambda$-cluster such that for some $p\in \set{0,\dots , n/2}$  \cref{eq:almost_h-r_vectors} stands for $i=1, \dots,p$, and moreover
	\begin{equation}
		\norm{\tilde w_i - \tilde{w}_i^\dagger} = \mc O \big(\norm{Z_2}_1\norm{E_{21}}_1 \big)
		\qquad \text{for } i=2p+1,\dots ,n.
	\end{equation}
	If~$\lambda$ is complex, then there exist a set of basis vectors $\set{\tilde w_i}_{i=1}^n$ spanning the right invariant subspace of $M$ with respect to the $\lambda$-cluster and a set of basis vectors $\set{\tilde w_i}_{i=n+1}^{2n}$ spanning the right invariant subspace of $M$ with respect to the $\lambda^\ast$-cluster such that
	\begin{equation}
		\norm{\tilde w_i^\dagger - \tilde{w}_{i+n}} = \mc O \big(\norm{X_2}_1\norm{E_{21}}_1+\norm{Y_2}_1\norm{E_{21}'}_1 \big)
		\qquad \text{for } i=1,\dots ,n,
	\end{equation}
	where $X_2$ is the submatrix of the basis transformation $(T_1,X_2)$, $(X_1,T_2)^H$ for the spectral resolution of $M$ for the subspace of the $\lambda$-cluster and its complement subspace,
	$E_{21}=T_2^H E T_1$,
	$Y_2$ is the submatrix of the basis transformation $(S_1,Y_2)$, $(Y_1,S_2)^H$ for the spectral resolution of $M$ for the subspace of the $\lambda^\ast$-cluster and its complement subspace, $E_{21}'=S_2^H E S_1$.
\end{theorem}

In principle, every basis set satisfying the bounds of \cref{thm:stability_h-r-_structure} should be used as a new eigenbasis of $R$ for the Lindbladian algorithms.
However, due to the fact that perturbation terms $ \norm{E_{12}}, \norm{E_{11}}$ and $\norm{E_{22}}$ cannot be evaluated precisely, in practice the expression in \cref{eq:almost_h-r_vectors} can only be estimated according to a guess on the perturbation matrix $E$ reflected by the tolerance parameter~$\varepsilon$. Nevertheless, when an exact hermiticity-preserving structure cannot be retrieved, this theorem motivates the approach of searching for a basis set made of pairs of eigenvectors that are close to being hermitian-related (or self-adjoint, respectively), instead of performing a brute-force search over all basis transformations in the subspace.

\vspace{12pt}

\begin{proof}[Proof of \cref{thm:stability_h-r-_structure}]
	In this proof we will make use of the vector 1-norm and matrix 1-norm only, and denote them simply by $\lVert \, \cdot \, \rVert$ to ease the notation. Assume that $\Lambda$ is an hermiticity-preserving operator such that $\Lambda = M - E$. Let $\{w_j\}_{j=1}^n$ be the set of eigenvectors of $M$ related to the cluster of $n$ eigenvalues $\{\tilde\lambda_j\}_j$ stemming from an $n$-degenerate eigenvalue $\lambda<0$ of $\Lambda$, and arranged in columns forming a $d^2 \times n$ matrix $W_1$. Let $\{ w_j\}_{j=n+1}^{d^2}$ be the remaining eigenvectors with respect to the remaining eigenvalues $\{\mu_j\}_j$ arranged in the columns forming the $d^2\times (d^2-n)$ matrix $W_2$.
	We can write
	\begin{equation}\label{eq:structure_approx_relation}
		(W_1,W_2)^{-1} M \, (W_1,W_2) =
		\begin{pmatrix} \mathrm{diag}(\tilde \lambda_1,\dots,\tilde\lambda_n) & 0\\ 0 & \mathrm{diag} (\mu_1,\dots\mu_{d^2-n})
		\end{pmatrix}.
	\end{equation}
	Consider a new basis of orthonormal vectors for the subspace of the $\lambda$-cluster $\{u_j\}_j$ given by $\ket{u_j}= \sum_{k=1}^n \varsigma_{jk} \, \ket{w_k}$, $j=1\dots,n$, arranged in columns to form a matrix $U_1$, together with a $d^2 \times (d^2-n)$ matrix $U_2$ such that $(U_1,U_2)$ is unitary. Construct matrices $Z_1$ and $Z_2$ such that $(U_1,Z_2)^{-1} = (Z_1,U_2)^H$. Under this basis transformation, we write
	
	\begin{align}\label{eq:block_structure}
		(Z_1,U_2)^H \, M \, (U_1,Z_2)
		&=
		\begin{pmatrix}M_1 & 0 \\ 0 & M_2 \end{pmatrix}\\
		\intertext{and}
		(Z_1,U_2)^H \, E \, (U_1,Z_2)
		&=
		\begin{pmatrix} E_{11} & E_{12} \\ E_{21} & E_{22} \end{pmatrix}.
	\end{align}
	
	With this choice of basis and matrix representation, we can apply \cref{thm:subspace_perturbation} in \cref{sec:preliminaries_perturbation}: the eigenspace of $\lambda$ is spanned by the columns $\{v_j\}_j$ of the $d^2\times n$ matrix $V_1 = U_1 + Z_2\,P$ for a unique operator $P$ with
	\begin{equation}
		\norm{P} < 2 \cdot \norm{E_{21}}/\left(\mathrm{sep}(M_1,M_2)-\norm{E_{11}}-\norm{E_{22}}\right).
	\end{equation}
	
	Now, since $\Lambda$ is hermiticity-preserving, then there are $n/2$ vectors of the form $\ket{\tilde v} = \sum_{j=1}^n \alpha_j \ket{v_j}$ such that $\ket{\tilde v^\dagger}$ is in the same subspace, i.e., $\ket{\tilde v^\dagger} = \sum_{j=1}^n \beta_j \ket{v_j}$ for some $\set{\beta_j}_j$. If the $\alpha$'s and $\beta$'s are a solution for this relation, then for $a_k \coloneqq \sum_j \alpha_j \varsigma_{jk}$ and $b_k \coloneqq \sum_j \beta_j \varsigma_{jk}$ and defining 
	$\tilde w_a \coloneqq \sum_{k=1}^n a_k \ket{w_k}$ and 
	$\tilde w_b \coloneqq \sum_{k=1}^n b_k \ket{w_k}$ we obtain
	\begin{align}
		&\Big\Vert \tilde w_a^\dagger - \tilde w_b \Big\Vert
		=
		\bigg\Vert \sum_{k=1}^n( a_k \ket{w_k})^\dagger - \sum_{k=1}^n b_k \ket{w_k} \bigg\Vert
		=
		\bigg\Vert \sum_{j=1}^n ( \alpha_j \ket{u_j})^\dagger - \sum_{j=1}^n \beta_j \ket{u_j} \bigg\Vert \\
		&\leq
		\bigg\Vert \sum_{j=1}^n ( \alpha_j \ket{v_j})^\dagger - \sum_{j=1}^n \beta_j \ket{v_j} \bigg\Vert
		+
		\sum_{j=1}^n \abs{\alpha_j} \, \norm{\ket{u_j^\dagger} - \ket{v_j^\dagger} }
		+
		\sum_{j=1}^n \abs{\beta_j} \, \Big\lVert \ket{u_j} - \ket{v_j} \Big\rVert.\label{eq:bound_hp}
	\end{align}
	The first term in~\cref{eq:bound_hp} is 0 by construction. Moreover, the  1-norm is invariant under hermitian conjugation, i.e.
	\begin{equation}
		\Vert \ket{u_j^\dagger} - \ket{v_j^\dagger} \Vert =  \Vert(\ket{u_j} - \ket{v_j})^\dagger \Vert =  \Vert \Flip(\ket{u_j} - \ket{v_j})^\ast \Vert =  \Vert \ket{u_j} - \ket{v_j} \Vert .
	\end{equation}
	Thus we have
	\begin{align}
		\Big\Vert \tilde w_a^\dagger - \tilde w_b \Big\Vert
		&\leq
		2\, \max_j\set{\abs{\alpha_j},\abs{\beta_j}} \, \norm{V_1 - U_1} \\
		&\leq
		2\, \max_j\set{\abs{\alpha_j},\abs{\beta_j}} \, \norm{Z_2} \norm{P}\\
		&\leq
		\frac{4\, \max_j\set{\abs{\alpha_j},\abs{\beta_j}} \, \norm{Z_2} \norm{E_{21}}}{\mathrm{sep}\left(M_1,M_2\right) - \norm{E_{11}}- \norm{E_{22}}}. \label{eq:bound_approx_hp}
	\end{align}
	
	As discussed above, if $\lambda>0$ then self-adjoint vectors should be taken into account. Assume the same structure as in \cref{eq:structure_approx_relation,eq:block_structure} again with $\ket{u_j}= \sum_{k=1}^n \varsigma_{jk} \ket{w_k}$, $j=1\dots,n$ orthonormal vectors and with the columns of $V_1 = U_1 + Z_2\, P$ spanning the eigenspace of $\lambda$. Assume that there exists a set of coefficients $\set{\alpha_j}_j$ such that $\ket{\tilde v} =\sum_j \alpha_j \ket{v_j}$ and $\tilde v^\dagger=\tilde v$. Then a bound follows in the same fashion as~\cref{eq:bound_approx_hp},
	\begin{align}\label{eq:bound_approx_hp_selfadjoint}
		\Big\Vert \tilde w_a^\dagger - \tilde w_a \Big\Vert
		&\leq
		2\, \max_j\set{\abs{\alpha_j}} \, \norm{Z_2} \norm{P}\\
		&\leq
		\frac{4\, \max_j\set{\abs{\alpha_j}} \, \norm{Z_2} \norm{E_{21}}}{\mathrm{sep}\left(M_1,M_2\right) - \norm{E_{11}}- \norm{E_{22}}},
	\end{align}
	where again $a_k = \sum_j \alpha_j \varsigma_{jk}$.

	An analogous bound can be derived for a complex eigenvalue $\lambda$, where the condition should also account for a partner eigenspace with respect to $\lambda^\ast$. In this case, we have two different basis choices to consider. Let the columns of $T_1$ span the subspace of the $\lambda$-cluster and let $(T_1,T_2)$ be unitary and such that
	\begin{equation}
		(X_1,T_2)^H \, M \, (T_1,X_2)
		=
		\begin{pmatrix}M_1 & 0 \\ 0 & A_2 \end{pmatrix}.
	\end{equation}
	Similarly, let the columns of $S_1$ be a basis of the subspace for the $\lambda^\ast$-cluster so that
	\begin{equation}
		(Y_1,S_2)^H \, M \, (S_1,Y_2)
		=
		\begin{pmatrix}N_1 & 0 \\ 0 & B_2 \end{pmatrix}.
	\end{equation}
	with $(S_1,S_2)$ unitary. Then the columns of $L_1 = T_1 + X_2\,P$ span the eingenspace of $\lambda$ and those of $R_1= S_1+ Y_2 \, Q$ span that of $\lambda^\ast$.
	Additionally, assume that the eigenspace of $\lambda$ admits a basis of vectors whose hermitian counterparts span the eigenspace of $\lambda^\ast$, that is, there exist coefficients $\alpha$'s and $\beta$'s such that $\sum_j (\alpha_j \ket{\ell_j})^\dagger = \sum_j \beta_j \ket{r_j}$.
	
	Now, let $\set{w_k}_{k=1}^n$ denote the eigenvectors of the $\lambda$-cluster, and let\linebreak  $\set{w_k}_{k=n+1}^{2n}$ denote those of the $\lambda^\ast$-cluster and write $\ket{t_j} =\sum_j \varsigma_{jk} \ket{w_k}$ and $\ket{s_j} = \sum_j \varsigma_{j,k+n} \ket{w_{k+n}}$.
	Then for 
	$a_k \coloneqq \sum_j \alpha_j \varsigma_{jk}$,
	$b_{k,n} \coloneqq \sum_j \beta_j \varsigma_{j,k+n}$
	and
	$\tilde w_a \coloneqq \sum_{k=1}^n a_k \ket{w_k}$,
	$\tilde w_{a,n} \coloneqq \sum_{k=1}^{n} b_{k,n} \ket{w_{k+n}}$
	\begin{equation}\label{eq:approximate_HP_complex_pair}
		\Big\Vert \tilde w_a^\dagger - \tilde w_{a,n} \Big\Vert
		\leq
		\max_j\set{\abs{\alpha_j},\abs{\beta_j}} \ \left( \norm{X_2}\,\norm{P} + \norm{Y_2}\,\norm{Q} \right).
	\end{equation}
\end{proof}

%======================================================================
\section{Estimation guarantees}\label{app:maths}
%======================================================================

%---------------------------------------------------------------------------------------
\subsection{Time dependent algorithm: error bounds and Lipschitz parameter} \label{app:timeDep} 
%---------------------------------------------------------------------------------------

%---------------------------------------------------------
\subsubsection{Estimation of time-dependent Markovian channels}\label{sec:bounds}
%---------------------------------------------------------

We want to investigate the error in the approximation of the time series snapshots of a time-dependent Markovian quantum channel with best-fit Lindbladians from our scheme, and in particular to link the degree of accuracy with the number of snapshots taken. In other words, we want to establish how fine-grained the tomographic measurements should be in order to achieve the desired level of accuracy.

To simplify calculations and presentation, we assume that all snapshots are taken at uniform time intervals $\frac{\mc T}{N}$ over a total run time $\mc T$, but we note that the results can be extended to the general case straightforwardly.
As a second assumption for $\Phi(t)$, we consider an evolution where the Lindbladian is time-dependent but with moderately varying fluctuation; as discussed, we impose a Lipschitz continuity
$\norm{L(t_2)-L(t_1)} \leq \eta (t_2 - t_1)$ for some Lipschitz constant~$\eta$.
We then have

\begin{theorem}[Snapshots approximation]
	Let $\set{M_p}_{p=1}^N$ be a time series of $N$ tomographic snapshots over a total run time $\mc T$ of a Markovian quantum channel acting on a $d$-dimensional space. Let $\eta$ be the Lipschitz constant of the generator for some norm $\norm{\ \cdot \ }$.
	Construct $\Theta_p=M_p M_{p-1}^{-1}$ and define $\wt M_p=\prod_{j=p}^1 \e^{L_j}$ with best-fit Lindbladians $\set{L_p}_p$ from \cref{alg:time_dependent}.
	Then for all $p=1,\dots, N$,
	\begin{equation}\label{eq:thm_expression_1}
		\norm{\Theta_p - \e^{L_p}} \leq \eta^2 \frac{\mc{T}^4}{N^3} + \mc O \Big(\frac{\mc T^5}{N^4}\Big)
	\end{equation}
	and
	\begin{equation}
		\norm{M_p - \wt M_p}
		\leq
		\sqrt d \Big[ \exp\Big(\sqrt d \, \eta^2 \frac{\mc{T}^4}{N^2}\Big) - 1 \Big] + \mc O \Big(\frac{\mc T^5}{N^3}\Big) .
	\end{equation}
\end{theorem}
The number of snapshots should hence scale quadratically with respect to the run time, and the difference between snapshots and reconstructed best-fit channel closes in the asymptotic limit $N\rightarrow \infty$, in line with the analysis presented in~\cite{WolfDivisibility}.

%-----------------------------------------------------
\begin{proof}
	We first investigate a bound on $\Theta_p$.
	Consider the time-averaged Lindbladian
	$L_\avg(p) = \frac{N}{\mc T} \int_{(p-1)\frac{\mc T}{N}}^{p\frac{\mc T}{N}} L(t) \dt$ and expand the time-ordered integral as a Dyson series,
	\begin{align}
		\Theta_p= &\mathbb T \exp \int_{(p-1)\frac{\mc T}{N}}^{p\frac{\mc T}{N}} L(t) \dt \\
		&=
		\1 + \int_{(p-1)\frac{\mc T}{N}}^{p\frac{\mc T}{N}} L(t) \dt
		+
		\int_{(p-1)\frac{\mc T}{N}}^{p\frac{\mc T}{N}} dt_1 \int_{(p-1)\frac{\mc T}{N}}^{t_1} dt_2 \, L(t_1) L(t_2) \\
		&+
		\int_{(p-1)\frac{\mc T}{N}}^{p\frac{\mc T}{N}} dt_1 \int_{(p-1)\frac{\mc T}{N}}^{t_1}  dt_2 \int_{(p-1)\frac{\mc T}{N}}^{t_2} \dt_3 \, L(t_1) L(t_2) L(t_3) + \dots \\
		&=
		\1 + L_\avg(p) \frac{\mc T}{N}
		+
		\sum_{k=2}^\infty
		\int_{(p-1)\frac{\mc T}{N}}^{p\frac{\mc T}{N}} dt_1 \int_{(p-1)\frac{\mc T}{N}}^{t_1} dt_2 \dots \int_{(p-1)\frac{\mc T}{N}}^{t_{k-1}} \dt_k \, L(t_1) L(t_2) \cdots L(t_k). \label{eq:Dyson_unfolded}
	\end{align}
	Let us consider the second Dyson term, and note that thanks to the Lipschitz continuity in the observed time interval  we can write
	$L(t) = L_\avg(p) + \eta \frac{\mc T}{N} Y(t)$ for some $Y(t)$ with $\norm{Y(t)} \leq 1$. Taking the norm yields
	\begin{align}
		&\norm{\int_{(p-1)\frac{\mc T}{N}}^{p\frac{\mc T}{N}} dt_1 \int_{(p-1)\frac{\mc T}{N}}^{t_1} dt_2 \, L(t_1) L(t_2)} \\
		&\leq
		\int_{(p-1)\frac{\mc T}{N}}^{p\frac{\mc T}{N}} dt_1 \int_{(p-1)\frac{\mc T}{N}}^{t_1} dt_2 \, \norm{(L_\avg(p) + \eta \frac{\mc T}{N} Y(t_1))(L_\avg(p)+ \eta \frac{\mc T}{N} Y(t_2)) } \\
		&\leq
		\frac 1 2 \norm{L_\avg(p)}^2 \frac{\mc{T}^2}{N^2} + 2 \eta \norm{L_\avg(p)} \norm{\max Y(t)} \frac 1 2 \frac{\mc{T}^3}{N^3} + \eta^2 \norm{\max Y(t)}^2 \frac 1 2 \frac {\mc{T}^4}{N^4}\\
		&\leq
		\frac 1 2 \norm{L_\avg(p)}^2 \frac{\mc{T}^2}{N^2} + 2 \eta \norm{L_\avg(p)} \frac 1 2 \frac{\mc{T}^3}{N^3} + \eta^2 \frac 1 2 \frac {\mc{T}^4}{N^4}.
	\end{align}
	In the same vein, the norm of $k$-th term (with $k\geq 2$) in the Dyson series can be bounded as
	\begin{align}
		&\norm{\int_{(p-1)\frac{\mc T}{N}}^{p\frac{\mc T}{N}} dt_1 \int_{(p-1)\frac{\mc T}{N}}^{t_1} dt_2 \dots \int_{(p-1)\frac{\mc T}{N}}^{t_{k-1}} \dt_k \, L(t_1) L(t_2) \cdots L(t_k)} \\
		&\leq
		\frac {1}{k!} \norm{L_\avg(p)}^k \frac{\mc{T}^k}{N^k}
		+ k \eta \norm{L_\avg(p)} \frac {1}{k!}\frac{\mc{T}^{k+1}}{N^{k+1}}
		+ \dots\\
		&=
		\sum_{\ell=0}^k \binom k \ell \norm{L_\avg(p)}^{k-\ell}\eta^\ell \frac {1}{k!} \frac{\mc{T}^{k+\ell}}{N^{k+\ell}} \\
		&=
		\sum_{\ell=0}^k \frac{1}{\ell!} \eta^\ell \left(\frac{\mc{T}}{N}\right)^{2\ell}
		\frac{1}{(k-\ell)!} \norm{L_\avg(p)}^{k-\ell} \frac{\mc{T}^{k-\ell}}{N^{k-\ell}}.
	\end{align}
	Now if we fix $\ell$ and sum over $k$ (that is, we sum the $\ell$-th item in the sum for the bound of each Dyson term) we get
	\begin{align}
		&\frac{1}{\ell!} \eta^\ell \left(\frac{\mc T}{N}\right)^{2\ell}
		\sum_{k\geq \ell}\frac{1}{(k-\ell)!} \norm{L_\avg(p)}^{k-\ell} \frac{\mc{T}^{k-\ell}}{N^{k-\ell}} \\
		&=
		\frac{1}{\ell!} \eta^\ell \left(\frac{\mc T}{N}\right)^{2\ell}
		\sum_{m\geq 0}\frac{1}{m!} \norm{L_\avg(p)}^{m} \frac{\mc{T}^{m}}{N^{m}} \\
		&=
		\frac{1}{\ell!} \eta^\ell \left(\frac{\mc T}{N}\right)^{2\ell} \exp \left(\norm{L_\avg(p)} \frac{\mc T}{N} \right), \label{eq:bound_for_ell}
	\end{align}
	and summing now \cref{eq:bound_for_ell} over $0\leq \ell \leq k $ gives
	\begin{equation}\label{eq:two_exponentials}
		\exp \left(\eta \frac{\mc{T}^2}{N^2}\right) \exp \left( \norm{L_\avg(p)} \frac{\mc T}{N}\right) .
	\end{equation}
	Note that in \cref{eq:two_exponentials} we have accounted for an extra term not appearing in \cref{eq:Dyson_unfolded}, that is, the item for $k=\ell=1$ that should thus be subtracted, yielding
	\begin{equation}
		\norm{\mathbb T \exp \int_{(p-1)\frac{\mc T}{N}}^{p\frac{\mc T}{N}} L(t) \dt}
		\leq
		\exp \left(\eta \frac{\mc{T}^2}{N^2}\right) \exp \left( \norm{L_\avg(p)} \frac{\mc T}{N}\right) - \eta \frac{\mc{T}^2}{N^2}.
	\end{equation}
	
	For the actual bound on the approximation of $\Theta_p$ we thus have
	\begin{align}
		\norm{\Theta_p - \exp L_p}
		&\leq
		\norm{\mathbb T \exp \int_{(p-1)\frac{\mc T}{N}}^{p\frac{\mc T}{N}} L(t) \dt - \exp L_p}  \\
		&\leq
		\norm{\mathbb T \exp \int_{(p-1)\frac{\mc T}{N}}^{p\frac{\mc T}{N}} L(t) \dt - \exp \Big(L_\avg(p) \frac{\mc T}{N}\Big)}\\
		&\leq
		\left(\exp \eta \frac{\mc{T}^2}{N^2}  - 1\right) \exp \Big(\norm{L_\avg(p)} \frac{\mc T}{N}\Big) - \eta \frac{\mc{T}^2}{N^2} \\
		&=
		\eta \frac{\mc{T}^2}{N^2} \left( \exp \Big(\norm{L_\avg(p)} \frac{\mc T}{N}\Big) - 1 \right)
		+ \mc O (\eta^2 \frac{\mc{T}^4}{N^4}).
	\end{align}
	Considering that by Lipschitz continuity we have $\norm{L_\avg(p)}\leq \eta \mc T$, neglecting terms $\mc O (\frac{\mc T^5}{N^4})$ we obtain the first result in the theorem in \cref{eq:thm_expression_1},
	\begin{equation}\label{eq:bound_T_p_approximation}
		\norm{\Theta_p - \exp L_p} \leq \eta^2 \frac{\mc{T}^4}{N^3} .
	\end{equation}
	
	%----------------------------------
	
	\medskip
	
	Now we turn to the error bound in the snapshot approximation, which is clearly the largest at the latest snapshot $M_N$.
	By writing $\Theta_p =   \exp L_p  + A_p$ we then have
	\begin{equation}
		\norm{M_N - \wt M_N}
		=
		\Big\Vert \prod_{p=N}^1 (\exp L_p  + A_p) - \prod_{p=N}^1 \exp L_p  \Big\Vert .
	\end{equation}
	Now, following the argument in~\cite{WolfDivisibility}, we note that $\prod_{p=N}^1 (\exp L_p  + A_p)$ contains $\binom{N}{k}$ terms with $N-k$ exponentials~$\exp L_p$ and $k$ operators from~$\set{A_p}_p$, where the former come in at most $k+1$ separated groups each of those bound by $\sqrt d$~\cite[Theorem 1]{contractivityPT}, and the latter according to \cref{eq:bound_T_p_approximation} are bound by
	$ \max \norm{A_p} \leq \eta^2 \frac{\mc{T}^4}{N^3} + \mc O (\frac{\mc T^5}{N^4})$ .
	Hence,
	\begin{align}
		\norm{M_N - \wt M_N}
		&\leq
		\sqrt d \Big[ \Big(1+\sqrt d \, \eta^2 \frac{\mc{T}^4}{N^3} + \mc O (\frac{\mc T^5}{N^4})\Big)^N - 1 \Big] \\
		&\leq
		\sqrt d \Big[ \exp \Big( \sqrt d \, \eta^2 \frac{\mc{T}^4}{N^2} \Big) - 1 \Big] + \mc O (\frac{\mc T^5}{N^3}) .
	\end{align}
	
\end{proof}

%---------------------------------------------------------
\subsubsection{Bound on difference between consecutive extracted Lindbladians}\label{sec:Lipschitz_bound}
%---------------------------------------------------------
In this section we prove the following result for the interpretation of the parameter $\beta(\eta)$ in \cref{alg:time_dependent}.
This can be expressed in relation to the Lipschitz constant~$\eta$ of the Lindblad operator, the total run time~$\mc T$ and the total number of snapshots taken, $N$. More precisely, $\beta(\eta)$ is characterised by the RHS of \cref{eq:bound_for_beta} of the following theorem.

\begin{theorem}[Lipschitz bound]
	Consider a time-dependent Markovian process $\Phi(t)$ characterised by a Lipschitz continuous Lindblad generator $L(t)$ with Lipschitz constant $\eta$ with respect to some norm $\norm{\ \cdot \ }$.
	Then the set of best-fit Lindbladians $\set{L_p}_{p=1}^N$ to
	\begin{equation}
		\log \Phi( (p-1)\mc T/N,p \mc T/N ) \quad p=2,\dots,N
	\end{equation}
	satisfies
	\begin{equation}\label{eq:bound_for_beta}
		\norm{L_p - L_{p-1}} \leq \eta \, \mc T^2 / N^2 + 2 \Big(\mathfrak R (p) + \mathfrak R (p-1) \Big)
	\end{equation}
	where $\mathfrak R(p)$ is the error for the truncation at the first order of the Magnus series over the interval $[t_{p-1},t_p]$.
\end{theorem}

\noindent In the specific case of snapshots taken at regular intervals $\mc T/N$, we thus have:
\begin{equation}
	\beta(\eta) = \eta \frac{\mc T^2}{N} + 4 \mathfrak R .
\end{equation}

\begin{proof}
	We consider the operator $K_p = \log \Phi( (p-1)\mc T/N,p \mc T/N )= L_{\avg} (p) \frac{\mc T}{N} + \mathfrak R (p)$, which is not guaranteed to be of Lindbladian form~\cite{Schnell20,Schnell21}. By triangle inequality we have
	\begin{align}
		\norm{L_p - L_{p-1}}
		&\leq
		\norm{L_p - K_p + K_p - L_{p-1} + K_{p-1} - K_{p-1} } \\
		&\leq
		\norm{L_p - K_p} + \norm{ L_{p-1} - K_{p-1} } + \norm{ K_p - K_{p-1} } \label{eq:bound_triangle}.
	\end{align}
	Since $L_p$ is to be considered the best-fit Lindbladian and observing that the first order Magnus term $L_\avg (p) \frac{\mc T}{N} = \int_{(p-1)\mc T /N}^{p \mc T /N}   L(t) \dt$ is a Lindbladian, it follows that
	\begin{equation}
		\norm{L_p - K_p} \leq \norm{L_\avg (p) \frac{\mc T}{N} - K_p} \eqqcolon \mathfrak R (p),
	\end{equation}
	and equivalently for  $\norm{ L_{p-1} - K_{p-1} }$.
	For the third term in \cref{eq:bound_triangle}, we write
	\begin{align}
		\norm{ K_p - K_{p-1} }
		&\leq
		\norm{L_\avg (p) \frac{\mc T}{N} - L_\avg (p-1) \frac{\mc T}{N} } + \mathfrak R (p) + \mathfrak R (p-1) \\
		&=
		\norm{ \int_{(p-1)\mc T /N}^{p \mc T /N} L(t)  \dt - \int_{(p-2)\mc T /N}^{(p-1) \mc T /N} L(t)  \dt} + \mathfrak R (p) + \mathfrak R (p-1) \\
		&=
		\norm{ \int_{(p-1)\mc T /N}^{p \mc T /N} \Big( L(t) - L(t-\mc T/N) \Big)  \dt } + \mathfrak R (p) + \mathfrak R (p-1) \\
		&\leq
		\eta \, T /N \int_{(p-1)\mc T /N}^{p \mc T /N} \dt + \mathfrak R (p) + \mathfrak R (p-1) \\
		&=
		\eta \, T^2 /N^2 + \mathfrak R (p) + \mathfrak R (p-1) .
	\end{align}
\end{proof}

%======================================================================
\subsection{Searching over branches of the matrix logarithm} \label{sec:Mloop}
%======================================================================

The loop in all the convex optimisation algorithms runs a brute-force search over the $\vec{m}$-branches of the matrix logarithm.
As previously explained, there are countably infinitely many such branches.
However, Khachiyan and Porkolab~\cite{Khachiyan+Porkolab} prove that setting $m_\mathrm{max} = \BigO(2^{2^{\poly(d)}})$ always suffices to find a solution if one exists.
In fact, the same authors prove that there is a polynomial-time algorithm for integer semidefinite programming with any fixed number of variables, so for any fixed Hilbert space dimension in our setting, based on the ellipsoid method (which is more efficient that performing brute-force search up to the upper-bound).
This polynomial-time algorithm is important theoretically, but not practical.
In reality, integer programming solvers use branch-and-cut methods~\cite{Mitchell2009}.

However, the increased implementation complexity of these more sophisticated techniques is unlikely to be justified here. Instead, we encode a simple, brute-force search over the branches of the logarithm, ordered by increasing $\abs{\vec{m}}$. There are physical reasons why a naive, brute-force search is likely to work well here. Large values of $m_j$ correspond to high-energy / frequency components of the noise. It is unlikely that very high energy underlying physical processes play a significant role, and it is also unlikely that our tomographic data will be sensitive enough to resolve very high-frequency components. Therefore, if there is a Lindblad generator consistent with the tomographic snapshot, it is most likely to occur at low values of $\abs{\vec{m}}$. Numerical studies on synthesised examples of 1-qubit tomography (see \cref{fig:m-variation}) confirm that setting a small value of $m_\mathrm{max}$ suffices in practice; indeed, even $m_\mathrm{max}=1$ suffices in all cases we tested. This is also corroborated by the numerics that was carried out in~\cite{Markdynamics08}.

It would be straightforward to replace the brute-force search by a call to an integer program solver if this ever proved necessary.

\begin{figure}[htbp]
	\centering
	\begin{subfigure}[t]{0.45\textwidth}
		\centering
		\includegraphics[width=\linewidth]{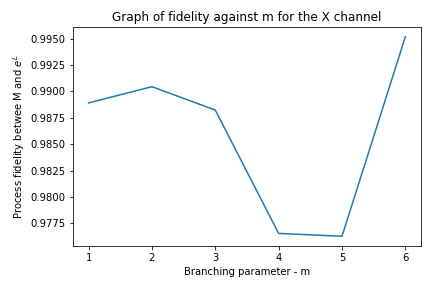}
		\caption{Process fidelity of the tomographic results and the Markovian channel output by our algorithm for an $X$-gate, running \cref{alg:MAIN}  for 500 random samples for each value of $m$.}
	\end{subfigure}
	\hfill
	\begin{subfigure}[t]{0.45\textwidth}
		\centering
		\includegraphics[width=\linewidth]{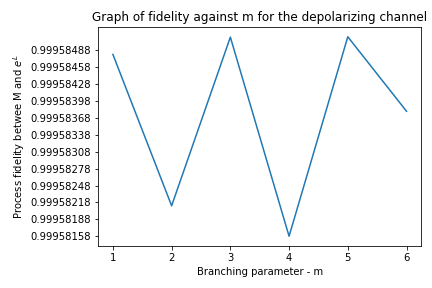}
		\caption{Process fidelity of the tomographic results and the Markovian channel output by our algorithm for a depolarizing channel with $p=0.3$. For each $m$-value, \cref{alg:MAIN} was run for 100 random samples.}
	\end{subfigure}
	\caption{Investigating how the optimal Lindbladian found by our algorithm varies with the number of branches of the logarithm searched over for the $X$-gate and depolarizing channel. There is some fluctuation in fidelity - this is due to the randomised nature of the algorithm (see \cref{sec:examples} for full results investigating how fidelity varies with random samples). }\label{fig:m-variation}
\end{figure}

%======================================================================
\subsection{Analytical derivation of the non-Markovianity parameter}\label{sec:mu_by_hand}
%======================================================================

In order to benchmark the output of \cref{alg:MR:NoMark_parameter}, we calculate by hand the value of $\mu$ for the unital quantum channel numerically investigated in \cref{sec:Unital_quantum_channel}. This particular example is convenient because there is no cluster of eigenvalues; indeed, all eigenvalues of the simulated perturbed channel are positive and far apart from each other. This means that all subspaces are 1-dimension and thus we don't have to undergo a search over (infinitely many) feasible basis vectors for the unperturbed channel.

Noting that the white noise addition cannot influence either the herm\-it\-ic\-ity-preserving or the trace-preserving condition, our strategy will be to construct a matrix from the tomographic data by ``filtering'' these two properties, and then calculate the minimal value of $\mu$ to increase the eigenvalues of its Choi representation in order to satisfy the conditionally complete positivity condition. This approach does not guarantee that we will obtain the closest map compatible with Markovian dynamics through white noise addition, but we can reasonably expect to retrieve an operator that is very close to the optimal one. We will then compare the value for $\mu$ from this calculation with the one returned by algorithm when setting $\varepsilon$ equal to the distance between the input matrix and the matrix exponential of the map calculated with this analytical method.

Consider the constraints on the eigenvalues and eigenvectors of an herm\-it\-ic\-ity- and trace-preserving channel having non-degenerate and non-negative spectrum. Hermiticity implies that all eigenvectors must be self-adjoint; the trace-preserving condition means that one eigenvalue must be equal to 0 with respect to a left eigenvector being proportional to the maximally entangled state $\bra{\omega}=\bra{(1,0,0,1)}$. By the biorthogonality conditions between left and right eigenvectors, this forces the three right eigenvectors related to the complement of the kernel to have first and fourth components of opposite sign. Hence, we will go through the following steps to turn the simulated perturbed unital channel $\mathcal{U}$ into an hermiticity- and trace-preserving matrix $H'$ and calculate its non-Markovianity measure by hand:
\begin{enumerate}[label=(\arabic*)]
	\item From the eigenvalues $\lambda_1>\lambda_2>\lambda_3>\lambda_4$, set $\lambda_4=0$ and define the matrix $D\coloneqq \mathrm{diag} (\lambda_1,\lambda_2,\lambda_3,0)$.
	\item From the right eigenvectors $\{w_j\}_j$ build self-adjoint vectors $\ket{\tilde{v}_j} = 1/2 \cdot (\ket{w_j}+ \ket{w_j^\dagger})$, for $j=1,\dots,4$.
	\item Construct three vectors orthogonal to $\bra{\omega}$ whose first and fourth components are respectively $v_j^{(1)} = 1/2 \cdot \big(\tilde v_j^{(1)} - \tilde v_j^{(4)}\big)$ and $v_j^{(4)}=-v_j^{(1)}$, and with $v_j^{(2)}=\tilde v_j^{(2)}$ and $v_j^{(3)}=\tilde v_j^{(3)}$, for $j=1,2,3$. Define then $S\coloneqq (v_1,v_2,v_3,\tilde v_4)$.
	\item Define $H' \coloneqq S D S^{-1}$ and calculate
	\begin{equation}
		\mu = 2 \cdot \lvert \lambda_{\min} (\omega_\perp  (H')^\Gamma \omega_\perp) \rvert
		\qquad \text{and} \qquad
		\varepsilon = \Fnorm{\mathcal{U} - \exp H'} .
	\end{equation}
\end{enumerate}
For the simulated channel, under this procedure we obtain a value of $\mu=5.76983$ and $\varepsilon = 0.01788$. As a comparison, by setting $\varepsilon = 0.01788$ in our algorithm (with $\delta_{\mathrm{step}}=10^{-5}$) this returns $\mu_{\min}=5.76539819$.
We refer to the accompanying Mathematica notebook in the Supplementary Material for the explicit calculation following this procedure.

\end{document}